\pdfoutput=1

\documentclass[letterpaper,11pt]{amsart}
\usepackage{amsmath}
\usepackage{amsthm}
\usepackage{amssymb}
\usepackage{amsfonts}
\usepackage{fullpage}
\usepackage[dvipsnames]{xcolor}
\usepackage[noadjust]{cite}
\usepackage[utf8]{inputenc}
\usepackage{bm}
\renewcommand{\mathbf}[1]{\bm{#1}}

\newcommand{\inp}[1]{\left(#1\right)}
\newcommand{\inb}[1]{\left[#1\right]}
\newcommand{\inbr}[1]{\left\{#1\right\}}

\newcommand{\Grad}{\ensuremath{\nabla}}

\usepackage[bookmarks=false]{hyperref}
\hypersetup{
  colorlinks=true,
  linkcolor=green!30!black,
  linktoc=all,
  citecolor=blue,
}

\newtheorem{theorem}{Theorem}[section]
\newtheorem{observation}[theorem]{Observation}
\newtheorem{corollary}[theorem]{Corollary}
\newtheorem{lemma}[theorem]{Lemma}
\newtheorem{proposition}[theorem]{Proposition}

\newtheorem{condition}{Condition}
 \newtheorem{consequence}[theorem]{Consequence}

\theoremstyle{definition}
\newtheorem{definition}[theorem]{Definition}

\newtheorem{numremark}{Remark}
\newtheorem{remark}[numremark]{Remark}

\newcommand{\norm}[2]{\ensuremath{\Vert #1 \Vert_{#2}}}
\newcommand{\abs}[1]{\left\vert#1\right\vert}
\DeclareMathOperator*{\argmax}{arg\,max}
\DeclareMathOperator*{\argmin}{arg\,min}
\newcommand{\set}[1]{\left\{#1\right\}}
\newcommand{\eps}{\varepsilon}

\newcommand{\defeq}{:=}

\newcommand{\im}{\ensuremath{\iota}}
\newcommand{\poly}[1]{\ensuremath{\mathop{\mathrm{poly}}\inp{#1}}}

\newcommand{\C}{\mathbb{C}}
\newcommand{\R}{\mathbb{R}}
\newcommand{\N}{\mathbb{N}}

\newcommand{\ratio}[2]{R_{#1}^{(#2)}}
\newcommand{\prob}[2]{\mathcal{P}_{#1} \inb{ #2 }}
\renewcommand{\Pr}[2]{\mathrm{Pr}_{#1} \inb{ #2 }}
\newcommand{\Z}[2]{Z_{#1}^{(#2)}}
\newcommand{\G}[2]{G_{#1}^{(#2)}}
\newcommand{\as}[2]{a_{#1}^{(#2)}}
\newcommand{\bs}[2]{b_{#1}^{(#2)}}
\newcommand{\Bad}[1]{B_{#1}}
\newcommand{\Bi}[1][i]{B(#1)}
\newcommand{\Bj}[1][j]{B(#1)}
\newcommand{\alphastar}{\alpha^{\star}}
\newcommand{\Yk}{Y_k}
\newcommand{\Yki}{Y_k^{(i)}}
\newcommand{\Ykj}{Y_k^{(j)}}

\newcommand{\st}[1][]{\ensuremath{\;\mathbf{#1\vert}\;}}

\renewcommand{\vec}[1]{\ensuremath{\mathbf{#1}}}

\newcommand{\good}[1]{\ensuremath{\Gamma_{#1}}}
\newcommand{\ifactor}[0]{\ensuremath{\rho_I}}
\newcommand{\rfactor}[0]{\ensuremath{\rho_R}}

\newcommand{\nice}[0]{nice}

\newcommand{\wt}{\ensuremath{{\tilde{w}}}}
\newcommand{\tw}{\ensuremath{\wt}}
\newcommand{\gammat}{\ensuremath{{\tilde{\gamma}}}}

\newcommand{\eqbreak}{}
\newcommand{\eqbreakno}{}
\newcommand{\maybealign}{}

\begin{document}
\title{Correlation decay and partition function zeros:\\ Algorithms and 
phase transitions} \thanks{A preliminary version of this paper, under the
  title ``A Deterministic Algorithm for Counting Colorings with $2\Delta$
  Colors," appeared in {\it Proceedings of the 60th Annual IEEE Symposium on
    Foundations of Computer Science}, 2019~\cite{focs-version}.  The present
    version contains several additional results.}
\author{Jingcheng Liu} \thanks{Jingcheng Liu, Nanjing University. Email: \texttt{liu@nju.edu.cn}.}\author{Alistair Sinclair} \thanks{Alistair Sinclair, University of California,
  Berkeley. Email: \texttt{sinclair@cs.berkeley.edu}.}\author{Piyush Srivastava} \thanks{Piyush Srivastava, Tata Institute of
  Fundamental Research.  Email: \texttt{piyush.srivastava@tifr.res.in}.}

\begin{abstract}
  \noindent We explore connections between the phenomenon of correlation decay
  (more precisely, strong spatial mixing) and the location of Lee-Yang
  and Fisher zeros for various spin systems.  In particular we show that, in
  many instances, proofs showing that weak spatial mixing on the Bethe lattice
  (infinite $\Delta$-regular tree)
  implies strong spatial mixing on all graphs of maximum degree~$\Delta$ can be
  lifted to the complex plane, establishing the absence of zeros of
  the associated partition function in a complex neighborhood of the
  region in parameter space corresponding to strong spatial mixing.  This allows
  us to give unified proofs of several recent results of this kind, including
  the resolution by Peters and Regts of the Sokal conjecture for the
  partition function of the hard core lattice gas.  It also allows us to prove new
  results on the location of Lee-Yang zeros of the anti-ferromagnetic Ising
  model.

  We show further that our methods extend to the case when weak spatial mixing
  on the Bethe lattice is not known to be equivalent to strong spatial mixing on
  all graphs.  In particular, we show that results on strong spatial mixing in
  the anti-ferromagnetic Potts model can be lifted to the complex plane to give
  new zero-freeness results for the associated partition function, significantly
  sharpening previous results of Sokal and others.  This new extension is also
  of independent algorithmic interest: it allows us to give the first polynomial
  time deterministic approximation algorithm (FPTAS) for counting the number of
  $q$-colorings of a graph of maximum degree~$\Delta$ provided only that
  $q\ge 2\Delta$, a question that has been studied intensively.  This matches
  the natural bound for randomized algorithms obtained by a straightforward
  application of Markov chain Monte Carlo.  In the case when the graph is also
  triangle-free, we show that our algorithm applies under the weaker condition
  $q \geq \alpha\Delta + \beta$, where $\alpha \approx 1.764$ and
  $\beta = \beta(\alpha)$ are absolute constants.

\end{abstract}

\maketitle
\setcounter{tocdepth}{1}

\vspace{-2\baselineskip}
\tableofcontents

\thispagestyle{empty}

\newpage
\setcounter{page}{1}

\newpage

\section{Introduction}
\subsection{Background and related work}
A standard approach to formalizing phase transitions in spin systems proceeds by
characterizing when long-range correlations between spins appear in the system.
More formally, one starts with an infinite graph, such as the Bethe lattice
(infinite $\Delta$-regular tree) or $\mathbb{Z}^d$, and then
characterizes the regions in the space of parameters of the model in which,
under the associated Gibbs distribution that assigns probabilities to configurations,
correlations between spins decay exponentially with the distance between them.
This correlation decay property is also known as ``spatial mixing."
This formalism can be extended to infinite families of finite graphs, and
has also been studied extensively due to its connections with the computational
complexity of Markov chain Monte Carlo methods for sampling from the associated
Gibbs distributions.

A related but different approach to studying phase transitions is via the so-called
Yang-Lee theory~\cite{leeyan52}.  Here, one views the infinite graph as a
suitable limit of a sequence of finite graphs of growing size, and studies the
convergence of the free energy density (the logarithm of the partition function
divided by the number of vertices) over this sequence.  Yang and Lee showed that,
under mild conditions, this limit exists and is analytic on a given subset~$S$ of the
parameter space of the system, provided that none of the partition functions of
the graphs appearing in the sequence have any roots in a \emph{complex} 
neighborhood of~$S$, uniformly over the graphs in the sequence.  This classical
approach has also recently found new algorithmic applications inspired by the
work of Barvinok~\cite{barvinok2017combinatorics}.

Classical work of Dobrushin and Shlosman~\cite{DS85,DS87} establishes an
equivalence between a strong version of spatial mixing (called ``strong spatial
mixing''\footnote{``(Weak) spatial mixing" simply refers to the decay of
  correlations property; ``strong spatial mixing" is said to hold when
  correlations between spins decay with distance even in the presence of fixed
  spins (boundary conditions) close to the spins being measured.  Strong spatial
  mixing is a crucial ingredient in the design of efficient algorithms,
  including Markov chain Monte Carlo and algorithms based on the self-avoiding
  walk tree.}) and the Yang-Lee formalism of phase transitions in the special
case of lattices~$\mathbb{Z}^d$.  However, their approach makes essential use of
the amenability of the lattice (in the form that the size of a neighborhood of
radius $r$ grows only as a polynomial in $r$), and does not extend to
non-amenable graph families.  Until recently, few formal connections between
them were known for the setting of general graphs.
Sokal~\cite{sokal2000personal} conjectured that, for the hard core lattice gas
model, there is a complex neighborhood $\mathcal{N}$ of the interval
$(0, \lambda_c(\Delta))$, where $\lambda_c(\Delta)$ is the critical activity of
the model on the infinite regular tree (Bethe lattice) of degree~$\Delta$, such
that the partition function of {\it any\/} finite graph of degree at
most~$\Delta$ does not vanish in~$\mathcal{N}$.  This conjecture was only
recently resolved by Peters and Regts~\cite{peters17:_sokal}.  More recently,
this correspondence between the two notions of phase transition for general
bounded-degree graphs has been
extended to the Ising model~\cite{liu2018fisher,peters_location_2018,shaosun19}.

In this paper, we further explore this correspondence with a view to establishing
it in more generality.  Our first set of results show that previous arguments
establishing an equivalence between weak spatial mixing on the infinite $\Delta$-regular
tree and
strong spatial mixing of the Gibbs measure on all graphs of maximum degree~$\Delta$
can be ``lifted'' to the complex plane, in such a way as to also prove that the
partition functions of all such graphs remain zero-free in a uniform complex
neighborhood of the real parameter interval in which strong spatial mixing 
holds.
This gives new and simpler proofs of Peters and Regts' resolution of the Sokal
conjecture~\cite{peters17:_sokal} described above, and of the results of the present authors 
on the Fisher zeros of the zero-field ferromagnetic Ising
model~\cite{liu2018fisher}.  In addition, our method allows us to prove new results on the
Lee-Yang zeros of the {\it anti-ferromagnetic\/} Ising model, which we now describe.

Formally, the partition function of the Ising model on a graph $G = (V, E)$
can be written in terms of an edge activity (nearest-neighbor interaction) $\beta > 0$
and a vertex activity (external field) $\lambda>0$ as follows:
\begin{equation}
  \label{eq:ising-partition}
Z_G(\beta, \lambda) \defeq \sum_{\sigma: V \rightarrow \inbr{+, -}}
  \beta^{d(\sigma)}\lambda^{p(\sigma)},
\end{equation}
where $\sigma$ ranges over assignments of spins $\{+,-\}$ to vertices,
$d(\sigma)$ is the number of edges $\inbr{u, v}$ for which
$\sigma(u) \neq \sigma(v)$, and $p(\sigma)$ is the number of vertices for which
$\sigma(v) = +$.  As usual, the partition function implicitly defines the Gibbs
distribution on configurations~$\sigma$ according to their weights
in~\eqref{eq:ising-partition}.  When $\beta<1$ the model is \emph{ferromagnetic}
and assigns larger weight to configurations with more aligned neighboring spins;
conversely, when $\beta > 1$ the model is \emph{anti-ferromagnetic}.  For the
infinite $\Delta$-regular tree, it is known that weak spatial mixing holds when
$\beta \in (1, \frac{\Delta}{\Delta-2})$ for all $\lambda > 0$, while for
$\beta > \frac{\Delta}{\Delta-2}$ there exists a $\lambda_c(\beta, \Delta) > 0$
such that weak spatial mixing holds if
$\abs{\log \lambda} > \log \lambda_c(\beta, \Delta)$~\cite[p.~255; see also the
remarks following Theorems~1 and~3 of
\cite{sinclair_approximation_2012}]{georgii_hans-otto_gibbs_2011}.  We show
that, in a complex neighborhood of this weak spatial mixing region, the model
has no zeros in the parameter\footnote{Since $Z_G(\beta,\lambda)$ is a bivariate
  polynomial, one can investigate phase transitions in terms of either~$\beta$
  or~$\lambda$.  Complex zeros of $Z_G$ as a function of~$\lambda$ are known as
  {\it Lee-Yang zeros}, while zeros of $Z_G$ as a function of~$\beta$ are known
  as {\it Fisher zeros}.}~$\lambda$.
\begin{theorem}
  \label{thm:af-ising-intro} Fix $\Delta \geq 3$ and let $\beta > 1$ and
  $\lambda > 0$ be such that weak spatial mixing for the anti-ferromagnetic
  Ising model with edge activity $\beta$ and vertex activity $\lambda$ holds on
  the infinite $\Delta$-regular tree.  Then there exists
  $r_{\beta, \lambda, \Delta} > 0$ such that $Z_G(\beta, \lambda') \neq 0$ for
  any $\lambda' \in \C$ satisfying
  $\abs{\lambda - \lambda'} \leq r_{\beta, \lambda, \Delta}$.
\end{theorem}
\begin{remark}
  Note that the width $r_{\beta, \lambda, \Delta}$ of the region depends
  on~$\beta$ and $\lambda$, and indeed tends to zero as the parameters approach
  their critical values.  For any fixed {\it compact\/} subset of values of
  $(\beta,\lambda)$ within the regime of weak spatial mixing, we get a
  fixed width~$r$ for the region.
\end{remark}

Our second main result goes beyond the setting where a translation from weak
spatial mixing on the infinite tree to strong spatial mixing on general graphs
is known, and considers the anti-ferromagnetic Potts model.  Even in this more
general setting, we show that currently known arguments for proving strong
spatial mixing for the model can be ``lifted'' to the complex plane to prove new
zero-freeness results for the Potts model partition function.  We now formally
describe these results.

The partition function of the anti-ferromagnetic Potts model (at zero field)
of a graph $G = (V, E)$ with a fixed number $q$ of spins (which we often refer to as
``colors'') can be written as
\begin{equation}\label{eq:partitionfun}
Z_{G}(q; w) := \sum_{\sigma:V\to[q]} w^{m(\sigma)}.
\end{equation}
Here $\sigma$ ranges over arbitrary assignments of spins (colors) to vertices,
and $m(\sigma)$ is the number of {\it monochromatic\/} edges, i.e., edges $\inbr{u, v}$ 
for which $\sigma(u) = \sigma(v)$.
Note that the number of \emph{proper} $q$-colorings of~$G$
(i.e., those with no monochromatic edges) is just~$Z_G(q; 0)$.  
The partition function again defines the Gibbs distribution on
colorings~$\sigma$ of $G$, according to their weights
in~\eqref{eq:partitionfun}.  We will often drop $q$ from the notation since it
will be clear from the context, and write $Z_{G}(q; w)$ simply as $Z_G(w)$.

\begin{theorem}\label{thm:zeros-intro}
  Fix a positive integer $\Delta$.  Then there exists a $\tau_\Delta > 0$ such
  that the following is true.  Let $\mathcal{D}_\Delta$ be a simply connected
  region in the complex plane obtained as the union of disks of radius
  $\tau_\Delta$ centered at all points on the segment $[0, 1]$.  For any graph
  $G$ of maximum degree at most $\Delta \geq 3$ and integer~$q \geq 2\Delta$,
  we have $Z_{G}(q; w) \neq 0$ when $w \in {\mathcal{D}}_\Delta$.
\end{theorem}

\begin{remark} The condition $q \geq 2\Delta$, as discussed in more detail
  below, corresponds to the Dobrushin uniqueness condition and the ``path
  coupling'' method for analyzing Gibbs samplers.  The previous best result in
  this direction is due to Bencs, Davies, Patel and Regts~\cite{bencs2018zero}
  and requires $q \geq e\Delta + 1$, which is much stronger than the Dobrushin
  bound. \end{remark}

As discussed later in Section\nobreakspace \ref {sec:technical-overview}, our technique is also
capable of directly harnessing tighter strong spatial mixing arguments used in
the analysis of Markov chains for special classes of graphs.  As an example, we
can exploit such an argument of Gamarnik, Katz and Misra~\cite{Gamarnik2015SSM}
to improve the bound on $q$ in~Theorem\nobreakspace \ref {thm:zeros-intro}
when the graph is triangle-free,
for all but small values of~$\Delta$. \begin{theorem}\label{thm:zeros-intro-176}
  Let $\alphastar \approx 1.7633$ be the unique positive solution of the
  equation $xe^{-1/x}=1$.  For every $\alpha > \alphastar$, there exists a
  $\beta = \beta(\alpha)$ such that for any integer $\Delta \geq 3$, there
  exists a $\tau_\Delta > 0$ for which the following is true.  Let
  $\mathcal{D}_\Delta$ be a simply connected region in the complex plane
  obtained as the union of disks of radius $\tau_\Delta$ centered at all points
  on the segment $[0, 1]$.  For any \emph{triangle-free} graph $G$ of maximum
  degree at most $\Delta$ and integer~$q \geq \alpha\Delta + \beta$ , we have
  $Z_{G}(q; w) \neq 0$ when $w \in {\mathcal{D}}_\Delta$.
\end{theorem}

Finally, for the special case of trees,
the argument leading to the above theorems also leads to the following improved
bound.
\begin{proposition}
  \label{thm:zeros-trees}
  Fix an integer $\Delta \geq 3$, and let $q \geq \Delta + 1$.  Then, there
  exists a $\tau_\Delta > 0$ for which the following is true.  Let
  $\mathcal{D}_\Delta$ be a simply connected region in the complex plane
  obtained as the union of disks of radius $\tau_\Delta$ centered at all points
  on the segment $[0, 1]$. Then, for every tree $T$ of maximum degree $\Delta$,
  we have $Z_{T}(q; w) \neq 0$ when $w \in {\mathcal{D}}_\Delta$.
\end{proposition}

We note that one can directly analyse the roots of the partition function
of the Potts model on a tree, and they all lie at the points $0$ and $1 - q$.
Nevertheless, we include the above observation since the bound $q\geq \Delta+1$
matches the optimal number of colors in the classical result of
Jonasson~\cite{jonasson} showing {\it weak\/} spatial mixing, and while proving
{\it strong\/} spatial mixing for $q\ge\Delta+1$ remains an important open
problem, we find it interesting that one can obtain the above bound for the
related property of zero-freeness using an argument that is based on spatial
mixing.  Further, the argument leading to Proposition~\ref{thm:zeros-trees}
appears to be robust, and we expect that it may extend to more general settings
(e.g., that of \emph{list colorings} discussed later in this section).

For ease of later reference, we record the above three results in
the following.
\begin{theorem}\label{thm:zeros}
  Fix an integer $\Delta \geq 3$. Then there exists a $\tau_\Delta > 0$ such
  that for the simply connected region $\mathcal{D}_\Delta$ in the complex plane
  obtained as the union of disks of radius $\tau_\Delta$ centered at all points
  on the segment $[0, 1]$, the following is true.  For any graph $G$ of maximum
  degree $\Delta$ and integer~$q$ satisfying the hypotheses of
  Theorem\nobreakspace \ref {thm:zeros-intro}, Theorem\nobreakspace \ref
  {thm:zeros-intro-176} or Proposition~\ref{thm:zeros-trees}, $Z_G(q; w) \neq 0$
  when $w \in {\mathcal{D}}_\Delta$.
\end{theorem}

  \begin{remark}
    Our proof of the above theorem actually holds under an abstract condition on
    coloring instances that we call \emph{admissibility} (see
    Definition~\ref{def:good-list-condition}).  We show in
    Section~\ref{sec:prop-real-valu} that the instances covered in Theorems~\ref
    {thm:zeros-intro} and \ref {thm:zeros-intro-176} and Proposition~\ref{thm:zeros-trees} are
    all admissible.  Proving admissibility for any larger class of instances would
    immediately extend Theorem~\ref{thm:zeros} to such a class.
  \end{remark}

There has been extensive previous work on the complex zeros of the anti-ferromagnetic
Potts model, which we now briefly summarize.  
Sokal~\cite{sokal2001bounds,sokal2005multivariate} proved (in the
language of the Tutte polynomial) that the partition function has no zeros
in the entire unit disk centered at $w=0$, under the strong condition
$q\ge 7.964\Delta$; the constant was later improved to 6.907 by Fern\'andez and
Procacci~\cite{fernandez2008regions} (see also~\cite{jps}).  The results in
these papers, since they are in terms of the Tutte polynomial, in fact extend to
\emph{complex} values of $q$---a setting which is not accessible in the
Potts model formulation we use---and hence are not directly comparable to our
result.  However, for the most natural case of positive integer~$q$,
our result significantly improves upon them.  Much more recently, the work of Bencs et
al.~\cite{bencs2018zero} referred to above gave a zero-free region analogous to
that in Theorem~\ref{thm:zeros-intro} above, but under the stronger condition
$q\ge e\Delta + 1$.  We note also that Barvinok and
Sober\'on~\cite{BarvinokSoberon16a} (see also \cite{barvinok2017combinatorics}
for an improved version) established a zero-free region in a disk centered at $w=1$
of radius significantly less than~1.

\subsection{Algorithmic implications for the problem of counting colorings}
\label{sec:algor-impl-probl}

The above theorems also allow us to make progress on a long-standing open
problem in theoretical computer science: that of approximately counting proper
colorings of a bounded degree graph using a deterministic algorithm.  The
problem of counting colorings is a benchmark problem in the theory of
approximate counting, due both to its importance in combinatorics and
statistical physics, as well as to the fact that it has repeatedly challenged existing
algorithmic techniques and stimulated the development of new ones.  Below, we
briefly summarize its history and current status.

Given a finite graph $G=(V,E)$ of maximum degree~$\Delta$, and a positive
integer~$q$, the goal is to count the number of (proper) vertex colorings of~$G$
with $q$ colors.  It is well known~\cite{brooks_1941} that a greedy coloring exists
if $q\ge \Delta+1$.  While counting colorings exactly is \#P-complete, a
long-standing conjecture asserts that approximately counting colorings is
possible in polynomial time provided $q\ge \Delta+1$.  It is known that when
$q < \Delta$, even approximate counting is NP-hard~\cite{gsv}.

This question has led to numerous algorithmic developments over the past 25 years.
The first approach was via Markov chain Monte Carlo (MCMC), based on the fact that
approximate counting can be reduced to sampling a coloring (almost) uniformly at random.
Sampling can be achieved by simulating a natural local Markov chain (or Glauber dynamics)
that randomly flips colors on vertices: provided the chain is rapidly mixing, this leads to an
efficient algorithm (a {\it fully polynomial randomized approximation scheme}, or
{\it FPRAS\/}).  

Jerrum's 1995 result~\cite{jerrum1995very} that the Glauber dynamics
is rapidly mixing for $q\ge 2\Delta$ gave the first non-trivial randomized approximation
algorithm for colorings and led to a plethora of follow-up work on MCMC (see, e.g.,
\cite{dyer2003randomly,dfhv,frieze2006randomly,goldberg2004strong,hayes2003randomly,hayes2003non,hayes2005coupling,molloy2004glauber,vigoda2000improved} 
and~\cite{frieze2007survey} for a survey), focusing on reducing the constant~2 in front of~$\Delta$.
The best constant known for general graphs remains essentially $\frac{11}{6}$, obtained 
by Vigoda~\cite{vigoda2000improved} using a more sophisticated Markov chain, though this 
was recently reduced to $\frac{11}{6} - \varepsilon\,$ for a very small~$\varepsilon$
by Chen et al.~\cite{chen2019improved}.  The constant can be substantially 
improved if additional restrictions are
placed on the graph: e.g., Dyer et al.~\cite{dfhv} achieve roughly $q\ge 1.49\Delta$
provided the girth is at least~6 and the degree is a large enough constant, while
Hayes and Vigoda improve this to $q\ge(1+\varepsilon)\Delta$ for girth at least~11
and degree $\Delta=\Omega(\log n)$, where $n$ is the number of vertices.

A significant recent development in approximate counting is the emergence of
{\it deterministic\/} approximation algorithms that in some cases match, or even
improve upon, the best known MCMC algorithms\footnote{In this case, the notion
  of an FPRAS is replaced by that of a {\it fully polynomial time approximation
    scheme}, or {\it FPTAS}.  An FPTAS for $q$-colorings of graphs of maximum
  degree at most $\Delta$ is an algorithm that, given as input the graph $G$
  and an error parameter $\varepsilon\in(0,1)$, produces a $(1\pm\varepsilon)$-factor
  multiplicative approximation of the number of $q$-colorings of $G$ in time
  $\poly{|G|, 1/\varepsilon}$. (The degree of the polynomial is allowed to depend upon
  the constants $q$ and~$\Delta$.)  }.
Interestingly, these algorithms have made use of one of two main techniques,
both of which are inspired by the two different notions of phase transitions
in statistical physics described above.  The first, based on {\it decay of
correlations}, exploits the decreasing influence of the spins (colors) on distant
vertices on the spin at a given vertex; while the second, based on {\it
polynomial interpolation}, uses the absence of zeros of the partition function
in a suitable region of the complex plane to perform a form of algorithmic
analytic continuation.  Early examples of the decay of the decay of correlations
approach include~\cite{Weitz,bandyopadhyay_counting_2008,bayati2007simple},
while for early examples of
the polynomial interpolation method we refer to the monograph of
Barvinok~\cite{barvinok2017combinatorics} (see also
\cite{barvinok2017weighted,helmuth2018algorithmic,patel2017deterministic,jenssen2019algorithms,guo2019zeros,liu2019ising,eldar2018approximating} for more recent examples).

Unfortunately, however, in the case of colorings on general bounded degree
graphs, these techniques have so far lagged well behind the MCMC algorithms
mentioned above.  One obstacle to getting correlation decay to work is the lack
of a
higher-dimensional analog of Weitz's beautiful algorithmic
framework~\cite{Weitz}, which allows correlation decay to be fully exploited via
strong spatial mixing in the case of spin systems with just two spins (as
opposed to the $q\ge 3$ colors present in coloring).  For polynomial interpolation,
the obstacle has been a lack of precise information about the location of the
zeros of associated partition functions.

So far, the best algorithmic condition for colorings obtained via correlation
decay is $q\ge 2.58\Delta + 1$, due to Lu and Yin~\cite{lu2013improved}, and this remains
the best available condition for any deterministic algorithm.  This improved on
an earlier bound of roughly $q\ge 2.78\Delta$ (proved only for triangle-free
graphs), due to Gamarnik and Katz~\cite{gamarnik_correlation_2007}.  For the special case
$\Delta=3$, Lu et al.~\cite{lu2017fptas} give a correlation decay algorithm for
counting 4-colorings.  Furthermore, Gamarnik, Katz and
Misra~\cite{Gamarnik2015SSM} establish the related property of ``strong spatial
mixing" under the weaker condition $q\ge \alpha\Delta + \beta$ for any constant
$\alpha>\alphastar$, where $\alphastar\approx 1.7633$ is the unique solution to
$xe^{-1/x} = 1$ and $\beta$ is a constant depending on~$\alpha$, and under the
assumption that $G$ is triangle-free (see also \cite{ge2011strong,goldberg2004strong} for similar results
on restricted classes of graphs).  However, as discussed in
\cite{Gamarnik2015SSM}, this strong spatial mixing result unfortunately does not
lead to a deterministic algorithm\footnote{The strong spatial mixing condition
  does imply fast mixing of the Glauber dynamics, and hence an FPRAS, but only
  when the graph family being considered is ``amenable'', i.e., if the size of
  the $\ell$-neighborhood of any vertex does not grow exponentially with~$\ell$.
  This restriction is satisfied by regular lattices, but fails, e.g., for random
  regular graphs.}.

The newer technique of polynomial interpolation, pioneered by
Barvinok~\cite{barvinok2017combinatorics}, has also recently been brought to bear on counting
colorings.  In a recent paper, Bencs et al.~\cite{bencs2018zero} use this technique to
derive a FPTAS for counting colorings provided $q\ge e\Delta + 1$.  Although this 
result is weaker than those obtained via correlation deay, it
is of independent interest because it uses a different algorithmic approach, and
because it establishes a new zero-free region for the associated partition
function in the complex plane (see below).

In this paper, we push the polynomial interpolation method further and obtain a
FPTAS for counting colorings under the condition $q\ge 2\Delta$.
\begin{theorem}\label{thm:main}
  Fix positive integers $q$ and $\Delta$ such that $q \geq 2\Delta$.  Then there
  exists a fully polynomial time deterministic approximation scheme (FPTAS) for
  counting $q$-colorings in any graph of maximum degree~$\Delta$.
\end{theorem}

This is the first deterministic algorithm (of {\it any\/} kind) that for all
$\Delta$ matches the ``natural'' bound for MCMC, first obtained by
Jerrum~\cite{jerrum1995very}.  Indeed, {$q \geq 2\Delta$} remains the best bound known
for rapid mixing of the basic Glauber dynamics that does not require either
additional assumptions on the graph or a spectral comparison with another Markov
chain: all the improvements mentioned above require either lower bounds on the
girth and/or maximum degree, or (in the case of Vigoda's result~\cite{vigoda2000improved})
analysis of a more sophisticated Markov chain.  This is for good reason, since
the bound {$q \geq 2\Delta$} coincides with the closely related Dobrushin
uniqueness condition from statistical physics~\cite{salas1997absence}, which in turn
is closely related~\cite{weitz2005combinatorial} to the path coupling method of Bubley and
Dyer~\cite{bubley1997path} that provides the simplest currently known proof of the
{$q\geq 2\Delta$} bound for the Glauber dynamics.

We therefore view our result also as a promising starting point for
deterministic coloring algorithms to finally compete with their randomized
counterparts.
As pointed out above, our technique is capable of harnessing strong spatial
mixing arguments used in the analysis of Markov chains for certain classes of
graphs in order to relax the requirements on $q$.  In particular, for the same
reason as in \MakeUppercase Theorem\nobreakspace \ref {thm:zeros-intro-176} above, we can exploit such an argument
of Gamarnik, Katz and Misra~\cite{Gamarnik2015SSM} to improve the bound on $q$
in~Theorem\nobreakspace \ref {thm:main} when the graph is triangle-free, for all but small values
of~$\Delta$. (Recall that $\alphastar \approx 1.7633$ is the unique positive
solution of the equation $xe^{-1/x}=1$.)
\begin{theorem}\label{thm:main-176}
  For every $\alpha > \alphastar$, there exists a $\beta = \beta(\alpha)$ such
  that the following is true.  For all integers $q$ and $\Delta$ such that
  $q \geq \alpha\Delta + \beta$, there exists a fully polynomial time
  deterministic approximation scheme (FPTAS) for counting $q$-colorings in any
  triangle-free graph of maximum degree~$\Delta$.
\end{theorem}

We mention also that our technique applies without further effort to the
more general setting of {\it list\/} colorings, where each vertex has a list of allowed
colors of size~$q$, under the same conditions as above on~$q$.
Indeed, our proofs are written to handle this more general situation.

We now describe in more detail the connection between our results on the zeros of the Potts
model and the above algorithmic results.

\subsubsection{Our approach}
Recall that the polynomial $Z_G(w)$ in eq.\nobreakspace \textup{(\ref {eq:partitionfun})}, being the
partition function of the Potts model, implicitly defines
a probability distribution on colorings~$\sigma$ according to their weights
in~\eqref{eq:partitionfun}.  The parameter $w$ measures the strength of
nearest-neighbor interactions.  The value $w=1$ corresponds to the trivial
setting where there is no constraint on the colors of neighboring vertices,
while $w=0$ imposes the hard constraint that no neighboring vertices receive the
same color.  For intermediate values $w\in (0,1)$, neighbors with the same color
are penalized by a factor of~$w$.  We establish Theorems~\ref {thm:main} 
and~\ref {thm:main-176} as special cases of the following more general theorem.
\begin{theorem}\label{thm:mainPotts}
  Suppose that the hypotheses of either Theorem\nobreakspace \ref {thm:main} or Theorem\nobreakspace \ref {thm:main-176}
  are satisfied, and fix $w \in [0, 1]$. Then there exists an FPTAS for the
  partition function $Z_G(w)$.
\end{theorem}
Note that Theorems~\ref {thm:main} and~\ref {thm:main-176} follow immediately
as the special case $w=0$ of Theorem~\ref {thm:mainPotts}; however, 
the extension to other values of~$w$ is of independent
interest as the computation of partition functions is a very active area of
study in statistical physics and combinatorics.

Theorem~\ref{thm:mainPotts} is obtained from our 
main result, \MakeUppercase Theorem\nobreakspace \ref {thm:zeros}, by
appealing to a recent algorithmic paradigm of
Barvinok~\cite{barvinok2017combinatorics}. The paradigm (see Lemma~2.2.3
of~\cite{barvinok2017combinatorics}) states that, for a partition function~$Z$
of degree~$m$, if one can identify a simply connected, zero-free
region~$\mathcal{D}$ for~$Z$ in the complex plane that contains a
$\tau$-neighborhood of the interval $[0, 1]$, and a point in that interval where
the evaluation of~$Z$ is easy (in our setting this is the point $w = 1$), then
using the first ${O}\inp{e^{\Theta(1/\tau)}\log\inp{m/\eps}}$ coefficients of
$Z$ one can obtain a $1\pm\eps$ multiplicative approximation of~$Z(x)$ at any
point $x\in{\mathcal{D}}$.  Barvinok's framework is based on exploiting the fact
that the zero-freeness of $Z$ in $\mathcal{D}$ is equivalent to $\log Z$ being analytic
in~$\mathcal{D}$, and then using a carefully chosen transformation to deform~$\mathcal{D}$
into a disk (with the easy point at the center) in order to obtain a convergent Taylor
expansion.  The coefficients of $Z$ are used to compute the coefficients of this
Taylor expansion.

Barvinok's framework in general leads to a {\it quasi-polynomial time\/} algorithm, because
the computation of the ${O}\inp{e^{\Theta(1/\tau)}\log\inp{m/\eps}}$ terms of the
expansion may take time ${O}\bigl({\inp{m/\eps}^{e^{\Theta(1/\tau)}\log m}}\bigr)$
(which is only quasi-polynomial in $m$) for the
partition functions considered here.  However, additional insights provided by
Patel and Regts~\cite{patel2017deterministic} (see, e.g., the proof of Theorem 6.2 in
\cite{patel2017deterministic}) show how to reduce this computation time to
${O}\bigl({\inp{m/\eps}^{e^{\Theta(1/\tau)}\log \Delta}}\bigr)$ for many models on
graphs of degree at most~$\Delta$, including the Potts model with a bounded
number of colors~$q$ at each vertex.  Hence we obtain an FPTAS.  
This (by now standard) reduction is the same path as that followed by
Bencs et al.~\cite[Corollary~1.2]{bencs2018zero}; for completeness, we sketch some
of the details in Section~\ref{sec:app_algorithm}.
We note that, for each fixed~$\Delta$ and~$q$, the running time of our final
algorithm is polynomial in~$n$ (the size of~$G$) and~$\varepsilon^{-1}$, as
required for an FPTAS.  However, as is typical of deterministic algorithms for
approximate counting, the exponent in the polynomial depends on~$\Delta$
(through the quantity $\tau_\Delta$ in Theorem\nobreakspace \ref {thm:zeros}, 
which in the case where all lists are subsets of $[q]$ is inverse polynomial in~$q$).

We conclude this introduction by sketching our approach to proving the
zero-freeness results, which constitute the main technical contribution of
the paper.

\subsection{Technical overview}
\label{sec:technical-overview}
We start with an outline of the proofs of our results for two-spin systems,
including Theorem~\ref{thm:af-ising-intro} and our simplified re-proofs of previous
zero-freeness results.  A standard observation in the
area is that proving $Z_G(\beta, \lambda) \neq 0$ is equivalent to showing
that the \emph{occupation ratio} $R_{v, G}(\beta, \lambda)$ at a fixed vertex $v$,
defined as the ratio of the sum of those terms in the partition function where
the vertex $v$ has spin $+$ to the sum of those terms in the partition function
where the vertex $v$ has spin $-$, is not equal to~$-1$.  In order to analyze this
quantity, another standard step is to use an observation of
Weitz~\cite{Weitz}\footnote{The ideas behind Weitz's reduction first appeared in
  the work of Godsil~\cite{godsil_matchings_1981}, and also later in the work of
  Scott and Sokal~\cite{ScottSokal}.},  which allows one to transfer the question
from general graphs to trees.  More precisely, for any fixed vertex $v$ in the
graph~$G$, Weitz's theorem constructs a finite tree $T = T_{v, G}$ (with carefully
chosen boundary conditions at the leaves), of maximum degree at most the
maximum degree of $G$, such that if $\rho$ is the root of $T$ then
$R_{\rho, T}(\beta, \lambda) = R_{v, G}(\beta, \lambda)$ for all positive real
$\beta, \lambda$.  On a tree, one can easily write down a recurrence for the
occupation ratio, and the problem then reduces to proving that, with initial
conditions corresponding to the boundary conditions of~$T$,
the recurrence never reaches~$-1$.

The convergence properties of such recurrences have been analyzed before, in the
context of proving that \emph{weak spatial mixing\/} (or uniqueness of the Gibbs
measure) on the infinite $\Delta$-regular tree implies \emph{strong spatial
  mixing\/} on all graphs of maximum degree~$\Delta$, for the hard core
model~\cite{Weitz,SSSY15,li_correlation_2011} and the Ising model with and
without field~\cite{sinclair_approximation_2012,zhaliabai09}.  These analyses,
which are restricted to positive, real values of the parameters, often take the
form of showing that the recurrence for an appropriate function $\phi(R)$ of the
occupation ratio is a uniform contraction.  Our main contribution is to show
that the arguments in the above references are in fact robust enough that one
can extend them to a {\it complex\/} neighborhood (independent of the size
of the graph) of the real intervals on which they hold.
Thus the behavior of the recurrence in this neighborhood
remains close to what one sees for positive real parameters, and in particular
the value of the occupation ratio remains away from $-1$, thus establishing
zero-freeness.

The situation is more complicated for the case of the Potts model (where the
number of spins is more than two), since neither the translation to trees, nor
the tight recurrence analyses for tree recurrences is known.  The starting
point for our proof of \MakeUppercase Theorem\nobreakspace \ref {thm:zeros} is a simple geometric observation,
versions of which have been used before for constructing inductive proofs of
zero-freeness of partition functions (see, e.g.,
\cite{barvinok2017combinatorics,bencs2018zero}). Fix a vertex $v$ in the graph
$G$.  Given $w \in \C$ and a color $k \in [q]$, let $\Z{v}{k}(w)$ denote the
\emph{restricted partition function} in which one sums only over those colorings
$\sigma$ in which $\sigma(v) = k$.  Then, since
$Z_G(w) = \sum_{k \in [q]}\Z{v}{k}(w)$, the zero-freeness of $Z_G$ will follow
if the angles between the complex numbers $\Z{v}{k}(w)$, viewed as vectors in
$\R^2$, are all small, and provided that at least one of the $\Z{v}{k}$ is
non-zero.  (In fact, this condition on angles can be relaxed for those
$\Z{v}{k}(w)$ that are sufficiently small in magnitude, and this flexibility will be
important for us when $w$ is a complex number close to $0$.)  Therefore, one is
naturally led to consider the so-called \emph{marginal ratios}:
\begin{displaymath}
  \ratio{G,v}{i,j}(w) \defeq \frac{\Z{v}{i}(w)}{\Z{v}{j}(w)}.
\end{displaymath}
(In the $q$-coloring problem, this ratio is $1$ by symmetry. However, in our recursive approach we have to handle the more general list-coloring problem, in which the ratio becomes non-trivial.)
We then require that, for any two colors $i, j$ for which $\Z{v}{k}(w)$ is large
enough in magnitude, the ratio $\ratio{G,v}{i,j}(w)$ is a complex number with small
argument.
This is what we prove inductively in
Sections\nobreakspace \ref {sec:induction-origin} and\nobreakspace  \ref {sec:induction-interval}.

The broad contours of our approach as outlined so far are quite similar to some
recent work~\cite{barvinok2017combinatorics,bencs2018zero}.
However, it is at the crucial step of how the marginal ratios are analyzed that
we depart from these previous results.  Instead of attacking the restricted
partition functions or the marginal ratios directly for given $w \in \C$, as in
these previous works, we crucially exploit the fact that for any real
$\tw \in [0, 1]$ close to the given $w$, these quantities have natural
probabilistic interpretations, and hence can be much better understood via
probabilistic and combinatorial methods.  For instance, when $\tw \in [0, 1]$,
the marginal ratio $\ratio{G,v}{i,j}(w)$ is in fact a ratio of the marginal
probabilities $\Pr{G,\tw}{\sigma(v) = i}$ and $\Pr{G,\tw}{\sigma(v) = j}$, under
the natural probability distribution on colorings~$\sigma$.
In fact, our analysis cleanly breaks into two separate
parts:
\begin{enumerate}
\item First, understand the behavior of true marginal probabilities of
  the form $\Pr{G,\tw}{\sigma(v) = i}$ for real $\tw \in [0, 1]$.  This is carried
  out in Section\nobreakspace \ref {sec:prop-real-valu}.
\item Second, argue that, for complex $w \approx \tw$,
	the ratios $\ratio{G,v}{i,j}(w)$ remain well-behaved.  This is carried
  out separately for the two cases when $w$ is close to $0$ (in
  Section\nobreakspace \ref {sec:induction-origin}) and when $w$ is bounded away from $0$ but still
  in the vicinity of $[0, 1]$ (in Section\nobreakspace \ref {sec:induction-interval}).
\end{enumerate}

A key technical point in our analysis is the notion of ``\nice{}ness'' of
vertices, which stipulates that the marginal
probability $\Pr{G,\tw}{\sigma(v) = i} \le \frac{1}{\deg_G(v) +2}$ where $\deg_G(v)$ is the degree of~$v$ in~$G$
(see~Definition\nobreakspace \ref {cond:nice}).
Note that this condition
refers only to real non-negative~$\tw$, and hence is amenable to
analysis via standard combinatorial tools.  Indeed, 
our proofs that the conditions on $q$ and $\Delta$
in~Theorems\nobreakspace \ref {thm:main} and\nobreakspace  \ref {thm:main-176} imply this \nice{}ness condition are
similar to probabilistic arguments used by Gamarnik et al.~\cite{Gamarnik2015SSM}
to establish strong spatial mixing (in the special case $\tw = 0$).  
We emphasize that this is the only place in our analysis where the lower bounds
on~$q$ are used.
One can therefore expect
that combinatorial and probabilistic ideas used in the analysis of strong
spatial mixing and the Glauber dynamics with a smaller number of colors in special
classes of graphs can be combined with our analysis to obtain deterministic
algorithms for those settings; indeed, our Theorem~\ref{thm:main-176}
demonstrates ths point for triangle-free graphs, leveraging the strong spatial
mixing argument of~\cite{Gamarnik2015SSM}.

The above ideas are sufficient to understand the real-valued case (part~1 above).
For the complex case in part~2, we start from a recurrence for the marginal
ratios $\ratio{G,v}{i,j}$ that is a generalization
(to the case $w \neq 0$) of a similar recurrence used by 
Gamarnik et al.~\cite{Gamarnik2015SSM}; this recurrence is defined
in \MakeUppercase Lemma\nobreakspace \ref {thm:recurrence}. 
The inductive proofs in Sections\nobreakspace \ref {sec:induction-origin} and\nobreakspace  \ref {sec:induction-interval} use this
recurrence to show that, if $\tw \in [0, 1]$ is close to
$w \in \C$, then all the relevant $\ratio{G,v}{i,j}(w)$ remain close to
$\ratio{G,v}{i,j}({\tw})$ throughout.  The actual induction,
especially in the case when $w$ is close to $0$, requires a delicate choice of
induction hypotheses (see Lemmas\nobreakspace \ref {lem:origin-induction} and\nobreakspace  \ref {lem:interval-induction}).
The key technical idea is to use the ``niceness'' property of vertices established in part~1
to argue that the two recurrences (real and complex) remain close at every
step of the induction.
This in turn depends upon a careful application of the mean value theorem,
\emph{separately} to the real and imaginary parts (see Lemma\nobreakspace \ref {thm:signed-mean}),
of a function $f_\kappa$ that arises naturally in the analysis of the recurrence (see
Lemma\nobreakspace \ref {obv:f-props-int}).

\subsubsection{Comparison with correlation-decay based algorithms} We conclude
this overview with a brief discussion of how we are able to obtain a better
bound on the number of colors than in correlation decay algorithms,
such as~\cite{gamarnik_correlation_2007,lu2013improved} cited earlier.
In these algorithms, one first uses recurrences similar to the one mentioned above to
\emph{compute} the marginal probabilities, and then appeals to self-reducibility
to compute the partition function.  Of course, expanding
the full tree of computations generated by the recurrence will in general give
an exponential time (but exact) algorithm.  The core of the analysis of these
algorithms is to exploit the correlation decay property to show that, even if this
tree of computations is only expanded to
depth about $O(\log(n/\eps))$, and the recurrence at that point is initialized
with \emph{arbitrary} values, the computation still converges to an
$\eps$-approximation of the true value.  However, the requirement that the
analysis be able to deal with arbitrary initializations implies that one cannot
directly use properties of the actual probability distribution (e.g., the
``\nice{}ness'' property alluded to above);  indeed, this issue is also pointed
out by Gamarnik et al.~\cite{Gamarnik2015SSM}.
In contrast, our analysis does not truncate the recurrence,
and thus only has to handle initializations that make sense in the context
of the graph being considered.  Moreover, the exponential size of the 
recursion tree is no longer a barrier for us since, in contrast to correlation
decay algorithms, we are using the tree only as a tool to establish zero-freeness;
the algorithm itself follows from Barvinok's polynomial interpolation paradigm.
Our approach suggests that this paradigm can be viewed as a method for
using (complex-valued generalizations of) strong spatial mixing results
to obtain deterministic approximation algorithms.

\newcommand{\betarange}{\ensuremath{(\frac{\Delta - 2}{\Delta}, \frac{\Delta }{\Delta - 2})}}
\newcommand{\bt}{\ensuremath{{{\beta'}}}}
\newcommand{\fpk}{f_{\beta,k}^{\varphi}}
\newcommand{\fpsk}{f_{\beta,k,s}^{\varphi}}
\newcommand{\Fpsk}{F_{\beta,k,s}^{\varphi}}
\newcommand{\Fpbsk}{F_{\beta',k,s}^{\varphi}}
\newcommand{\fpb}{f_{\beta',k}^{\varphi}}
\section{Correlation decay implies absence of zeros}
In this section, we present a sequence of results relating correlation decay and
the absence of zeros for two-spin systems.\footnote{{The results in this
    section were first reported in JL's PhD thesis~\cite{JingchengThesis}.
    Subsequently, similar results, in a slightly more general context, have
    independently been obtained by Shao and Sun~\cite{shaosun19}.}}  In addition
to their intrinsic interest, these results will also serve as a ``warm-up'' to
our results on the Potts model, which use similar methods in a more complex
setting.  We begin by
re-proving the main result of~\cite{liu2018fisher} on the Fisher zeros of the
Ising model (without external field).  While the proof in~\cite{liu2018fisher}
also implicitly used correlation decay, here we rewrite the argument as a
special case of a more general method for ``lifting'' already known correlation
decay results for various models to the complex plane.  We go on to apply this
generic method to prove new results on the Lee-Yang zeros of the
anti-ferromagnetic Ising model (with field), and to give a new, simpler proof of
the Sokal conjecture (first proved by Peters and Regts~\cite{peters17:_sokal})
on the zeros of the hard core partition function.  The ideas developed here will
be extended to the Potts model in the later sections of the paper.
\subsection{Ising model}
\label{sec:outline-proof}
In this section we show that there are no {\it Fisher zeros\/} of the Ising model
in a complex neighborhood around the correlation
decay interval of the infinite $\Delta$-regular tree (Bethe lattice).  This gives a different proof of the main
result of~\cite{liu2018fisher}, making the role of the correlation decay
arguments in the real domain more explicit.

Recall from~eq.\nobreakspace \textup {(\ref{eq:ising-partition})} 
that, given a graph $G$, an edge activity~$\beta$ and a vertex 
activity~$\lambda$, the Ising partition function is defined as
$Z_G(\beta, \lambda) = \sum_\sigma \beta^{d(\sigma)}\lambda^{p(\sigma)}$,
where $d(\sigma)$ is the number of edges between different spins,
and $p(\sigma)$ is the number of vertices with spin~$+$.
Formally, we view this partition function as a polynomial in $\beta$ for a fixed
$\lambda$, and study the complex zeros in $\beta$; these are known as 
{\it Fisher zeros}.  In fact, in this section we
fix $\lambda=1$, and hence we will simply write $Z_G(\beta)\defeq Z_G(\beta,1)$
for the rest of the section.  The correlation decay interval for the Ising model
has been well studied:
the Gibbs distribution of the
Ising model on any graph of maximum degree~$\Delta$ exhibits decay
of correlations when $\beta$ lies in the interval $\betarange$~\cite{zhaliabai09},
which corresponds exactly to the correlation decay interval for the
$\Delta$-regular tree~\cite{georgii_hans-otto_gibbs_2011}.
The main result of this section will be~Corollary\nobreakspace
\ref {thm:main-fisher}, which says that there is a complex neighborhood of the
correlation decay interval in which there are no Fisher zeros for the Ising model
on any graph of maximum degree~$\Delta$.  This provides a formal link between the 
``decay of correlations'' and ``analyticity of free energy density'' views of phase
transitions.  Further, as discussed in more detail in~\cite{liu2018fisher}, this
zero-freeness result also implies the existence of efficient approximation
algorithms for the partition function $Z_G(\beta)$ via Barvinok's paradigm
discussed in~Section\nobreakspace \ref {sec:app_algorithm}.

We recall some notation and definitions from~\cite{liu2018fisher}.  Let
$G$ be any graph of maximum
degree~$\Delta$.  For any non-isolated vertex~$v$ of~$G$, let $Z_{G,
  v}^+(\beta)$ (respectively, $Z_{G,
  v}^-(\beta)$) be the contribution to
$Z_{G}(\beta)$ from configurations with
$\sigma(v)=+$ (respectively, $\sigma(v)=-$), so that $Z_{G}(\beta) = Z_{G,
  v}^+(\beta) + Z_{G, v}^-(\beta)$.  We also define the ratios $R_{G,v}(\beta)
:= \frac{Z_{G, v}^+(\beta)}{Z_{G, v}^-(\beta)}$.  Note that $Z_{G,
  v}^+(\beta)$ and $Z_{G,
  v}^-(\beta)$ can be seen as Ising partition functions defined on the same
graph $G$ with the vertex
$v$ \emph{pinned} to the appropriate spin.  Without loss of generality, we
assume that every pinned vertex has degree exactly one.\footnote{ Suppose that a
  vertex $v$ of degree $k$ is pinned in a graph $G$, and consider the graph
  $G'$ obtained by replacing $v$ with
  $k$ copies of itself, each pinned to the same spin and connected to exactly
  one of the original neighbors of $v$.  Then $Z_G(\beta) =
  Z_{G'}(\beta)$ for all
  $\beta$.  } We will prove, inductively on the number of unpinned vertices,
that neither $Z_{G,v}^+(\beta)$ nor $Z_{G,
  v}^-(\beta)$ vanishes.  Under this induction hypothesis, the condition
$Z_{G}(\beta) \neq 0 $ is clearly equivalent to $R_{G, v}(\beta) \neq
-1$.  As we will see, for $\beta \in \R$, $R_{G, v}(\beta)
>0$.  Thus it suffices to show that, for complex
$\beta$ sufficiently close to the correlation decay interval on the real line,
$R_{G, v}(\beta) \approx R_{G, v}(\Re \beta) $.

As in~\cite{liu2018fisher}, our development in this section is also based on the formal
recurrences derived by Weitz~\cite{Weitz} for computing ratios such as
$R_{G, v}(\beta)$ in two-state spin systems.  However, instead of following
\cite{liu2018fisher}, where Weitz's reduction to the so-called self-avoiding walk
tree was used directly, we provide here a self-contained description in a form
that is a simplification of the more complicated recurrences for the Potts model
that we study in Sections~\ref{sec:preliminaries} and beyond.

We start with some notation and definitions.  For a vertex $u$ in a graph $G$,
if $u$ has $s^+$ neighbors pinned to spin~$+$, and $s^-$ neighbors pinned to
spin~$-$, then we say that $u$ has $(s^- - s^+)$ \emph{signed pinned neighbors}.

\begin{definition}[\textbf{The graphs $G_i$}]
  \label{def:graphs-gi}
  Given a graph $G$ and an unpinned vertex $u$ in $G$, let $v_1, \cdots, v_k$ be
  the unpinned neighbors of $u$.  We define $G_i$ (the vertex $u$ will be
  understood from the context) to be the graph obtained from $G$ as follows:
  \begin{itemize}
	  \item first, replace vertex $u$ with $u_1, \cdots, u_k$, and
    connect $u_1$ to $v_1$, $u_2$ to $v_2$, and so on;
  \item next, pin vertices $u_1, \cdots, u_{i-1}$ to spin $+$, and vertices
    $u_{i+1}, \cdots, u_{k}$ to spin $-$;
  \item finally, remove vertex $u_i$.
  \end{itemize}
  Note that the graph $G_i$ has one fewer unpinned vertex than $G$.
  Moreover, the number of unpinned neighbors of $v_i$ in $G_i$ is at most $\Delta-1$.
\end{definition}
\begin{lemma}
	Let $\omega$ be a formal variable. Given a graph $G$ and an unpinned vertex $u$, let $k$ be the number of unpinned neighbors of $u$, and $s$ be the number of signed pinned neighbors of $u$. 
	Defining $h_\omega(x) \defeq \frac{\omega + x}{\omega x + 1}$, we have
	\[
		R_{G,u}(\omega) = \omega^s \prod_{i=1}^k h_\omega\inp{ R_{G_i, v_i}(\omega) }.
	\]
	\label{lem:weitz-Fisher}
\end{lemma}
\begin{remark}
  (i)~Note that the above formal equalities becomes valid numerical equalities
  when a numerical value $\beta \in \C$ is substituted for
  $\omega$, provided that (a) $\beta x_i + 1 \neq 0$ for any $\vec{x}$ appearing in
  the computation, and (b) $Z_{G, u}^{-}(\beta) \neq 0$.  \ \ (ii)~Moreover, since the
  number of unpinned neighbors of $v_i$ in $G_i$ is at most $\Delta-1$, the tree
  recurrence will be applied with $k \leq \Delta - 1$ except possibly at the
  root where $k$ may be~$\Delta$.  
  \end{remark}

\begin{proof}
	Let $v_1, v_2, \cdots, v_k$ be the unpinned neighbors of $u$, and $v_{k+1}, \cdots, v_{\deg_G(u)}$ be its pinned neighbors.
  For $0 \leq i \leq \deg_G(u)$, let $H_i$ be the graph obtained from $G$ as follows:
  \begin{itemize}
  \item replace vertex $u$ with $u_1, \cdots, u_{\deg_G(u)}$, and
    connect $u_1$ to $v_1$, $u_2$ to $v_2$, and so on;
  \item pin vertices $u_1, \cdots, u_{i}$ to spin $+$, and vertices
    $u_{i+1}, \cdots, u_{\deg_G(u)}$ to spin $-$.
  \end{itemize}
  Note that $H_i$ is the same as $G_i$, except that the last step of the construction of $G_i$ is skipped, i.e, the vertex $u_i$ is not removed,
  and, further, $u_i$ is pinned to spin $+$.  We can now write
	\begin{align}  \label{eq:astmp1}
		R_{G,u}(\omega)
	=\frac{Z_{G,u}^{+}(\omega)}{Z_{G,u}^{-}(\omega)}
	=\frac{Z_{H_{\deg_G(u)}} (\omega)}{Z_{H_0}(\omega)}
	=\prod_{i=1}^{\deg_G(u)} \frac{Z_{H_i}(\omega)}{Z_{H_{i-1}}(\omega)} 
	=\omega^s \cdot \prod_{i=1}^{k} \frac{Z_{H_i}(\omega)}{Z_{H_{i-1}}(\omega)} ,
	\end{align}
where $k,s$ are the numbers of unpinned neighbors and signed pinned neighbors,
	respectively, of~$u$.
    We observe that
	\begin{align*}
		Z_{H_i}(\omega) = Z_{G_i,v_i}^{+} + \omega \cdot Z_{G_i,v_i}^{-};\\
		Z_{H_{i-1}}(\omega)= \omega \cdot Z_{G_i,v_i}^{+} + Z_{G_i,v_i}^{-}.
	\end{align*}
	Substituting these expressions into~eq.\nobreakspace \textup {(\ref{eq:astmp1})} gives
	\[
		R_{G,u}(\omega) =\omega^s \cdot \prod_{i=1}^{k}\frac{Z_{G_i,v_i}^{+} + \omega \cdot Z_{G_i,v_i}^{-}}{\omega \cdot Z_{G_i,v_i}^{+} + Z_{G_i,v_i}^{-}} 
		= \omega^s \cdot \prod_{i=1}^{k}\frac{\frac{Z_{G_i,v_i}^{+}}{Z_{G_i,v_i}^{-}} + \omega }{\omega \cdot \frac{Z_{G_i,v_i}^{+}}{Z_{G_i,v_i}^{-}} + 1}
		= \omega^s \prod_{i=1}^k h_\omega\inp{ R_{G_i, v_i}(\omega) }.
	\]
This completes the proof.
\end{proof}

Lemma\nobreakspace \ref {lem:weitz-Fisher} leads to the following recurrence
relation on the ratios:
\begin{equation}
  \label{eq:recurrence-tmp}
  F_{\beta, k, s}(\vec{x}) \defeq \beta^s \prod_{i=1}^k h_\beta(x_i),
\end{equation}
where as before $h_\beta(x) \defeq \frac{\beta +x}{\beta x + 1}$.  As in several
previous studies of this recurrence in the literature (see,
e.g.,~\cite{lyons_ising_1989,zhaliabai09,liu2018fisher}), it is useful to re-parameterize it in terms of
logarithms of likelihood ratios as follows.  Let $\varphi(x) \defeq \log x$ and
define
\begin{equation}
  \label{eq:recurrence-potential}
  \Fpsk(\bm x) \defeq \inp{\varphi \circ F_{\beta,k, s} \circ \varphi^{-1}} (\bm x) = s \log \beta + \sum_{i=1}^k \log h_\beta(e^{x_i}).
\end{equation}
One may then derive the correlation decay property in the form of a convenient
step-wise contraction~\cite{zhaliabai09}.  The version here (and its proof,
which we include for completeness) is taken from~\cite{liu2018fisher}.

\begin{proposition}
Fix a degree $\Delta \geq 3$ and integers $k \geq 0$ and $s$. If
$\frac{\Delta - 2}{\Delta} < \beta <\frac{\Delta }{\Delta - 2}$ then
  there exists $\eta > 0$ (depending upon $\beta$ and $\Delta$) such
  that $\norm{\Grad \Fpsk (\vec{x})}{1} \le \frac{k}{\Delta - 1} (1 - \eta)$ for every $\vec{x} \in \R^k$.
\label{lem:correlation-decay-Fisher}
\end{proposition}
\begin{proof}
  A direct calculation of the derivative gives
\begin{displaymath}
\norm{\Grad \Fpsk (\vec{x})}{1} =
	  \sum_{i=1}^k\frac{ \abs{1-\beta^2}}{\beta^2 +
      1 + \beta(e^{x_i} + e^{-x_i})}.
  \end{displaymath}
  Since $e^x + e^{-x} \ge 2$ for every real $x$, the right hand side is at most
  $k \times \frac{\abs{1-\beta}}{1+\beta}$.  The condition on $\beta$ 
  implies that $\frac{\abs{1-\beta}}{1+\beta} \le \frac{1-\eta}{\Delta - 1}$ for some
  fixed $\eta >0$.  Therefore, we have
  $\norm{\Grad \Fpsk (\vec{x})}{1} \le k \times \frac{\abs{1-\beta}}{1+\beta}
  \le \frac{k}{\Delta - 1} (1-\eta)$.
\end{proof}
We recall a few further computations from~\cite{liu2018fisher}.  First, we bound
$R_{G,u}(\beta)$ for real-valued $\beta$.  From \nobreakspace \textup
{(\ref{eq:recurrence-tmp})}, for any integers $k \geq 0$ and $s$ and a positive
real $\beta$, we have
$\beta^{k + \abs{s}} \le F_{\beta,k,s}(\vec{x}) \le \frac{1}{\beta^{k +
    \abs{s}}}$ when $\beta \leq 1$, and
$ \frac{1}{\beta^{k + \abs{s}}}\le F_{\beta,k,s}(\vec{x}) \le \beta^{k +
  \abs{s}}$ when $\beta\ge 1$, for all
$\vec{x} \in \bigl(\R_+\cup\{0,\infty\}\bigr)^k$.
Noting that $k+|s|\le\Delta$, and taking the logarithm of these bounds
motivates the definition of the intervals $I_0(\beta,\Delta)$ as follows:
\begin{equation}
  I_0 = I_0(\beta, \Delta) \defeq \inb{-\Delta\abs{\log\beta}, \Delta\abs{\log\beta}}.
\label{eq:intervalI0}
\end{equation}

Recalling~Lemma\nobreakspace \ref {lem:weitz-Fisher}, we see that the ratios $R_{G,u}(\beta)$ can be obtained by recursively applying the recurrence $F_{\beta,k,s}(\vec{x})$. Therefore, for $\beta \in \R_+$, any graph $G$ and unpinned vertex $u$, we have $\log R_{G,u}(\beta) \in I_0(\beta,\Delta)$.
Our second point of departure from the strategy followed in \cite{liu2018fisher}
is the following corollary of~Proposition\nobreakspace \ref {lem:correlation-decay-Fisher} in the complex plane.
\begin{corollary} 
Fix a degree $\Delta \geq 3$ and integers $k \geq 0$ and $s$. If
$\frac{\Delta - 2}{\Delta} < \beta <\frac{\Delta }{\Delta - 2}$ then there
  exist positive constants $\eta , \eps, \delta$ (depending upon $\beta$ and
  $\Delta$) such that the following is true.  Let $D \defeq D(\beta,\Delta)$ be
  the closed and convex set of points within distance $\eps$ of
  $I_0(\beta,\Delta)$ in $\C$.  Then
  $\norm{\Grad \Fpsk (\vec{x})}{1} \le (1 - \eta/2)$ for all $k, s$ satisfying
  $k + \abs{s} \leq \Delta - 1$ and every $\vec{x} \in
  D^k$.  Moreover, there is a finite constant $M \geq 1$ (depending upon
  $\beta$ and $\Delta$) such that\begin{align}
    \sup_{\vec{x} \in D^k,\; \bt \in \C: \abs{\bt - \beta} < \delta} \abs{
    \Fpsk(\vec{x}) - \Fpbsk(\vec{x})}
    &\leq M \abs{\beta - \bt}\label{eq:45};\\
    \sup_{x,y : \varphi(x), \varphi(y) \in D} \abs{\varphi(x) - \varphi(y)}
    &\leq M \abs{x - y}\label{eq:46};\\
    \sup_{x,y \in D} \abs{\varphi^{-1}(x) - \varphi^{-1}(y)}
    &\leq M \abs{x - y};\text{
      and}\label{eq:47}\\
    \sup_{\vec{x} \in D^k}
    \norm{\Grad \Fpsk (\vec{x})}{1} &\leq M \text{ when $k + \abs{s} = \Delta$}.\label{eq:48}
  \end{align}
\label{cor:correlation-decay-Fisher}
\end{corollary}
\begin{proof}
  Observe that
  $ \norm{\Grad \Fpsk (\vec{x})}{1} = \sum_{i=1}^k\frac{
    \abs{1-\beta^2}}{\beta^2 + 1 + \beta(e^{x_i} + e^{-x_i})} $ is a continuous
  function in $x_i$ for every~$i$.  Since by
  Proposition~\ref{lem:correlation-decay-Fisher} it is uniformly upper bounded
  by $\frac{k}{\Delta - 1}(1-\eta)$ for all $\vec{x}\in {I_0(\beta, \Delta)}^k$,
  for small enough~$\eps$ the expression can be bounded by
  $\frac{k}{\Delta - 1}(1-\eta/2)$ for all $\vec{x} \in D^k$; this in turn is bounded
  by $(1-\eta/2)$ when $k+|s|\le\Delta-1$.  

  Finally, the existence of $M$ follows from the analyticity of 
  $\Fpsk$ on~$D^k$, of $\varphi^{-1}$ on~$D$, and of $\varphi$ on $\varphi^{-1}(D)$, respectively.
\end{proof}
We will also need the following standard consequence of the mean value theorem
for complex functions (also used, e.g., in~\cite{liu2018fisher}).
\begin{lemma}
	Let $K(\vec{x})$ be a holomorphic function on a convex subset $D$ of $\C^k$.
For any $\vec{x}, \vec{x'} \in D$, we have
	\[
      \abs{K(\vec{x}) - K(\vec{x'})} \le \sup_{\vec{\xi} \in D^k}\norm{\Grad K
        (\vec{\xi})}{1} \cdot \norm{\vec{x} - \vec{x'} }{\infty}.
	\]
	\label{lem:mean-value}
\end{lemma}
\begin{proof}
  Consider $g(t)\defeq K\inp{\vec{x} + t(\vec{x'} - \vec{x}) } $ for
  $t \in [0, 1]$. Since $D$ is convex, $\vec{x} + t(\vec{x'} - \vec{x})$ lies in
  $D$ for all $t \in [0, 1]$.  Now, observe that
	\begin{align*}
		g'(t) = \Grad K \inp{\vec{x} + t(\vec{x'} - \vec{x})}^\intercal \inp{\vec{x} - \vec{x'}}. 
	\end{align*}
	Thus, for any $\vec{x}, \vec{x'} \in D$, we have
	\begin{align*}
		\abs{K(\vec{x}) - K(\vec{x'})} 
		=\abs{g(1) - g(0)} 
		=\abs{\int_{0}^1 g'(t) d t}
		\le \sup_{t\in[0,1]} \abs{g'(t)}
		\le \sup_{\vec{\xi} \in D^k}\norm{\Grad K (\vec{\xi})}{1} \cdot
      \norm{\vec{x} - \vec{x'} }{\infty}.
	\end{align*}
\end{proof}

Finally, we are ready to give a proof of the main result of this section (also
the main result of~\cite{liu2018fisher}).  The proof of the theorem below serves
as a template for our arguments for establishing zero-free regions for the Potts
model partition function in Sections~\ref{sec:induction-origin} and
\ref{sec:induction-interval}.
\begin{theorem}
  Fix a degree $\Delta \ge 3$, and let $\beta \in \betarange$.  There exist
  positive constants $\delta_0, \tau$ (both depending on $\beta$ and $\Delta$)
  such that, for any graph $G$ of maximum degree $\Delta$, any vertex $u$ in
  $G$, and any complex $\bt$ with $\abs{\bt - \beta} < \delta_0$, the following are
  true:
	\begin{enumerate}
		\item $|{Z_{G, u}^+(\bt)}|> 0, |{Z_{G, u}^-(\bt)}|> 0$.
		\item $\abs{\varphi\inp{ R_{G,u}(\beta)} - \varphi\inp{ R_{G,u}(\bt)}} <
          \tau$ if $u$ has degree at most $\Delta - 1$ in $G$.
        \item $Z_{G}(\bt) \neq 0$.
	\end{enumerate}
	\label{thm:ind-hyp-fisher}
  \end{theorem}
  We will refer to the two items above as the ``induction hypothesis.''  We
  remark that the assumption $\beta \in \betarange$ is needed only so that we
  may appeal to correlation decay (in the form of~Corollary\nobreakspace \ref
  {cor:correlation-decay-Fisher}).   
  \begin{proof}
	We use induction on the number of unpinned vertices in $G$.  Without loss of
    generality, we assume that the graph $G$ is connected.

    For the base case, if $u$ is the only unpinned vertex in $G$, with $s^+$
    neighbors pinned to spin $+$ and $s^-$ neighbors pinned to spin $-$, then
    $Z_{G, u}^+(\bt) = (\bt)^{s^-}, Z_{G, u}^-(\bt) = (\bt)^{s^+}$.  Item 1 of
    the induction hypothesis is thus satisfied for all small enough positive
    $\delta_0$. For item 2, we note that $R_{G,u}(\beta') = (\beta')^{s^+-s^-}$,
    and, since $0 \leq s^+, s^- \leq \Delta$, also that
    $\varphi(R_{G,u}(\beta)) \in I_0(\beta, \Delta).$ Now, let
    $\eta, \eps, \delta, M$ be the constants and $D$ the closed convex set
    (depending on $\beta$ and $\Delta$) whose existence is guaranteed
    by~Corollary\nobreakspace \ref {cor:correlation-decay-Fisher}.  For all
    small enough positive $\delta_0$, we then also have
    (i)~$\Re\inp{R_{G,u}(\beta')} > 0$ (since $\beta > (\Delta - 2)/\Delta$ and
    $\abs{\beta' - \beta} \leq \delta_0$); and (ii)~$\varphi(R_{G,u}(\beta')) \in D$.
    Combined with eq.~\eqref{eq:46} in the statement of the Corollary,
    inequality (ii) implies item~2, provided $\tau$ is chosen to be
    small enough (in terms of $\eps$ and $M$).  Item~3 follows from item~1
    and inequality~(i), since
    $\Re\inp{R_{G,u}(\beta')} = \Re\bigl({Z_{G,u}^+(\beta')/Z_{G,u}^-(\beta')}\bigr) >
    0$.

    We now proceed to the inductive step.  In this case, $G$ has at least two
    unpinned vertices.  We begin by deriving a useful consequence of the
    induction hypothesis.  Let $u$ be an arbitrary unpinned vertex in $G$, with
    $s$ pinned and $k$ unpinned neighbors.  Let the $k$ unpinned neighbors be
    $v_1, \cdots, v_k$.  We denote
    $B_i(\beta) \defeq \varphi \inp{R_{G_i,v_i}(\beta)}$,
    $\vec{B}(\beta) \defeq \set{B_1(\beta), B_2(\beta), \cdots, B_k(\beta)}$,
    and
    $H_\beta(x_1, x_2, \cdots, x_k) \defeq \Fpsk \inp{x_1, x_2, \cdots, x_k}$,
    where the graphs $G_i$ are as in Definition~\ref{def:graphs-gi}.  Note that the above quantities
    are all well defined: this is because, by construction, each $G_i$ has one
    fewer unpinned vertex than $G$, and also the degree of $v_i$ in $G_i$ is at
    most $\Delta-1$ (since $u$, which is an unpinned neighbor of $v_i$ in $G$,
    is not present in $G_i$), so that items~1 and~2 of the induction hypothesis
    apply at vertex $v_i$ in $G_i$, and also (by item~3) $Z_{G_i}(\bt) \neq 0$.
    These items also imply that $B_i(\beta), B_i(\beta') \in D$, and further that
    $\abs{B_i(\beta) - B_i(\beta')} < \tau$ (provided that $\tau \leq \eps$ and
    $\delta_0$ is small enough).  The triangle inequality then gives (again,
    assuming $\delta_0 \leq \delta$):
	\begin{align*}
      \abs{H_\beta(\vec{B}(\beta)) - H_{\beta'}(\vec{B}(\beta'))}
      &\le \abs{ H_\beta \inp{\vec{B}(\beta)} - H_\beta \inp{\vec{B}(\bt)}} + \abs{ H_\beta \inp{\vec{B}(\bt)} - H_{\bt} \inp{\vec{B}(\bt)}} \\
      &\le\sup \norm{\Grad \Fpsk }{1} \cdot  \max_{i} \abs{B_i(\beta) - B_i(\bt) }
        + M \abs{\beta - \bt},\\
      &\leq \tau \sup \norm{\Grad \Fpsk }{1}  + M \delta_0.
\end{align*}
    Here, in the second line, the first term comes from Lemma~\ref{lem:mean-value}
    (where the supremum is over all $\vec{x} \in D^k$), and the second
    term from eq.~\eqref{eq:45} of Corollary~\ref{cor:correlation-decay-Fisher}. 
    Now let $\delta_0$ be chosen so that it is also smaller than
    $\tau\cdot\min\inbr{1, \eta/2M}$. We have two cases:
    
    \noindent \textbf{Case 1: $k + |s| = \Delta$.}  In this case, we use eq.~\eqref{eq:48} 
    of Corollary~\ref{cor:correlation-decay-Fisher} to get
    \begin{equation}
      \abs{H_\beta(\vec{B}(\beta)) - H_{\beta'}(\vec{B}(\beta'))} \le M(\tau +
      \delta_0) < 2M\tau.\label{eq:49}
    \end{equation}
    \noindent \textbf{Case 2: $k + |s| \leq \Delta - 1$.}  In this case, we use
    the case $k + |s| \leq \Delta -1$ of
    Corollary~\ref{cor:correlation-decay-Fisher} to get
    \begin{equation}
      \abs{H_\beta(\vec{B}(\beta)) - H_{\beta'}(\vec{B}(\beta'))} \le
      (1-\eta/2)\tau + M\delta_0 < \tau.\label{eq:50}
    \end{equation}

    Armed with these consequences of the induction hypothesis, we now proceed to
    establish the inductive step.  Before proceeding, we note that, by
    Lemma~\ref{lem:weitz-Fisher},
    $H_{\beta}(\vec{B}(\beta)) = \varphi(R_{G,u}(\beta))$, and, when
    $Z_{G,u}^+, Z_{G,u}^- \neq 0$,
    $H_{\beta'}(\vec{B}(\beta')) = \varphi(R_{G,u}(\bt))$.

    For item 1, we consider the graph $G'$ where we pin vertex $u$ to spin $+$.
    Note that by definition, $Z_{G, u}^+ (\bt) = Z_{G'}(\bt)$.  Let $v$ be any
    unpinned vertex in $G'$.  Since $G'$ has one fewer unpinned vertex than $G$,
    by the induction hypothesis we have $|{Z_{G',v}^-(\bt)}| > 0$.  Thus,
    $R_{G', v}(\beta')$ is well defined and is in $D$.  The calculations
    leading to eqs.~\eqref{eq:49} and \eqref{eq:50} applied to the vertex $v$ in
    $G'$ imply that
    $\abs{\varphi\inp{ R_{G',v}(\beta)} - \varphi\inp{ R_{G',v}(\bt)}} < 2M\tau$
    (in fact, the upper bound improves to $\tau$ in case $v$ has degree at most
    $\Delta - 1$ in $G'$).  Applying eq.~\eqref{eq:47} of
    Corollary~\ref{cor:correlation-decay-Fisher} then shows that
    $\abs{R_{G',v}(\beta) - R_{G',v}(\bt)} < 2M^2\tau$.  Thus, since
    $R_{G',v}(\beta) > \min\inbr{\beta^\Delta, 1/\beta^\Delta}$, we have, for
    all small enough $\tau$ and $\delta_0$,
    \begin{equation}
      \Re \inp{R_{G',v}(\bt)} > 0.\label{eq:51}
    \end{equation}  We can therefore write
	\begin{align*}
\abs{Z_{G'}(\bt)} = |{Z_{G',v}^+(\bt) + Z_{G',v}^-(\bt)}|
      &= |{Z_{G',v}^-(\bt)}| \cdot \abs{1 + R_{G',v}(\bt)}.
	\end{align*}
    But then, by eq.~\eqref{eq:51},
    $\abs{1 + R_{G',v}(\bt)} \geq \Re\inp{ 1 + R_{G',v}(\bt)} > 1$.  Thus,
	\[
      |{Z_{G, u}^+ (\bt)}| = |{Z_{G'}(\bt)}| \geq |{Z_{G',v}^-(\bt)}| >0.
    \]
    An identical argument also proves that $|{Z_{G, u}^- (\bt)}| >0$,
    completing the verification of item~1 of the induction hypothesis.
    
    Now, since $Z_{G,u}^-(\bt)$ and $Z_{G,u}^+(\bt)$ have both been proved to be
    non-zero, it follows that $R_{G, u}(\bt)$ is well-defined, and, by
    Lemma~\ref{lem:weitz-Fisher}, is equal to $H_{\beta'}(\vec{B}(\beta'))$.
    Item 2 of the induction hypothesis then follows immediately from eq.~\eqref{eq:50}.
    Finally, item 3 follows from item 1 and eq.~\eqref{eq:51}.
\end{proof}

The main result of this section, establishing the absence of Fisher zeros in a complex region
around the correlation decay interval, now follows immediately from item 3 of the
previous theorem.\begin{corollary}
  Fix a degree $\Delta \ge 3$, and let $\beta \in \betarange$.  There exists a
  positive constant $\delta_0$ (depending on $\beta$ and $\Delta$) such that,
  for any graph $G$ of maximum degree $\Delta$, and any complex $\bt$ with
  $\abs{\bt - \beta} < \delta_0$, we have $Z_G(\bt) \neq 0$.
  \label{thm:main-fisher}
\end{corollary}

\subsection{Anti-ferromagnetic Ising model}
\newcommand{\lat}{\ensuremath{\lambda'}} In this section we consider the
anti-ferromagnetic Ising model.  
Recall from the introduction that, for the infinite
$\Delta$-regular tree, weak spatial mixing holds 
when $\beta \in (1, \frac{\Delta}{\Delta-2})$ for all $\lambda > 0$, while for
$\beta > \frac{\Delta}{\Delta-2}$ there exists a
$\lambda_c(\beta, \Delta) > 0$ such that weak spatial mixing holds
if $\abs{\log \lambda} > \log \lambda_c(\beta,
\Delta)$~\cite{georgii_hans-otto_gibbs_2011, sinclair_approximation_2012}.
We will refer to this as the
\emph{correlation decay region} for the anti-ferromagnetic Ising model.  Fix any
$\beta,\lambda$ in the correlation decay region. As claimed in Theorem~\ref{thm:af-ising-intro}
of the introduction, we will show that there exists
$\delta>0$ such that, for any $\lat$ with $\abs{\lat - \lambda} < \delta$, the
partition function $Z_G(\beta,\lat)\neq 0$.  
This is apparently the first result precisely
relating correlation decay to absence of Lee-Yang zeros for the antiferromagnetic Ising model
on general graphs.

As before, for a fixed vertex $v$, we write
$Z_G(\beta,\lambda) = Z_{G,v}^+(\beta,\lambda) + Z_{G,v}^-(\beta,\lambda)$ and let $R_{G,v}(\beta,\lambda)\defeq \frac{Z_G^+(\beta,\lambda)}{  Z_G^-(\beta,\lambda)}$.
Then we can write a formal recurrence relation analogous to that in 
 Lemma\nobreakspace \ref {lem:weitz-Fisher}, as follows.
\begin{lemma}
  Let $\omega_\beta,\omega_\lambda$ be formal variables. Given a graph $G$ and
  an unpinned vertex $u$, let $k$ be the number of unpinned neighbors of $u$,
  and $s$ be the number of signed pinned neighbors of $u$. Denoting
  $h_\omega(x) \defeq \frac{\omega + x}{\omega x + 1}$, we have
	\[
		R_{G,u}(\omega_\beta, \omega_\lambda) = \omega_\lambda \omega_\beta^s \prod_{i=1}^k h_{\omega_\beta}\inp{ R_{G_i, v_i}(\omega_\beta, \omega_\lambda) },
	\]
	where the graphs $G_i$ are defined as in~Definition\nobreakspace \ref {def:graphs-gi}.
	\label{lem:weitz-antiferro}
\end{lemma}

\newcommand{\Fbl}{F_{\beta,\lambda,k,s}}
\newcommand{\Fblp}{F_{\beta,\lambda,k,s}^{\varphi}}

Given integers $k$ and $s$, let
$\Fbl(\vec{x}) \defeq \lambda \beta^s \prod_{i=1}^k h_\beta(x_i)$.  This
recurrence has been studied before in the
literature~\cite{li_approximate_2012,sinclair_approximation_2012}, and as in the
case of the ferromagnetic Ising model, it has been found useful to
reparameterize $\Fbl$ with a ``potential function'' $\varphi$ as follows:
$\Fblp \defeq \varphi \circ \Fbl \circ \varphi^{-1}$. 
In~\cite{sinclair_approximation_2012} the function
$\varphi(x) \defeq \log \frac{x+D}{1-x +D}$ was used, where $D>0$ is a
constant depending on $\beta$ and $\Delta$ (but not on $\lambda$).  (This choice
of $\varphi$ is by no means unique: alternative choices can be found in,
e.g.,~\cite{li_approximate_2012,li_correlation_2011}.)  For this choice of
$\varphi$, the following step-wise correlation decay in the $1$-norm is proved
in~\cite{sinclair_approximation_2012}:

\begin{theorem}[\cite{sinclair_approximation_2012}] 
  Fix a degree $\Delta \geq 3$ and integers $k > 0, s$ such that
  $k + \abs{s} \leq \Delta-1$.  If $(\beta,\lambda)$ is in the correlation decay
  region of the infinite $\Delta$-regular tree, then there exists an $\eta > 0$
  (depending upon $\beta,\lambda$ and $\Delta$) such that
  $\norm{\Grad \Fblp (\vec{x})}{1} < 1 - \eta$ for every
  $\vec{x} \in \R^k$.\footnote{Ref.~\cite{sinclair_approximation_2012} uses a
    different convention for the Ising model, in which $\beta$ 
    corresponds to our $1/\beta$ (see eq.~(1) of~\cite{sinclair_approximation_2012}).} 
\label{lem:correlation-decay-antiferro}
\end{theorem}

We also note that an analog of the calculation leading to~eq.\nobreakspace
\textup {(\ref {eq:intervalI0})} gives the bound
$\frac{\lambda}{\beta^\Delta} \le R_{G,u}(\beta,\lambda)\le \lambda
\beta^\Delta$.  Thus we define the analogous interval
\begin{equation}
  I_0(\beta,\lambda,\Delta) \defeq \inb{\varphi\inp{{\textstyle\frac{\lambda}{\beta^\Delta}}},\varphi\inp{ \lambda \beta^\Delta}}.
\end{equation}
The following corollary is analogous to 
Corollary~\ref{cor:correlation-decay-Fisher}, and is an immediate consequence of 
Theorem~\ref{lem:correlation-decay-antiferro} and the analyticity of $\Fblp$ and
$\varphi^{-1}$ at points close to $I_0(\beta, \lambda, \Delta)$.
\begin{corollary} 
  Fix a degree $\Delta \geq 3$ and integers $k \ge 0$ and $s$.  If
  $(\beta,\lambda)$ is in the correlation decay region of the infinite $\Delta$-regular
  tree, then there exist positive constants $\eta , \eps, \delta$
  (depending upon $\beta,\lambda$ and $\Delta$) such that the following is true.
  Let $D \defeq D(\beta,\lambda,\Delta)$ be the set of points within distance $\eps$
  of $I_0(\beta,\lambda,\Delta)$ in $\C$.  Then
  $\norm{\Grad \Fblp (\vec{x})}{1} < 1 - \eta/2$ for every $\vec{x} \in D^k$
  whenever $k + \abs{s} \leq \Delta - 1$.  Moreover, there is a finite constant
  $M \geq 1$ (depending upon $\beta,\lambda$ and $\Delta$) such that
  \begin{align*}
    \sup_{\vec{x} \in D^k, \;\lat \in \C: \abs{\lat - \lambda} < \delta} \abs{
    \Fblp(\vec{x}) - F_{\beta,\lambda',k,s}^\varphi(\vec{x})}
    &\leq M \abs{\lat - \lambda};\\
    \sup_{x,y : \varphi(x), \varphi(y) \in D} \abs{\varphi(x) - \varphi(y)}
    &\leq M \abs{x - y};\\
    \sup_{x,y \in D} \abs{\varphi^{-1}(x) - \varphi^{-1}(y)}
    &\leq M \abs{x - y};\text{
      and}\\
    \sup_{\vec{x} \in D^k}
    \norm{\Grad \Fblp (\vec{x})}{1} &\leq M \text{ when $k + \abs{s} = \Delta$}.
  \end{align*}
  \label{cor:correlation-decay-antiferro}
\end{corollary}

Finally, given~Lemma\nobreakspace \ref {lem:weitz-antiferro} and\nobreakspace Corollary\nobreakspace \ref {cor:correlation-decay-antiferro}, an identical argument to that in the proof of~Theorem\nobreakspace \ref {thm:ind-hyp-fisher} establishes the following:
\begin{theorem}
  Fix a degree $\Delta \ge 3$, and let $(\beta,\lambda)$ be in the correlation
  decay region for the infinite $\Delta$-regular tree.  There exist positive
  constants $\delta_0, \tau$ (both depending on $\beta,\lambda$ and $\Delta$)
  such that, for any graph $G$ of maximum degree $\Delta$, any unpinned vertex
  $u$ in $G$, and any complex $\lat$ with $\abs{\lat - \lambda} < \delta_0$, the
  following are true:
	\begin{enumerate}
		\item $|{Z_{G, u}^+(\beta,\lat)}|> 0, |{Z_{G, u}^-(\beta,\lat)}|> 0$.
		\item
          $\abs{\varphi\inp{ R_{G,u}(\beta,\lambda)} - \varphi\inp{
              R_{G,u}(\beta,\lat)}} < \tau$ if $u$ has degree at most
          $\Delta - 1$ in $G$.
        \item $Z_G(\beta, \lat) \neq 0$.
	\end{enumerate}
\end{theorem}
The main result of this section, which is a restatement of Theorem~\ref{thm:af-ising-intro}
 in the introduction, is a direct consequence of item~3 of the above theorem:
\begin{corollary}
  Fix a degree $\Delta \ge 3$, and let $(\beta,\lambda)$ be in the correlation
  decay region for the antiferromagnetic Ising model on 
  the infinite $\Delta$-regular tree.  Then, there exists a positive
  constant $\delta_0$ (depending on $\beta,\lambda$ and $\Delta$) such that,
  for any graph $G$ of maximum degree $\Delta$, and any complex $\lat$ with
  $\abs{\lat - \lambda} < \delta_0$, we have $Z_G(\beta,\lat) \neq 0$.
  \label{thm:main-antiferro}
\end{corollary}

\newcommand{\J}{\mathcal{I}}
\subsection{Hard-core model}
\label{sec:hardcore}
In this section we consider the \emph{independence polynomial}, which is the partition function of the hard-core model.
Formally, given a graph $G=(V,E)$ and a \emph{vertex activity} $\lambda>0$, we let $\J(G)$ be the set of independent sets of vertices in $G$. 
Then the independence polynomial is given by
\[
	Z_G(\lambda) = \sum_{I \in \J(G)} \lambda^{\abs{I}}.
	\]
	The hard-core model is a simple model of the ``excluded volume'' phenomenon:
	vertices in the independent set $I$ correspond to particles, each of which prevents neighboring sites from being occupied.
	The parameter $\lambda$ controls the density of particles in the system.

It is known from seminal work of Weitz and Sly that there
is a critical activity $\lambda_c(\Delta)$ such that, when $\lambda < \lambda_c(\Delta)$, 
	the partition function for graphs of maximum degree~$\Delta$ can be
	approximated efficiently~\cite{Weitz},  while for $\lambda > \lambda_c(\Delta)$
	it is NP-hard to approximate the partition function~\cite{Sly2010CompTransition}
	(see also~\cite{Vigoda-hard-core-11,sly12}).
	We will refer to $\lambda<\lambda_c(\Delta)$ as the \emph{correlation decay interval} for the hard-core model.
	In this section, we view $Z_G(\lambda)$ as a polynomial in $\lambda$ and study its complex zeros.
	The main result of this section will again be that there are no zeros in a complex neighborhood of the correlation decay interval $(0,\lambda_c(\Delta))$.

	In similar fashion to the Ising model, for a fixed vertex $v$ we write the partition function as
	$
		Z_G(\lambda) = Z_{G\setminus v}(\lambda) + \lambda \cdot Z_{G\setminus N_G[v]} (\lambda)
		$,
	and let $R_{G,v}(\lambda)\defeq \frac{Z_{G\setminus N_G[v]} (\lambda)}{  Z_{G\setminus v} (\lambda)}$.
	Note that $Z_{G\setminus v}(\lambda)$ corresponds to pinning~$v$ to be ``unoccupied'' (not in the independent set) in~$G$, while $Z_{G\setminus N_G[v]} (\lambda)$ 
	corresponds to pinning~$v$ to be ``occupied'' (in the independent set) in~$G$.
	By analogy with~Lemmas\nobreakspace \ref {lem:weitz-Fisher} and\nobreakspace  \ref {lem:weitz-antiferro}, we have the following formal recurrence relation for $R_{G,u}$~\cite{Weitz}, which is easily verified.
\begin{lemma}
  Let $\omega$ be a formal variable. Given a graph $G$ and a vertex $u$ in $G$,
  let $k$ be the number of neighbors of~$u$. We then have
	\[
	R_{G,u}(\omega) = \lambda \prod_{i=1}^k \frac{1}{1+R_{G_i, v_i}(\omega) } ,
	\]
	where the graphs $G_i\defeq G \setminus \set{u, v_1, \cdots, v_{i-1}}$ are defined in
	analogous fashion to~Definition\nobreakspace \ref {def:graphs-gi}.
	\label{lem:weitz-hardcore}
\end{lemma}

\newcommand{\Fl}{F_{\lambda,k}}
\newcommand{\Flp}{F_{\lambda,k}^{\varphi}}

For a non-negative integer~$k$, let
$\Fl(\vec{x}) \defeq \lambda \prod_{i=1}^k \frac{1}{1+x_i}$.  This recurrence
has been studied before in the literature. As with the Ising model examples
above, it has been found useful to reparameterize $\Fl$ using a ``potential
function'' $\varphi$ in the form
$\Flp \defeq \varphi \circ \Fl \circ \varphi^{-1}$.  As shown by Li, Lu and
Yin~\cite{li_correlation_2011}, using the function
$\varphi(x) = 2 \sinh^{-1} (\sqrt{x}\,)$ leads to the following step-wise
correlation decay in the $1$-norm:\footnote{This is a special case of Lemma 4.4
  of \cite{li_correlation_2011}, taken in combination with item 5 of Lemma 3.1
  of that paper, obtained by setting $\beta = 0$ and $\gamma = 1$ in their
  notation.  Also note that in \cite{li_correlation_2011}, only the derivative
  $\Phi$ of the message $\varphi$ is explicitly mentioned (at the bottom of
  p.~76, at the end of column 1).  The function $\varphi(x) = 2\sinh^{-1}(\sqrt{x})$ is
  obtained by integrating $\Phi(x) = 1/\sqrt{x(1+x)}$.}
\begin{theorem}[\cite{li_correlation_2011}]
  Fix a degree $\Delta \geq 3$, and let $k \leq \Delta -1$ be a positive
  integer.  If $\lambda$ is in the correlation decay interval, then there exists
  an $\eta > 0$ (depending upon $\lambda$ and $\Delta$) such that
  $\norm{\Grad \Flp (\vec{x})}{1} < 1 - \eta$ for every $\vec{x} \in \R^k$.
\label{lem:correlation-decay-hardcore}
\end{theorem}

Again by analogy with the Ising model, we have the bound
$\lambda/\inp{1 + \lambda}^\Delta \le R_{G,u}(\lambda)\le \lambda$,
leading us to define the following analog of the interval in~eq.\nobreakspace
\textup {(\ref {eq:intervalI0})}:
\begin{equation}
  I_0(\lambda,\Delta) \defeq \inb{\varphi\bigl(\lambda/\inp{1 + \lambda}^\Delta\bigr),\varphi\inp{ \lambda }}.
\end{equation}
The following corollary is again a consequence of
Theorem~\ref{lem:correlation-decay-hardcore} and the smoothness properties of
$\Flp$ and $\varphi^{-1}$ at points close to the set $I_0(\lambda, \Delta)$.
\begin{corollary} 
  Fix a degree $\Delta \geq 3$ and let $k \ge 0$.  If $\lambda$ is in the
  correlation decay interval, then there exist positive constants
  $\eta , \eps, \delta$ (depending on $\lambda$ and $\Delta$) such that the
  following is true.  Let $D \defeq D(\lambda,\Delta)$ be the set of points within
  distance $\eps$ of $I_0(\lambda,\Delta)$ in $\C$.  Then, whenever
  $k \leq \Delta - 1$, $\norm{\Grad \Flp (\vec{x})}{1} < 1 - \eta/2$ for every
  $\vec{x} \in D^k$.  Moreover, there is a finite constant $M \geq 1$ (depending
  on $\lambda$ and $\Delta$) such that:
  \begin{align*}
    \sup_{\vec{x} \in D^k, \;\lambda' \in \C: \abs{\lambda' - \lambda} < \delta}
    \abs{ \Flp(\vec{x}) - F_{\lambda',k}^\varphi(\vec{x})} 
    &\leq M \abs{\lambda - \lambda'} ; \\
    \sup_{x,y : \varphi(x), \varphi(y) \in D} \abs{\varphi(x) - \varphi(y)}
    &\leq M \abs{x - y} ; \\
    \sup_{x,y \in D} \abs{\varphi^{-1}(x) - \varphi^{-1}(y)}
    &\leq M \abs{x - y}; \text{
      and}\\
    \sup_{\vec{x} \in D^k}
    \norm{\Grad \Flp (\vec{x})}{1} &\leq M \text{ when $k = \Delta$}.
  \end{align*}
\label{cor:correlation-decay-hardcore}
\end{corollary}

Finally, given~Lemma\nobreakspace \ref {lem:weitz-hardcore} and\nobreakspace Corollary\nobreakspace \ref {cor:correlation-decay-hardcore}, an identical argument to that in the proof of~Theorem\nobreakspace \ref {thm:ind-hyp-fisher} establishes the following:
\begin{theorem}
  Fix a degree $\Delta \ge 3$, and let $\lambda$ be in the correlation decay
  interval.  Then there exist positive constants $\delta_0$ and $\tau$
  (depending on $\lambda$ and $\Delta$) such that, for any graph $G$ of maximum
  degree $\Delta$, any unpinned vertex $u$ in $G$, and any $\lambda'$ with
  $\abs{\lambda' - \lambda} < \delta_0$, the following are true:
	\begin{enumerate}
		\item $\abs{Z_{G}(\lambda')}> 0$.
		\item $\abs{\varphi\inp{ R_{G,u}(\lambda')} - \varphi\inp{ R_{G,u}(\lambda)}} < \tau$.
	\end{enumerate}
\end{theorem}
The main result of this section, establishing a zero-free region containing the correlation
decay interval, now follows as an immediate corollary of the above theorem:
\begin{corollary}
	Fix a degree $\Delta  \ge 3$, and let $\lambda$ lie in the correlation decay interval for 
	the hard-core model on the infinite $\Delta$-regular tree.
	There exist positive constants $\delta, \eps$ (both depending on $\lambda$ and $\Delta$) such that,
	for any graph $G$ of maximum degree $\Delta$, and any $\lambda'$ with $\abs{\lambda' - \lambda} < \delta$, we have $Z_G(\lambda') \neq 0$.
  \label{thm:main-hardcore}
\end{corollary}
The above result was conjectured by Sokal~\cite{sokal2000personal}, and first
proved (with more detailed information about the geometry of the zero-free region)
by Peters and Regts~\cite{peters17:_sokal} via a different argument involving a tailor-made
``potential function''.  Our argument above describes a simpler route to the
result starting from the previously known correlation decay properties for real parameters.

\subsection{Related work and discussion}
There are a few recent papers which use correlation decay-like arguments for
proving absence of complex zeros: Peters and Regts~\cite{peters17:_sokal}
considered the case of the roots of the independence polynomial, while an
earlier paper by the present authors~\cite{liu2018fisher} looked at the Fiser
zeros of the zero-field Ising model.  A recent paper of Peters and
Regts~\cite{peters_location_2018} on the Lee-Yang zeros of the
anti-ferromagnetic Ising model on graphs of maximum degree at most $\Delta$ for
$\beta \in \betarange$ is also in a similar spirit.  The main message of this
section is that the somewhat different arguments used in these results can in
fact be carried out in a unified framework which allows one to ``lift'' known
analyses of Weitz recurrences for the corresponding
models~\cite{Weitz,li_correlation_2011,sinclair_approximation_2012,zhaliabai09}
to show that, in each case, there is a zero-free region of constant width that
contains the entire correlation decay interval.  Thus, as mentioned earlier,
this puts on a more formal footing the observation that Weitz's algorithm can be
seen as a bridge between the ``decay of correlations'' and ``analyticity of free
energy density'' formalisms of phase transitions.  We also note in passing that,
via Barvinok's general paradigm, the results in this section lead to polynomial
time approximation algorithms for the model partition functions in the
respective correlation decay intervals (and indeed in a complex neighborhood of
those intervals).  However, in all of the above cases, these algorithmic
consequences (at least for real-valued parameters) can be derived directly from
correlation
decay~\cite{zhaliabai09,sinclair_approximation_2012,Weitz,li_correlation_2011},
so we do not pursue this direction here.

In the following section we turn to the Potts model, where such tight
correlation decay results are not known.  We show that, with a more careful
analysis, less tight correlation decay arguments can also be lifted to the
complex plane in similar fashion to the results of this section.  Further, in contrast to the
two-spin systems considered in this section, the algorithmic consequences are also novel
and resolve open questions; indeed, it is not yet known how to obtain them directly
from correlation decay without passing to the complex plane.

\section{Potts model: Preliminaries}
\label{sec:preliminaries}
\subsection{Colorings and the Potts model}

Throughout, we assume that the graphs that we consider are augmented with a list
of colors for every vertex.  Formally, a graph is a triple $G=(V,E,L)$, where $V$ is the
vertex set, $E$ is the edge set, and $L: V\to 2^{\N}$ specifies a list of colors
for every vertex.  The partition function as defined in the introduction
generalizes naturally to this setting: the sum is now over all those colorings
$\sigma$ which satisfy $\sigma(v) \in L(v)$.

We also allow graphs to contain \emph{pinned} vertices: a vertex $v$ is said to
be \emph{pinned} \index{pinning} to a color $c$ if only those colorings of $G$
are allowed in which $v$ has color $c$.  Suppose that a vertex $v$ of degree
$d_v$ in a graph $G$ is pinned to a color $c$, and consider the graph $G'$
obtained by replacing $v$ with $d_v$ copies of itself, each of which is pinned
to $c$ and connected to exactly one of the original neighbors of $v$ in $G$.  It
is clear that $Z_{G'}(w) = Z_G(w)$ for all $w$.  We will therefore
  assume that the operation of \emph{pinning} a vertex comprises this operation
  as well; in particular, this means we can assume that all pinned
  vertices in our graphs have degree at most one.  Further, if a pinned
  vertex $u$ has another pinned vertex $v$ as a neighbor, then $u$
  and $v$ must form a connected component consisting of a single edge.
  The \emph{size} of graph~$G$ is defined to be the number of unpinned vertices.  Note that the
above operation of duplicating pinned vertices does not change the size of the
graph.

Let $G$ be a graph and $v$ an unpinned vertex in $G$.  A color~$c$ in the list
of~$v$ is said to be \emph{good}\index{good color} for $v$ if every pinned
neighbor~$u$ of~$v$ is pinned to a color different from~$c$.  The set of good
colors for a vertex $v$ in graph $G$ is denoted $\good{G, v}$.  We sometimes
omit the graph $G$ and write $\good{v}$ when $G$ is clear from the context.  A
color $c$ that is not in $\good{v}$ is called \emph{bad}\index{bad color} for
$v$.  Further, given a graph~$G$ possibly with pinned vertices, we say that the
graph is \emph{unconflicted} \index{unconflicted graphs} if no two neighboring
vertices in $G$ are pinned to the same color.  Note that since all pinned
vertices have degree exactly one, any conflicted graph is the vertex-disjoint union of an
unconflicted graph and a collection of disjoint, conflicted edges.

We will assume throughout that all unconflicted graphs $G$ we consider have at
least one proper coloring: this will be guaranteed in our applications since we
will always have $\abs{L(u)} \geq \deg_G(u) + 1$ for every unpinned vertex $u$ in $G$.

\begin{definition}For a graph $G$, a vertex $v$ and a color $i \in L(v)$, the \emph{restricted
    partition function} $\Z{G,v}{i}(w)$ is the partition function restricted to
  colorings in which vertex~$v$ receives color~$i$.  
\end{definition}

\begin{definition}Let $\omega$ be a formal variable.  For any $G$, a vertex $v$ and colors
  $i,j \in L(v)$, we define the \emph{marginal ratio} of color $i$ to color $j$
  as
  $
	\ratio{G,v}{i,j}(\omega) \defeq
    \frac{\Z{G,v}{i}(\omega)}{\Z{G,v}{j}(\omega)}.
    $
    Similarly we also define
  formally the corresponding \emph{pseudo marginal probability} as $
	\prob{G,\omega}{ c(v) = i } \defeq \frac{\Z{G,v}{i}(\omega)}{Z_G(\omega)}.
	$
\label{def:marginals}
\end{definition}
\begin{remark}\label{rem:conditions}
Note that when a numerical value $w \in \C$ is substituted in place of $\omega$
  in the above formal definition, $\ratio{G,v}{i,j}(w)$ is numerically
  well-defined as long as $\Z{G,v}{j}(w) \neq 0$, and  $\prob{G,w}{ c(v) = i }$ is numerically well-defined as long as
  $Z_G(w) \neq 0$.
  In the proof of the main theorem in
  Sections\nobreakspace \ref {sec:induction-origin} and\nobreakspace  \ref {sec:induction-interval}, we will ensure that the
  above definitions are numerically instantiated only in cases where the above
  conditions for such an instantiation to be well-defined are satisfied.  
  For instance, when $w \in [0,1]$, this is the
  case for the first definition when either (i) $w \neq 0$; or (ii) $w = 0$,
  $G$ is unconflicted and $j \in \good{G, v}$.  And for the second definition,
  this is the case when either (i) $w \neq 0$; or (ii) $w = 0$ and $G$ is
  unconflicted.
\end{remark}

\begin{remark}
	Note also that when $w \in [0,1]$, the pseudo probabilities, if well-defined, are actual marginal probabilities.  
	In this case, we will
  also write $\prob{G, w}{c(v) = i}$ as $\Pr{G,w}{ c(v) = i}$.  For arbitrary
  complex $w$, this interpretation as probabilities is of course not valid
  (since $\prob{G, w}{c(v) = i}$ can be non-real), but provided that
  $Z_G(w) \neq 0$ it is still true that
$
    \sum_{i \in L(v)} \prob{G,w}{c(v) = i} = \frac{1}{Z_G(w)}\sum_{i \in
      L(v)}\Z{G, v}{i}(w) = \frac{Z_G(w)}{Z_G(w)} = 1.
      $
We also note that if $v$ is pinned to color $k$, then $\prob{G,w}{c(v) = i}$
  is $1$ when $k = i$ and $0$ when $k \neq i$.
\end{remark}
\par\noindent
{\textbf{Notation.}} For the case $w=0$ (proper colorings) we will sometimes shorten the
notations $\prob{G, 0}{c(v) = i}$ and $\Pr{G, 0}{c(v) = i}$ to
$\prob{G}{c(v) = i}$ and $\Pr{G}{c(v) = i}$ respectively.

\begin{definition}[\textbf{The graphs $\G{k}{i,j}$}]
  \label{def:graphs-gik}
  Given a graph $G$ and a vertex $u$ in $G$, let $v_1, \cdots, v_{\deg_G(u)}$ be
  the neighbors of $u$.  We define $\G{k}{i,j}$ (the vertex $u$ will be
  understood from the context) to be the graph obtained from $G$ as follows:
  \begin{itemize}
  \item first we replace vertex $u$ with $u_1, \cdots, u_{\deg_G(u)}$, and
    connect $u_1$ to $v_1$, $u_2$ to $v_2$, and so on;
  \item next we pin vertices $u_1, \cdots, u_{k-1}$ to color $i$, and vertices
    $u_{k+1}, \cdots, u_{\deg_G(u)}$ to color $j$;
  \item finally we remove the vertex $u_k$.
  \end{itemize}
  Note that the graph $\G{k}{i,j}$ has one fewer unpinned vertex than $G$.
  Moreover, $u_1, \cdots, u_{\deg_G(u)}$ are of degree $1$, so this construction maintains the property that pinned vertices have degree $1$.
\end{definition}

We now derive a recurrence relation between the marginal ratios of the graph $G$
and pseudo marginal probabilities of the graphs $\G{k}{i,j}$.  This is an
extension to the Potts model of a similar recurrence relation derived by
Gamarnik, Katz and Misra~\cite{Gamarnik2015SSM} for the special case of colorings (that is, $w=0$).
\begin{lemma}
  Let $\omega$ be a formal variable.  For a graph $G$, a vertex $u$ and colors
  $i,j \in L(u)$, we have
  \[
    \ratio{G,u}{i,j}(\omega) = \prod_{k=1}^{\deg_G(u)}\frac{ 1 - \gamma \cdot
        \prob{\G{k}{i,j},\omega}{ c(v_k) = i} }{1 -
        \gamma \cdot \prob{\G{k}{i,j},\omega}{ c(v_k) = j}  },
  \]
  where we define $\gamma\defeq 1-\omega$.  In particular, when a numerical
  value $w \in \C$ is substituted in place of $\omega$, the above recurrence is
  valid as long as the quantities $Z_{\G{k}{i,j}}(w)$ and
  $1 - \gamma \cdot \prob{\G{k}{i,j},w}{ c(v_k) = j}$ for
  $1 \leq k \leq \deg_G(u)$ are all non-zero.
  \label{thm:recurrence}
\end{lemma}
\begin{proof}
	Let $t \defeq \deg_G(u)$. 
For $0 \leq k \leq t$, let $H_k$ be the graph obtained from $G$ as
  follows:
  \begin{itemize}
  \item first we replace vertex $u$ with $u_1, \cdots, u_{t}$, and
    connect $u_1$ to $v_1$, $u_2$ to $v_2$, and so on;
  \item we then pin vertices $u_1, \cdots, u_{k}$ to color $i$, and vertices
    $u_{k+1}, \cdots, u_{t}$ to color $j$.
  \end{itemize}
  Note that $H_k$ is the same as $\G{k}{i,j}$, except that the last step of the
  construction of $\G{k}{i,j}$ is skipped, i.e, the vertex $u_k$ is not removed,
  and, further, $u_k$ is pinned to color $i$.  
We can now write
	\begin{align*}
	\ratio{G,u}{i,j}(\omega)
	=\frac{\Z{G,u}{i}(\omega)}{\Z{G,u}{j}(\omega)}
	=\frac{Z_{H_{t}} (\omega)}{Z_{H_0}(\omega)}
	=\prod_{k=1}^{t} \frac{Z_{H_k}(\omega)}{Z_{H_{k-1}}(\omega)} .
	\end{align*}
	Next, for $1 \leq k \leq t$, let $\Yk \defeq Z_{\G{k}{i,j}}(\omega)$
    and $\Yki \defeq \Z{\G{k}{i,j}, v_k}{i} (\omega)$.  We observe that
	\begin{align*}
	\prob{\G{k}{i,j},\omega}{ c(v_k) = i} &= \frac{\Yki}{\Yk};\\
		Z_{H_k}(\omega) &= \Yk - (1-\omega) \cdot \Yki;\\
		Z_{H_{k-1}}(\omega)&= \Yk - (1-\omega) \cdot \Ykj.
	\end{align*}
	Therefore we have
	\begin{align*}
      \ratio{G,u}{i,j}(\omega) 
      =\maybealign \prod_{k=1}^{t} \frac{\Yk - (1-\omega) \cdot \Yki}{\Yk - (1-\omega) \cdot \Ykj} \eqbreak
	=\maybealign\prod_{k=1}^{t}\frac{ 1 - \gamma \cdot \prob{\G{k}{i,j},\omega}{ c(v_k) = i} 
      }{ 1 - \gamma \cdot \prob{\G{k}{i,j},\omega}{ c(v_k) = j}  },
	\end{align*}
	where $\gamma = 1-\omega$.  The claim about the validity of the recurrence
    on numerical substitution then follows from the conditions outlined in Remark~\ref{rem:conditions}.
\end{proof}

\subsection{Complex analysis}

In this subsection we collect some tools and observations from complex analysis.
Throughout this paper, we use $\im$ to denote the imaginary unit $\sqrt{-1}$, in
order to avoid confusion with the symbol ``$i$'' used for other purposes.  For a
complex number $z = a + \im b$ with $a, b \in \R$, we denote its real part $a$
as $\Re z$, its imaginary part $b$ as $\Im z$, its \emph{length}
$\sqrt{a^2 + b^2}$ as $\abs{z}$, and, when $z \neq 0$, its \emph{argument}
$\sin^{-1} ({\frac{b}{\abs{z}}}) \in (-\pi, \pi]$ as $\arg z$.  We also
generalize the notation $[x, y]$ used for closed real intervals to the case when
$x, y \in \C$, and use it to denote the closed straight line segment joining $x$
and $y$.

We start with a consequence of the mean value theorem for complex functions,
specifically tailored to our application.
Let $D$ be any domain in $\C$ with the following properties.
  \begin{itemize}
  \item For any $z \in D$, $\Re z \in D$.
  \item For any $z_1, z_2\in D$, there exists a point $z_0 \in D$ such that one
    of the numbers $z_1 - z_0, z_2 - z_0$ has zero real part while the other has
    zero imaginary part.
  \item If $z_1, z_2 \in D$ are such that either $\Im z_1 = \Im z_2$ or
    $\Re z_1 = \Re z_2$, then the segment $[z_1, z_2]$ lies in $D$.
  \end{itemize}
  We remark that a rectangular region symmetric about the real axis will satisfy all of the above properties.
\begin{lemma}[\textbf{Mean value theorem for complex functions}]Let $f$
  be a holomorphic function on a domain~$D$ as above such that, for $z\in D$, $\Im f(z)$ has the
  same sign as $\Im z$.  Suppose further that there exist positive constants $\ifactor$
  and $\rfactor$ such that
  \begin{itemize}
  \item for all $z \in D$,  $\abs{\Im f'(z)} \le \ifactor$;
  \item for all $z \in D$, $\Re f'(z) \in[0, \rfactor]$. 
\end{itemize}
  Then for any $z_1, z_2 \in D$, there exists $C_{z_1, z_2} \in [0, \rfactor]$ such that
  \begin{align*}
	  \abs{\Re\inp{f(z_1) - f(z_2)} - C_{z_1, z_2} \cdot \right. \maybealign \left. \Re \inp{z_1 - z_2}} \eqbreak
    \maybealign\leq  \ifactor \cdot |\Im\inp{z_1 - z_2}|,
  \end{align*}
  and furthermore,
  \begin{align*}
    \maybealign \abs{\Im\inp{f(z_1) - f(z_2)}} \eqbreak
    \maybealign\leq \rfactor \cdot
    \begin{cases}
      |\Im (z_1 - z_2)|,  \text{ when $(\Im z_1)\cdot(\Im z_2) \leq 0$;}\\
      \max\inbr{\abs{\Im z_1}, \abs{\Im {z_2}}},  \text{ otherwise.}
    \end{cases}
\end{align*}
  \label{thm:signed-mean}
\end{lemma}
\begin{proof}
  We write $f = u + \im v$, where $u, v : D \rightarrow \R$ are
  seen as differentiable functions from $\R^2$ to $\R$ satisfying the
  Cauchy-Riemann equations
  \begin{displaymath}
    u^{(1, 0)} = v^{(0,1)} \quad \text{ and } \quad u^{(0, 1)} = -v^{(1,0)}.
  \end{displaymath}
  This implies in particular that $\Re f'(z) = u^{(1, 0)}(z) = v^{(0,1)}(z)$ and
  $\Im f'(z) = v^{(1,0)}(z) = -u^{(0,1)}(z)$.

  Let $z_0$ be a point in $D$ such that $\Re (z_2 - z_0) = 0$ and
  $\Im (z_1 - z_0) = 0$ (by the conditions imposed on $D$, such a $z_0$ exists,
  possibly after interchanging $z_1$ and~$z_2$). Now we have
  \begin{align*}
    \maybealign \Re \inp{f(z_1) - f(z_2)} \eqbreak
    &
      = u(z_1) - u(z_0) + u(z_0) - u(z_2)\\
    &
      = u^{(1, 0)}(z') \cdot \Re(z_1 - z_0) + u(z_0) - u(z_2),
  \end{align*}
  where $z'$ is a point lying on the segment $[z_0, z_1]$, obtained by applying
  the standard mean value theorem to the function $u$ along this segment (note
  that the segment is parallel to the real axis).  On the other hand, since the
  segment $[z_0, z_2]$ is parallel to the imaginary axis, we may apply the standard
  mean value theorem to the real valued function $u$ to get (after recalling
  that $\abs{u^{(0,1)}(z)} = \abs{\Im f'(z)} \leq \rho_I$ for all $z \in D$)
  \begin{displaymath}
    \abs{u(z_0) - u(z_2)} \leq \ifactor \abs{\Im\inp{z_2 - z_0}} = \ifactor
    \abs{\Im (z_2 - z_1)}.
  \end{displaymath}
  This proves the first part, once we set
  $C_{z_1, z_2} = u^{(1, 0)}(z') = \Re f'(z')$, which must lie in
  $[0, \rfactor]$ since $z' \in D$.

  For the second part, we note that since $\Im f(z) = 0$ when $\Im z = 0$, we
  have for $z \in D$,
  \begin{align*}
    \Im f(z) = \Im \inp{f(z) - f(\Re z)}
    &= v(z) - v(\Re z) \\
    &= v^{(0, 1)}(z') \cdot \Im{z} ,\end{align*}
  where $z'$ is a point lying on the segment $[z, \Re z]$, obtained by applying
  the standard mean value theorem to the function $v$ along this segment (note that the
  segment is parallel to the imaginary axis). 

  Since $v^{(0, 1)}(z') =  u^{(1, 0)}(z')\in [0, \rfactor]$ for all $z' \in D$,
  there therefore exist
  $a, b \in [0, \rfactor]$ such that
  \begin{displaymath}
    \abs{\Im(f(z_1) - f(z_2))} = \abs{a\Im z_1 - b \Im z_2},
  \end{displaymath}
  so that we get
  \begin{align*}
  \abs{\Im \right. \maybealign \left. (f(z_1) - f(z_2))}
    = \abs{a\Im z_1 - b \Im z_2} \eqbreak
    \maybealign\leq \rfactor \cdot
    \begin{cases}
      |\Im (z_1 - z_2)|,  \text{ when $(\Im z_1)\cdot(\Im z_2) \leq 0$;}\\
      \max\inbr{\abs{\Im z_1}, \abs{\Im {z_2}}},  \text{ otherwise.}
    \end{cases}
  \end{align*}
This completes the proof.  
\end{proof}

Later, we will apply the above lemma to the function
\begin{equation}
f_\kappa(x) \defeq -\ln(1 - \kappa e^x),\label{eq:28}
\end{equation}
which will play a central role in our proofs.  (We note
that here, and also later in the paper, we use $\ln$ to denote the principal
branch of the complex logarithm; i.e., if $z = re^{\iota \theta}$ with $r > 0$
and $\theta \in (-\pi, \pi)$, then $\ln z = \ln r + \iota \theta$.)  In the following
lemma, we verify that, for real $\kappa\in [0,1]$, $f_\kappa$ indeed satisfies
the hypotheses of  Lemma~\ref{thm:signed-mean} so that such an application
is valid, and also quantify the deviation in~$f_\kappa(z)$ for complex~$z$
close to the real interval.
\begin{lemma}\label{obv:f-props-int}
  Consider the domain $D$ given by
  \begin{displaymath}
    D \defeq \inbr{z \st \Re z \in (-\infty, -\zeta)\text{ and } |\Im z| < \tau},
  \end{displaymath}
  where $\tau < 1/2$ and $\zeta$ are positive real numbers such that
  $\tau^2 + e^{-\zeta} < 1$.  Suppose $\kappa \in [0,1]$ and consider the
  function $f_\kappa$ defined in eq.\nobreakspace \textup {(\ref {eq:28})}.  Then:
  \begin{enumerate}
  \item The function $f_\kappa$ and the domain $D$ satisfy the hypotheses of
    \MakeUppercase Lemma\nobreakspace \ref {thm:signed-mean}, if $\rho_R$ and $\rho_I$ in the statement of the
    theorem are taken to be $\frac{e^{-\zeta}}{1 - e^{-\zeta}}$ and
    $\frac{\tau\cdot e^{-\zeta}}{\inp{1 - e^{-\zeta}}^2}$, respectively.
  \item If $\eps > 0$ and $\kappa'\in\C$ are such that
    $\abs{\kappa' - \kappa} < \eps$ and $(1+\eps) < e^\zeta$, then for any
    $z \in D$,
    \begin{displaymath}
      \abs{f_{\kappa'}(z) - f_\kappa(z)} \leq \frac{\eps}{e^\zeta - 1 - \eps}.
    \end{displaymath}
  \end{enumerate}
\end{lemma}
\begin{proof}
  Note first that the domain $D$ is rectangular and symmetric about the real axis, so it satisfies
  the properties listed before \MakeUppercase Lemma\nobreakspace \ref {thm:signed-mean}.
We also note that since $\kappa \leq 1$, $f_{\kappa}(z)$ is well defined when
  $\Re z < 0$, and maps real numbers in $D$ to real numbers.  Further, a direct
  calculation shows that $\Im f_\kappa(z) = - \arg (1 - \kappa e^z)$ has the
  same sign as $\sin (\Im z)$ when $\Re z < 0$ (since $\kappa \in [0,1]$).
  Since $\abs{\Im z} \leq \tau < \pi$, we see therefore that $\Im f_\kappa(z)$
  has the same sign as $\Im z$, and hence $f_\kappa$ satisfies the first hypothesis of
  \MakeUppercase Lemma\nobreakspace \ref {thm:signed-mean}.

  Note that $f_\kappa'(z) = \frac{\kappa e^z}{1 - \kappa e^z}$.  A direct
  calculation then shows that
  $\Re f_\kappa'(z) = \frac{\kappa\Re e^z - \kappa^2\abs{e^z}^2}{\abs{1 -
      \kappa e^z}^2}$ and
  $\Im f_\kappa'(z) = \frac{\kappa \Im e^z}{\abs{1 - \kappa e^z}^2}$.  Now,
  for $z \in D$, $\abs{\arg e^z} \leq \tau$, so that
  $\Re e^z \geq \abs{e^z}\cos \arg e^z \geq \abs{e^z} (1-\tau^2)$.  Thus, we see
  that
  $\kappa\Re e^z - \kappa^2\abs{e^z}^2 \geq \kappa\abs{e^z}\inp{1 - \tau^2 -
    \kappa\abs{e^z}} \geq \kappa\abs{e^z}\inp{1 - \tau^2 -
    \kappa e^{-\zeta}}$.  Since $\kappa \in [0, 1]$ and
  $\tau^2 + e^{-\zeta} < 1$ by assumption, we therefore have
  $\Re f_\kappa'(z) \geq 0$.  Further,
  $\Re f_\kappa'(z) \leq \abs{f_\kappa'(z)} =
  \frac{\kappa\abs{e^z}}{\abs{1 - \kappa e^z}} \leq \frac{\kappa \abs{e^z}}{1
    - \kappa \abs{e^z}} \leq \frac{\kappa e^{-\zeta}}{1-e^{-\zeta}}$, since
  $\kappa \in [0, 1]$.  Together, these show that
  $\Re f_\kappa'(z) \in \inb{0, \frac{e^{-\zeta}}{1- e^{-\zeta}}}$ for $z \in D$, so that the
  claimed choice of the parameter $\rho_R$ in Lemma\nobreakspace \ref {thm:signed-mean} is
  justified.

  Similarly, for the imaginary part, we have
  $\abs{\Im f_\kappa'(z)} = \frac{\kappa \abs{\Im e^z}}{\abs{1 -\kappa
      e^z}^2}$, which in turn is at most
  $\frac{\kappa\cdot\tau \cdot e^{-\zeta}}{(1-\kappa e^{-\zeta})^2}$ for
  $z \in D$.  Since $\kappa \in [0, 1]$, this justifies the choice of the
  parameter $\rho_I$ and concludes the verification of item~1.

  We now turn to item~2.  The derivative of $f_x(z)$
  with respect to $x$ is $\frac{e^z}{1-x e^z}$, which for $x$ within distance
  $\eps$ (satisfying $(1+\eps) < e^\zeta$) of $\kappa$ and $z \in D$ has length
  at most $\frac{1}{e^\zeta - 1- \eps}$.  Thus, the standard mean value theorem
  applied along the segment $[\kappa, \kappa']$ (which is of length at most
  $\eps$) yields the claim.
\end{proof}

We will also need the following simple geometric lemma, versions of which have
been used in the work of Barvinok~\cite{barvinok2017combinatorics} and also Bencs et al.~\cite{bencs2018zero}. \begin{lemma}Let $z_1, z_2, \dots, z_n$ be complex numbers such that the angle between any
  two non-zero $z_i$ is at most $\alpha \in [0, \pi/2)$.  Then
  $ \abs{\sum_{i=1}^n z_i} \geq \cos(\alpha/2)\sum_{i=1}^n \abs{z_i}$.
  \label{lem:geometric-proj}
\end{lemma}

\begin{proof}
  Fix a non-zero $z_i$, and without loss of generality let $z_1$ and $z_2$ be
  the non-zero elements giving the maximum and minimum values, respectively, of
  the quantity $\arg(z_j/z_i)$, as $z_j$ varies over all the non-zero elements
  (breaking ties arbitrarily).  Consider the ray $z$ bisecting the angle between
  $z_1$ and $z_2$.  Then, by the assumption, the angle made by $z$ and any of
  the non-zero $z_i$ is at most $\alpha/2$, so that the projection of $z_i$ on
  $z$ is of length at least $\abs{z_i}\cos(\alpha/2)$ and is in the same
  direction as $z$.  Thus, denoting by $S'$ the projection of
  $S = \sum_{i=1}^n z_i$ on $z$, we have
  \begin{displaymath}
    |S| \geq |S'| \geq \sum_{i=1}^n\abs{z_i}\cos(\alpha/2).\qedhere
  \end{displaymath}
\end{proof}
\subsection{Sketch of the algorithm}
\label{sec:app_algorithm}
In this subsection we outline how to apply Barvinok's algorithmic paradigm to 
translate our zero-freeness result (Theorem~\ref{thm:zeros}) into the FPTAS
claimed in Theorem~\ref{thm:mainPotts}.  Let $G$ be a graph with $n$ vertices
and $m$ edges and maximum degree~$\Delta$.  Recall that our goal is to obtain a 
$1\pm\varepsilon$ approximation of the Potts model partition function $Z_G(w)$
at any point $w\in[0,1]$.  
Note that $Z_G$ is a polynomial of degree~$m$, and that computing
$Z_G$ at $w=1$ is trivial since $Z_G(1)=q^n$.  
Recall also that Theorem~\ref{thm:zeros} ensures that $Z_G$ has no zeros in 
the region~$\mathcal{D}_\Delta$ of width $\tau_\Delta$ around the real
interval $[0,1]$.  For technical convenience we will actually work with a slightly
smaller zero-free region consisting of the rectangle $$
  \mathcal{D}'_\Delta = \{w\in\mathbb{C}: -\tau'_\Delta \le \Re w \le 1+\tau'_\Delta;\, |\Im w| \le \tau'_\Delta\},$$
where $\tau'_\Delta  = \tau_\Delta / \sqrt{2}$.  Note that
$\mathcal{D}'_\Delta \subset  \mathcal{D}_\Delta$ so $\mathcal{D}'_\Delta$ is also zero-free.
In the rest of this section, we drop the subscript~$\Delta$ from these quantities.

Now let $f(z)$ be a complex polynomial of degree~$d$ for which $f(0)$ is easy to evaluate,
and suppose we wish to approximate~$f(1)$.  Barvinok's basic 
paradigm~\cite[Section 2.2]{barvinok2017combinatorics}
achieves this under the assumption that $f$ has no zeros in the open disk $\mathcal{B}(0,1+\delta)$
of radius $1+\delta$ centered at~$0$: the approximation simply consists of the first 
$k=O(\frac{1}{\delta}\log(\frac{d}{\varepsilon\delta}))$ 
terms of the Taylor expansion of $\log f$ around~0.  (Note that this expansion is absolutely convergent
within $\mathcal{B}(0,1+\delta)$ by the zero-freeness of~$f$.)  These terms can in turn be
expressed as linear combinations of the first $k$ coefficients of~$f$ itself. 
We now sketch how to reduce our computation of $Z_G(w)$ to this situation.  

First, for any fixed $w\in [0,1]$,
define the polynomial $g(z) := Z_G(z(w-1) + 1)$.  Note that $g(0) = Z_G(1)$ is trivial,
while $g(1) = Z_G(w)$ is the value we are trying to compute.  Moreover, plainly $g(z)\ne 0$
for all $z\in\mathcal{D}'$.  Next, define a polynomial~$\phi:\mathbb{C}\to\mathbb{C}$
that maps the disk $\mathcal{B}(0,1+\delta)$ into the rectangle $\mathcal{D}'$, so 
that $\phi(0) = 0$ and $\phi(1) = 1$; Barvinok~\cite[Lemma 2.2.3]{barvinok2017combinatorics} gives an explicit
construction of such a polynomial, with degree $N = \exp(\Theta(\tau^{-1}))$
and with $\delta=\exp(-\Theta(\tau^{-1}))$.   Now we have reduced
the computation of $Z_G(w)$ to that of $f(1)$, where $f(z):=g(\phi(z))$ is a polynomial of degree
$\deg(g)\cdot\deg(\phi) = mN$ that is non-zero on the disk $\mathcal{B}(0,1+\delta)$, 
so the framework of the previous paragraph applies.
Note that the number of terms required in the Taylor expansion of $\log f$ is 
$k=O(\frac{1}{\delta}\log(\frac{mN}{\varepsilon\delta})) = 
\exp(O(\tau^{-1}))\log(\frac{n\Delta}{\varepsilon}) $.

Naive computation of these $k$ terms requires time $n^{\Theta(k)}$, which yields only
a quasi-polynomial algorithm since $k$ contains a factor of $\log n$.  This complexity
comes from the need to enumerate all colorings of subgraphs induced by up to~$k$
edges.  However, a technique of Patel and Regts~\cite{patel2017deterministic},
based on Newton's identities and an observation of Csikvari and Frenkel~\cite{csikvari2016benjamini},
can be used to reduce this computation to an enumeration over subgraphs induced
by {\it connected\/} sets of edges (see \cite[Section~6]{patel2017deterministic} for details).
Since $G$ has bounded degree, this reduces the complexity to
$\Delta^{O(k)} = (\frac{n\Delta}{\varepsilon})^{\log(\Delta)\exp(O(\tau^{-1}))}$.
For any fixed~$\Delta$ this is polynomial in~$(n/\varepsilon)$, thus satisfying
the requirement of a FPTAS.  

Note that the degree of the polynomial is 
exponential in~$\tau^{-1}$; since $\tau^{-1}$ in turn is exponential in~$\Delta$
(see the discussion following the proof of Theorem~\ref{thm:zeros}), the degree of the polynomial
is doubly exponential in~$\Delta$.  The same discussion explains how this can be
improved to singly exponential for the case of uniformly large list sizes.

\newcommand{\impted}{marked}
\section{Properties of the real-valued recurrence}
\label{sec:prop-real-valu}
In this section we prove some basic properties of the real-valued recurrence in 
Lemma\nobreakspace \ref {thm:recurrence},
that is, in the case where $w \in [0,1]$ is real (and hence $\gamma = 1 - w \in[0,1]$).

We remark that in all graphs $G$ appearing in our analysis, we will be able to
assume that for any unpinned vertex $u$ in $G$, $|L(u)| \geq \deg_G(u) + 1$.
Thus, $Z_G(w) \neq 0$ whenever either (i) $w \in (0, 1]$; or (ii) $w = 0$ and
$G$ is unconflicted.  As discussed in the previous section, this implies that the
marginal ratios and the pseudo marginal probabilities are
well-defined, and, further, the latter are actual probabilities.  Note also that,
$G$ is not connected, and $G'$ is the connected component containing $u$, then
we have $\ratio{G,u}{i,j}(w) = \ratio{G',u}{i,j}(w)$
and $\prob{G,w}{c(u)=i} = \prob{G',w}{c(u)=i}$.

As noted in the introduction, we will prove our main theorem about zero-freeness
under a certain abstract condition on list-coloring instances which we call
{\it admissibility}.  In this section, we define admissibililty and then show that
all three classes of instances referred to in
Theorems~\ref{thm:zeros-intro} and \ref {thm:zeros-intro-176}, and Proposition~\ref{thm:zeros-trees}
are admissible.  The last two sections of the paper will be devoted to proving
zero-freeness for all admissible instances.

To define admissibility, we augment our list-coloring instances by {\it marking\/}
certain unpinned vertices; we call the resulting instances {\it marked\/} instances.

The first key property of admissible instances is that they are ``hereditary", in the
following sense.
\begin{definition}[\textbf{Hereditary}]
  A condition on \impted{} list-coloring instances is \emph{hereditary} if it is
  preserved under each of the following operations:
  \begin{enumerate}
    \item Remove a pinned vertex from the graph, without changing the
        set of marked vertices.
    \item Pin a marked vertex $u$ to any color in its list $L(u)$, and mark (if
      they are not already marked) all
    unpinned neighbors of $u$.  (Note that $u$ itself is no longer marked since it is now a
        pinned vertex, while all other marked vertices, if any, remain marked.)
\item Given a graph $G$ and a marked unpinned vertex $u$ in $G$,
      take any neighbor $v_k$ of $u$ and colors $i,j$ in the list
        $L(u)$, and construct the instance $\G{k}{i,j}$ as in Definition~\ref{def:graphs-gik},
        with the set of marked vertices in $\G{k}{i,j}$ consisting of all unpinned
        neighbors of $u$ in $G$ and all other marked vertices $v'$ in $G$.
    \item Take a connected component $H$ of $G$, with the set of marked
        vertices in $H$ being those vertices of $H$ that were marked in~$G$.
	\end{enumerate}
	\label{def:hereditary}
  \end{definition}
  
The second key property of admissible instances is that the marginal distributions
of colors on certain vertices have ``large'' min-entropy. This ``niceness" property is spelled out 
in the following definition.
We emphasize that establishing niceness is the only place in our analysis where
the lower bounds on the list sizes are used.

\begin{definition}[\textbf{Niceness}]
  Given a graph $G$ and an unpinned vertex $u$ in $G$, let $d$ be the number of
  unpinned neighbors of $u$.  We say the vertex $u$ is \emph{\nice{}} in $G$ if
  for any real $w \in [0,1]$ and any color $i \in L(u)$,
  $ \Pr{G,w}{c(u) = i} \le \frac{1}{d + 2}.  $
	\label{cond:nice}
\end{definition}

We are now in a position to define admissible instances, as previously advertised.

\begin{definition}[\textbf{Admissibility}]
	A condition $\mathcal{L}$ on marked list-coloring instances is an 
	\emph{admissible} list condition if it satisfies all of the following properties:
	\begin{enumerate}
    \item[(i)] $\mathcal{L}$ is hereditary;
    \item[(ii)] if a list-coloring instance $G$ satisfies $\mathcal{L}$, then
        for every unpinned vertex $u$ in $G$, $\abs{L(u)} \geq \deg_G(u) + 1$;
    \item[(iii)] if a list-coloring instance $G$ satisfies $\mathcal{L}$, and $G$ has
      at least one unpinned vertex, then $G$ also has at least one marked unpinned
      vertex;
    \item[(iv)] if a list-coloring instance $G$ satisfies $\mathcal{L}$,
          then for any marked unpinned vertex $u$ in $G$, and
            any unpinned neighbor $v_k$ of $u$, $v_k$ is nice in $\G{k}{i,j}$.
	\end{enumerate}
	\label{def:good-list-condition}
\end{definition}

We now recall the three conditions on coloring instances from the introduction,
appropriately generalized to include list-coloring and marking.

\begin{condition}\label{asm:2G}
$\abs{L(v)} \ge \max\inbr{2, 2\cdot \deg_G(v)}$ for every
    unpinned vertex $v$ in $G$, and all unpinned vertices are marked.
\end{condition}

\begin{condition}\label{asm:176} 
The graph $G$ is triangle-free and further, for every unpinned vertex $v$ of $G$,
	\[
		\abs{L(v)} \ge \alpha \cdot \deg_G(v) + \beta,
	\]
   where $\alpha$ is any fixed constant larger than the unique positive
    solution $\alpha^\star$ of the equation $x e^{-\frac{1}{x}} = 1$ and
    $\beta = \beta(\alpha) \ge 2\alpha$ is a constant chosen so that
    $ \alpha \cdot e^{- \frac{1}{\alpha} (1+\frac{1}{\beta})} \ge 1 $.
    We note that $\alpha^\star$ lies in the interval $[1.763, 1.764]$, and
    $\beta$ as chosen above is at least $7/2$.  Further, all unpinned vertices
    are marked.
\end{condition}

\begin{condition}\label{asm:tree} 
$G$ is a forest of maximum degree $\Delta$, with
  the {\it same\/} list of $q \ge \Delta + 1$ colors for every unpinned vertex.
  Further, each connected component of the forest that does not consist
    entirely of pinned vertices has exactly one marked unpinned vertex,
    and all unpinned vertices with pinned vertices as neighbors are marked.
\end{condition}

\begin{remark}
  Note that the condition $\abs{L(v)} \geq 2$ imposed in Condition~\ref{asm:2G} above is
  without loss of generality, since any vertex with $\abs{L(v)} = 1$ can be
  removed from $G$ after removing the unique color in its list from the lists of
  its neighbors without changing the number of colorings of~$G$.
\end{remark}

\begin{remark}
  Condition~\ref{asm:176} is essentially identical to Assumption 1 of Gamarnik,
  Katz, and Misra~\cite{Gamarnik2015SSM}.  Indeed, an important technical
  calculation for us, which appears in Lemma~\ref{lem:real-estimate176}, is
  essentially identical to a similar calculation in~\cite{Gamarnik2015SSM}.  The
  differences between Condition~\ref{asm:176} and Assumption~1 
  of~\cite{Gamarnik2015SSM} are of a technical nature, and are driven by the form
  of the upper bound we require in Lemma~\ref{lem:real-estimate176}.  In
  particular, Assumption~1 of~\cite{Gamarnik2015SSM} puts a somewhat weaker
  restriction on $\beta$ ($\beta \geq 2 + \sqrt{2}$), but then requires the
  stronger condition
  $(1- 1/\beta)\cdot\alpha\cdot\exp(-1/\alpha \cdot ( 1 + 1/\beta) > 1$ on
  $\alpha$ and $\beta$ together.
\end{remark}

Our goal in the remainder of this section is to prove that all three of the above
conditions are admissible.

\begin{lemma}
Conditions~\ref{asm:2G}, \ref{asm:176} and~\ref{asm:tree} above are all admissible.
  \label{obv:good-to-nice}
\end{lemma}

To prove this lemma, we first verify the easy fact that all three list conditions are hereditary.

\begin{proposition}\label{prop:hereditary}
Conditions~\ref{asm:2G}, \ref{asm:176} and~\ref{asm:tree} above are all hereditary.
\end{proposition}
\begin{proof}
  Recall that hereditary conditions must be preserved under the four operations
  listed in Definition~\ref{def:hereditary}.

  For the first operation, observe that removing any number of pinned vertices
  does not increase the degree or change the lists at any unpinned vertices.  
  Further, if the graph is triangle-free, it remains so after such a removal.  
  Finally, this operation does not change which vertices are marked. Hence the
  first operation preserves all three conditions.
 
  For the second operation, we note that pinning a vertex does not change the degree
  or the list at any unpinned vertex.  Further, if the graph is
  either triangle-free or a tree, it remains so after the operation of pinning a
  vertex.  This already establishes that the second operation
  preserves Conditions~\ref{asm:2G} and~\ref{asm:176}, as all unpinned vertices
  remain marked.
  For Condition~\ref{asm:tree}, we note that on pinning a marked vertex~$u$ in the
  forest, the component in which $u$ lies breaks into connected
  components (trees) indexed by the neighbors of $u$, none of which are marked
  in $G$ (since, by the hypothesis, $u$ is the unique marked vertex
  in its connected component).  Further, the components indexed by the pinned
  neighbors of $u$ are just single edges with both endpoints pinned, while those indexed
  by an unpinned neighbor $v$ of $u$ get $v$ as their unique marked vertex.
  Thus, Condition~\ref{asm:tree} is also preserved under the second operation.

  We now turn to the third operation.  Again, as in the second operation, none
  of the lists at the unpinned vertices change, while the degree of $v_k$ drops
  by one.  As all unpinned vertices remain marked, this already establishes that
  this operation preserves Conditions~\ref{asm:2G} and~\ref{asm:176}.  For the
  case when $G$ is a forest (Condition~\ref{asm:tree}), we note that
  in $\G{k}{i,j}$, the component of $G$ containing~$u$ breaks into
  connected components (trees) indexed by the neighbors of $u$ in $G$.
  Further, since only $u$ was marked in its connected component in $G$,
    and only the unpinned neighbors of $u$ get marked in the new connected
    components created in $\G{k}{i,j}$, the condition that each connected
    component not consisting entirely of pinned vertices must have exactly one
    marked vertex is satisfied.  Finally, we observe that the only pinned
    vertices in the newly created connected components in $\G{k}{i,j}$ must
    correspond to either (i) pinned neighbors of $u$ in $G$; or (ii) pinned
    copies of $u$ that are now neighbors of (marked) vertices that were the
    unpinned neighbors of $u$ in $G$.  All these pinned vertices have either a
    pinned vertex or a marked vertex as their (unique) neighbor.  This
    establishes that Condition~\ref{asm:tree} is also preserved by the third
    operation.
  
  Finally, the fourth operation of passing to a connected
  component trivially maintains all three conditions.
\end{proof}

Continuing with our proof of Lemma~\ref{obv:good-to-nice},
we note next that property~(ii) is trivially true for all three of 
Conditions~\ref{asm:2G}, \ref{asm:176} and~\ref{asm:tree},
while property~(iii) is also easily verified in all three cases.
To conclude the proof, it therefore remains only to prove the niceness property~(iv).
We do this separately for each of the three Conditions in the following
subsections.

\begin{remark} In the remainder of this section, we adopt the convention that if
  $G$ is a conflicted graph (so that it has no proper colorings) and $w = 0$,
  then $\Pr{G,w}{c(u) = i} = 0$ for every color $i$ and every unpinned vertex
  $u$ in $G$. This is just to simplify the presentation in this section by
  avoiding the need to explicitly exclude this case from the lemmas below.  In
  the proof of our main result in
  Sections\nobreakspace \ref {sec:induction-origin} and\nobreakspace  \ref {sec:induction-interval}, we will never consider
  conflicted graphs in a situation where $w$ could be $0$, so that this
  convention will then be rendered moot.
  \label{rem:convention}
\end{remark}

\subsection{Analysis for Condition~\ref {asm:2G}}

\begin{lemma}
  Let $G$ be a graph that satisfies Condition~\ref {asm:2G}.
  Then for any unpinned vertex $u$ in $G$, and any unpinned neighbor $v_k$ of
  $u$, we have that $v_k$ is \nice{} in $\G{k}{i,j}$.
	\label{lem:real-estimate2}
\end{lemma}
\begin{proof}
  For ease of notation, we denote $\G{k}{i,j}$ by $H$ and $v_k$ by $v$.  Since
  $G$ satisfies Condition~\ref{asm:2G}, and
  $\deg_{H} (v) =\deg_G(v_k) - 1$ (since the neighbor $u$ of $v_k$ in $G$ is
  dropped in the construction of $H = \G{k}{i,j}$), we have
  $
	|L_{H} (v)| = \abs{ L_G(v_k)} \ge 2 \deg_G(v_k) \ge 2\cdot \deg_{H} (v) + 2.
	$

  Consider any valid coloring\footnote{Here, we say that a coloring $\sigma$ is
    \emph{valid} if the color $\sigma$ assigns to any vertex $v$ is from $L(v)$,
    and further, in case $w = 0$, no two neighbors are assigned the same color
    by $\sigma$.}  $\sigma'$ of the neighbors of $v$ in $H$.  For
  $k \in L_H(v)$, let $n_k$ denote the number of neighbors of $v$ that are
  colored $k$ in $\sigma'$.  Then for any $w \in [0, 1]$ and $i \in L_H(v)$,
  \begin{align*}
    \Pr{H,w}{c(v) = i | \sigma'} \maybealign = \frac{w^{n_i}}{\sum_{j \in L_H(v)}w^{n_j}} \eqbreak
    \maybealign \leq \frac{1}{\abs{L_H(v)} - \deg_H(v)},
  \end{align*}
  since at most $\deg_H(v)$ of the $n_j$ can be positive.  Note in particular
  that if $i$ is not a good color for $v$ in $H$, then the probability is $0$.
  Since this holds for any coloring $\sigma'$, we have
  $\Pr{H,w}{c(v) = i} \le \frac{1}{\abs{L_{H}(v)} - \deg_{H}(v)}.$ 
Now, let $d$ be the number of unpinned neighbors of $v$ in $H$. Noting that
  $\deg_{H}(v) \geq d$, and recalling the observation above that
  $\abs{L_H(v)} \geq 2\deg_H(v) + 2$, we thus have
  \begin{align*}
    \Pr{\G{k}{i,j},w}{c(v_k) = i } \maybealign = \Pr{H,w}{c(v) = i } \eqbreak
    \maybealign \le \frac{1}{\abs{L_{H}(v)} - \deg_{H}(v)} \leq \frac{1}{d + 2}.
\end{align*}
  Thus $v_k$ is \nice{} in $\G{k}{i,j}$.
\end{proof}

\subsection{Analysis for Condition~\ref {asm:176}}

Note that, as established in Proposition~\ref{prop:hereditary}, if $G$
satisfies Condition~\ref{asm:176} then so does $\G{k}{i,j}$.  Thus
in order to show that $v_k$ is \nice{} in $\G{k}{i,j}$, it suffices to show the
following more general fact.

\begin{lemma}Let $G$ be any graph that satisfies Condition~\ref{asm:176},
  and let $u$ be any unpinned vertex in $G$.  Then $u$ is \nice{} in $G$.
\label{lem:real-estimate176}
\end{lemma}
  The proof of this lemma is almost identical to arguments that appear in the
  work of Gamarnik, Katz and Misra~\cite{Gamarnik2015SSM} on strong spatial
  mixing; we include a proof here for completeness.
\par\smallskip  
\begin{proof}
  We show first that $\Pr{G,w}{c(u) = i} \le \frac{1}{\beta}$ whenever
  $L_G(u) \geq \deg_G(u) + \beta$; this will be required later in the proof.  To
  do so, we repeat the arguments in the proof of \MakeUppercase Lemma\nobreakspace \ref {lem:real-estimate2} to
  see that $\Pr{G,w}{c(u) = i} \leq \frac{1}{\abs{L(u)} - \deg_G(u)}$.  The
  claimed bound then follows since $\abs{L(u)} - \deg_G(u) \ge \beta$.

  Next we show that the upper bound of $\frac{1}{d+2}$, where $d$ is the number
  of unpinned neighbors of $u$ in $G$, holds conditioned on every coloring of
  the neighbors of the (unpinned) neighbors of $u$, by following a similar path
  as in~\cite{Gamarnik2015SSM}.  Consider any valid coloring~$\sigma'$ 
  (defined as in the proof of the previous lemma) of the
  vertices at distance \emph{two} from $u$.  Since $G$ is triangle free, we
  claim there is a tree $T$ of depth two rooted at
  $u$, with all the leaves pinned according to $\sigma'$, such that
  \begin{equation}
    \Pr{G,w}{c(u) = i | \sigma'} = \Pr{T,w}{c(u)=i}.
    \label{eq:31}
  \end{equation}
To see this, notice that once we condition on
  the coloring of the vertices at distance $2$ from $u$, the distribution of the
  color at $u$ becomes independent of the distribution of colors of vertices at
  distance $3$ or more. Further, because of triangle freeness, no two
  neighbors of $u$ have an edge between them, and hence any cycle in the
  distance-$2$ neighborhood, if one exists, must go through at least one pinned
  vertex.  We then observe that such a cycle can be broken by replacing any
  pinned vertex $v'$ in it with $\deg(v')$ copies, one for each of its neighbors: as
  discussed earlier, this operation cannot change the partition function or
  probabilities.  This operation therefore ensures that every pinned vertex in
  the resulting graph is now a leaf of a tree $T$ of depth $2$ rooted at
  $u$. Further, in $T$, the root $u$ has $d$ unpinned children, and all vertices
  at depth $2$ are pinned according to $\sigma'$.

  Let $v_1, \cdots, v_d$ be the $d$ unpinned neighbors of $u$ in $T$, and let
  $T_1, \cdots, T_d$ be the subtrees rooted at $v_1, \cdots, v_d$ respectively.
  For each $k \in L_G(u)$, let $n_k$ be the number of neighbors of $u$ that are
  pinned to color $k$.  Then by~Lemma\nobreakspace \ref {thm:recurrence},
  \[
    \ratio{T,u}{j,i}(w) =\frac{w^{n_j}\cdot \prod_{k=1}^d \inp{1 - \gamma \cdot
        \prob{T_k,w}{ c(v_k) = j} } }{w^{n_i} \cdot \prod_{k=1}^d \inp{1 -
        \gamma \cdot \prob{T_k,w}{ c(v_k) = i} } }.
  \]
  Define $t_{k j} \defeq \gamma \cdot \Pr{T_k, w}{c(v_k) = j }$, and note that
  from the calculation at the beginning of the proof, we have
  $0 \le t_{k j} \le \frac{\gamma}{\beta} \leq \frac{1}{\beta} \leq 1/2$.  Note
  also that $t_{k j} = 0$ if $j \not\in L(v_k)$.  Thus, we have
  \begin{equation}
    \sum_{j \in \good{u}} t_{k j} = \gamma \sum_{j \in \good{u} \cap
      L(v_k)} \Pr{T_k, w}{c(v_k) = j} \le \gamma \le 1.\label{eq:29}
  \end{equation}
  Therefore,
  \begin{align}
    \Pr{T,w}{ c(u) = i }
    \maybealign =\frac{1}{\sum_{j \in L(v)} \ratio{T,v}{j,i}(w)} \eqbreak
    \maybealign =\frac{w^{n_i}\cdot\prod_{k=1}^d \inp{1 - t_{k i} } }{\sum_{j \in L(u)} w^{n_j} \prod_{k=1}^d \inp{1 - t_{k j} }}\eqbreak
    \maybealign \le \frac{1}{\sum_{j\in \good{u}}\prod_{k=1}^d \inp{1 - t_{k j}} }, \label{eq:32}
  \end{align}
  where, in the last inequality we use that $n_j = 0$ when $j$ is good for $u$
  in $G$, and also that $w \in [0,1]$.

  Since $ \Pr{G,w}{c(u) = i | \sigma' } =\Pr{T,w}{c(u) = i }$, it remains to
  lower bound the denominator term $\sum_{j\in \good{u}}\prod_{k=1}^d \inp{1 - t_{k j}}$.  We begin
  by recalling the following standard consequence of the Taylor expansion of
  $\ln(1-x)$ around $0$: when $0 \leq x \leq \frac{1}{\beta} < 1$, and $\beta$
  is such that $(1-1/\beta)^2 \geq 1/2$,
  \begin{align}
    \ln(1 - x) \maybealign\geq {-x - \frac{x^2}{2(1-1/\beta)^2}} \eqbreak
    \maybealign\geq {-x - x^2} \eqbreak
    \maybealign\geq {-\inp{1+\frac{1}{\beta}}x}.
    \label{eq:30}
  \end{align}
  Note that the condition required of $\beta$ is satisfied since
  $\beta \geq 2\alpha \geq 7/2$, as stipulated in Condition~\ref{asm:176}.
  Since $0 \leq t_{k j} \leq 1/\beta$, we therefore obtain, for every $j \in \good{u}$,

\begin{align}
    \prod_{k=1}^d (1 - t_{k j}) \maybealign\ge \prod_{k=1}^d \exp\inp{-\inp{1+\frac{1}{\beta}}t_{k j}} \eqbreak
    \maybealign = \exp\inp{-\inp{1+\frac{1}{\beta}} \sum_{k=1}^d t_{k j}}.
    \label{eq:taylor-ln-exp}
  \end{align}
  For convenience of notation, we denote $\abs{\good{u}}$ by $q_u$.  Note that
  since $\abs{L(u)} \geq \alpha\deg(u) + \beta$, and $u$ has $\deg(u) - d$
  pinned neighbors, we have
  \begin{align}
    q_u \maybealign\geq \abs{L(u)} - (\deg(u) - d) \eqbreak
    \maybealign\geq \abs{L(u)} - \alpha (\deg(u) - d) \eqbreak
    \maybealign\geq \alpha d + \beta,
    \label{eq:qu-bound}
  \end{align}
  where in the second inequality we use $\alpha \geq 1$. Now, by the AM-GM
  inequality, we get

  \begin{align*}
	\maybealign\sum_{j\in \good{u}}\prod_{k=1}^d \inp{1 - t_{k j}}\eqbreak
	&\ge q_u \inp{\prod_{j \in \good{u}} \prod_{k=1}^d \inp{1 - t_{k j}}}^{\frac{1}{q_u}} \\
	&\ge q_u \exp\inp{-\frac{1+1/\beta}{q_u}  \cdot  \sum_{k=1}^d \sum_{j \in \good{u}} t_{k j}}, & \hbox{  using~eq.\nobreakspace \textup {(\ref {eq:taylor-ln-exp})}} \\
	&\ge (\alpha d + \beta) \exp\inp{-\frac{d(1+1/\beta)}{\alpha d + \beta}  } , & \hbox{  by eqs.\nobreakspace \textup {(\ref{eq:29}) and \textup{(\ref{eq:qu-bound})} }}\\
	&\ge (d+2) \alpha \cdot  \exp\inp{-\frac{(1+1/\beta)}{\alpha}  } , & \hbox{  using $\beta\ge 2\alpha$}\\
	&\ge (d+2),
  \end{align*}
  where the last line uses the stipulation in Condition~\ref{asm:176} that $\alpha$ and $\beta$ satisfy
  $\alpha \cdot  \exp\inp{-\frac{(1+1/\beta)}{\alpha}} \geq 1$.  From eqs.\nobreakspace \textup {(\ref {eq:31})} and\nobreakspace  \textup {(\ref {eq:32})} we
  therefore get
  \begin{displaymath}
    \Pr{G,w}{c(u) = i|\sigma'} \le \frac{1}{d + 2}.
\end{displaymath}
  Since this holds for any conditioning $\sigma'$ of the colors of the neighbors
  of the neighbors of $u$ in $G$, we then have
  \begin{displaymath}
    \Pr{G,w}{c(u) = i} \le \frac{1}{d + 2},
\end{displaymath}
  which concludes the proof.
\end{proof}
\subsection{Analysis for Condition~\ref{asm:tree}}
\begin{lemma}
  Let $G$ be a list-coloring instance that satisfies Condition~\ref{asm:tree}
  (in particular, $G$ is a forest), and let $u$ be a marked unpinned
    vertex in $G$.  Then any unpinned neighbor $v_k$ of $u$
  is \nice{} in $\G{k}{i,j}$.
	\label{lem:real-estimate-tree}
\end{lemma}
\begin{proof}
Since $G$ is a forest, and all pinned vertices in the connected
    component of $u$ in $G$ must be neighbors of $u$ (since $u$ is, by
    Condition~\ref{asm:tree}, the unique marked vertex
    in its component), we see that the connected
    component of $v_k$ in $\G{k}{i,j}$ contains no pinned vertices.  Since all
    unpinned vertices in $G$ have the same list, which is of size $q \geq \Delta + 1$
    (where $\Delta$ is the maximum degree of $G$), it follows by symmetry that
    the marginal distribution of the color of $v_k$ is uniform.  Further, since
  the neighbor $u$ of $v_k$ in $G$ is not present in $\G{k}{i,j}$, we know that
  $v_k$ has $d \leq \Delta - 1$ unpinned neighbors in $\G{k}{i,j}$.  Thus, for
  each $i \in L(v_k)$,
  \[\Pr{\G{k}{i,j},w}{c(v_k) = i } = \frac{1}{q} \le \frac{1}{\Delta+1} \le
    \frac{1}{d+2},\] which establishes that $v_k$ in nice in $\G{k}{i,j}$.
\end{proof}

\begin{proof}[Proof of Lemma~\ref{obv:good-to-nice}]
  The proof of Lemma~\ref{obv:good-to-nice} now follows by combining
  Proposition~\ref{prop:hereditary} and Lemmas~\ref{lem:real-estimate2},
  \ref{lem:real-estimate176} and~\ref{lem:real-estimate-tree}, along with the
  simple observations about properties~(ii) and~(iii) preceding
  Remark~\ref{rem:convention}.
\end{proof}

We conclude this section by noting that the niceness condition can be
strengthened in the case when all the list sizes are uniformly large (e.g., as
in the case of standard $q$-colorings).
\begin{numremark} \label{rem:nicer} In Conditions~\ref{asm:2G}
  and~\ref{asm:176}, if we replace the degree of a vertex by the maximum degree
  $\Delta$ (i.e., in Condition~\ref{asm:2G}, if we assume
  $\abs{L(v)} \geq 2\Delta$, and in Condition~\ref{asm:176}, we assume
  $\abs{L(v)} \geq \alpha\Delta + \beta$, for each $v$), then for every vertex
  $v$ in the graph $G$ we also have
  $\Pr{G,w}{c(v) = i} < \min\set{\frac{4}{3\Delta}, 1}$.

  To see this, notice that the same calculation as in the proof of
  Lemma\nobreakspace \ref {lem:real-estimate2} above gives
  \[\Pr{G,w}{c(v) = i} \le \frac{1}{\abs{L(v)} - \Delta } < \frac{4}{3\Delta},
  \]
  under the maximum degree versions of both Conditions~\ref{asm:2G} and \ref{asm:176}.
We will refer to this stronger condition on list sizes as the \emph{uniformly large list
size} condition.  Note that the maximum degree versions of the conditions are also admissible by
the same arguments as those for Conditions~\ref{asm:2G} and~\ref{asm:176}.
\end{numremark}

\section{Zero-free region for small $\abs{w}$}
\label{sec:induction-origin}
As explained in the introduction, all our algorithmic results follow from
Theorem\nobreakspace \ref {thm:zeros}, which establishes a zero-free region for the partition
function $Z_G(w)$ around the interval $[0,1]$ in the complex plane.  We split
the proof of Theorem\nobreakspace \ref {thm:zeros} into two parts: in this section, we establish the
existence of a zero-free disk around the endpoint $w=0$ (see
Theorem\nobreakspace \ref {thm:main-origin-restated}): this is the most delicate case because $w=0$
corresponds to proper colorings.  Then in Section\nobreakspace \ref {sec:induction-interval} (see
Theorem\nobreakspace \ref {thm:main-interval-restated}) we derive a zero-free region around the
remainder of the interval, using a similar but less delicate approach.  Taken
together, Theorems\nobreakspace \ref {thm:main-origin-restated} and\nobreakspace  \ref {thm:main-interval-restated} immediately
imply Theorem\nobreakspace \ref {thm:zeros}, so this will conclude our analysis.
\begin{theorem}
  Fix a positive integer $\Delta$, and let $\mathcal{L}$ be an
    admissible list condition.  There exists a $\nu_w = \nu_w(\Delta)$ such
  that the following is true.  Let $G$ be a graph of maximum degree $\Delta$
  satisfying the admissible list condition $\mathcal{L}$, and having no
  pinned vertices.  Then, $Z_G(w) \neq 0$ for any $w$ satisfying
  $\abs{w} \leq \nu_w$.
  \label{thm:main-origin-restated}
\end{theorem}

In the proof, we will encounter several constants which we now fix.  
Given the degree bound $\Delta \geq 1$, we define
\begin{equation}
  \eps_R \defeq \frac{0.01}{\Delta^2},\,
  \eps_I \defeq \eps_R\cdot\frac{0.01}{\Delta^2},\text{ and }
  \eps_w \defeq \eps_I\cdot\frac{0.01}{\Delta^3}.\label{eq:def-eps}
\end{equation}
We will then see that the quantity $\nu_w$ in the statement of the theorem can
be chosen to be $0.2\eps_w/2^\Delta$.  (In fact, we will show that if one has
the slightly stronger assumption of uniformly large list sizes, as considered in
\protect \MakeUppercase {R}emark\nobreakspace \ref {rem:nicer}, then $\nu_w$ can be chosen to be $\eps_w/(300\Delta)$.)

Throughout the rest of this section, we fix $\Delta$ to be the maximum degree of
the graphs, and let $\eps_w,\eps_I,\eps_R$ be as above.

We now briefly outline our strategy for the proof.  Recall that, for a vertex
$u$ and colors $i,j$, the marginal ratio is given by
$ \ratio{G,u}{i,j}(w) = \frac{\Z{G,u}{i}(w)}{\Z{G,u}{j}(w)}.  $ When $G$ is an
unconflicted graph, $\ratio{G,u}{i,j}(0)$ is always a well-defined non-negative
real number.  Intuitively, we would like to show that
$\ratio{G,u}{i,j}(w) \approx \ratio{G,u}{i,j}(0)$, independent of the size of
$G$, when $w \in \C$ is close to $0$.  Given such an approximation one can use a
simple geometric argument (see~Consequence\nobreakspace \ref
{lem:sum-lowerbounds}) to conclude that the partition function does not vanish
for such $w$.  In order to prove the above approximate equality inductively for
a given graph $G$, we take an approach that exploits the properties of the
``real'' case (i.e., of $\ratio{G,u}{i,j}(0)$) and then uses the notion of
``\nice{}ness'' of certain vertices described earlier to control the
accumulation of errors.  To this end, we will prove the following lemma via
induction on the number of unpinned vertices in $G$.  \MakeUppercase
Theorem\nobreakspace \ref {thm:main-origin-restated} will follow almost
immediately from the lemma; see the end of this section for the details.
Throughout the section, we fix an admissible list condition $\mathcal{L}$, and a
$w \in \C$ satisfying $|w| \leq \nu_w$ (as in the statement of
Theorem~\ref{thm:main-origin-restated}).

\begin{lemma}
  \label{lem:origin-induction}
  Let $G$ be an unconflicted graph of maximum degree $\Delta$ satisfying
  an admissible list condition $\mathcal{L}$, and let $u$ be any
  marked unpinned vertex in $G$.  Then, the following are
  true (with $\eps_w, \eps_I,$ and $\eps_I$ as defined in eq.\nobreakspace
  \textup {(\ref {eq:def-eps})}):
  \begin{enumerate}\item \label{ind-0} For $i \in \good{u}$, $\abs{ \Z{G,u}{i} (w)} > 0 $.
  \item \label{ind-1} For $i,j\in \good{u}$, if $u$ has all neighbors pinned,
    then $\ratio{G,u}{i,j}(w) = \ratio{G,u}{i,j}(0) = 1$. \item \label{ind-2} For $i,j\in \good{u}$, if $u$ has $d \geq 1$ unpinned
    neighbors, then
    \[
      \frac{1}{d}\abs{\Re \ln \ratio{G,u}{i,j}(w) - \Re \ln
        \ratio{G,u}{i,j}(0)} < \eps_R.
    \]
  \item \label{ind-3} For any $i,j\in \good{u}$, if $u$ has $d \geq 1$ unpinned
    neighbors, we have
    $\frac{1}{d} \abs{\Im \ln \ratio{G,u}{i,j}(w)} < \eps_I$.
  \item \label{ind-4} For any $i \not\in \good{u}, j \in \good{u}$, then
    $\abs{\ratio{G,u}{i,j}(w)} \le \eps_w$.
  \end{enumerate}
\end{lemma}
We will refer to items\nobreakspace  \ref {ind-0} to\nobreakspace  \ref {ind-4}  as ``items of the
induction hypothesis''.  The rest of this section is devoted to the proof of
this lemma via induction on the number of unpinned vertices in $G$.

We begin by verifying that the induction hypothesis holds in the base case when
$u$ is the only unpinned vertex in an unconflicted graph $G$.  In this case,
items\nobreakspace \ref {ind-2} and\nobreakspace  \ref {ind-3} are vacuously true since $u$ has no unpinned neighbors.
Since all neighbors of $u$ in $G$ are pinned, the fact that all pinned vertices
have degree at most one implies that $G$ can be decomposed into two disjoint
components $G_1$ and $G_2$, where $G_1$ consists of $u$ and its pinned
neighbors, while $G_2$ consists of a disjoint union of unconflicted edges (since
$G$ is unconflicted).  Now, since $G_1$ and $G_2$ are disjoint components, we
have $\Z{G, u}{i}(w) = Z_{G_2}(w) = 1$ for all $i \in \good{G, u}$ and all
$w \in \C$.  This proves items\nobreakspace \ref {ind-0} and\nobreakspace  \ref {ind-1}.  Similarly, when
$i \not\in \good{G, u}$, we have $\Z{G, u}{i}(w) = w^{n_i}$, where $n_i \geq 1$
is the number of neighbors of $u$ pinned to color $i$.  This gives
\begin{displaymath}
  {\abs{\ratio{G,u}{i,j}(w)}} \leq \abs{w}^{n_i} \leq \eps_w,
\end{displaymath}
since $\abs{w} \leq \eps_w \leq 1$, and proves item\nobreakspace \ref {ind-4}.

We now derive some consequences of the above induction hypothesis that will be
helpful in carrying out the induction.  Throughout, we assume that $G$ is an
unconflicted graph satisfying an admissible list condition
  $\mathcal{L}$, and $u$ is a marked unpinned vertex in $G$.
\begin{consequence}
  \label{lem:sum-lowerbounds}
    $\abs{Z_G(w)} \ge 0.9\min_{i \in \good{u}} \abs{\Z{G,v}{i}(w)} >
    0.$
\end{consequence}
\begin{proof}
  Note that $Z_G(w) = \sum_{i \in L(u)} \Z{G,u}{i} (w)$.  Recall also that since
  $u$ is an unpinned vertex in $G$ and $G$ satisfies an admissible list
  condition $\mathcal{L}$, we have
  \[\abs{L(u)} \geq \deg_{G}(u) + 1.\]
  Now, from item\nobreakspace \ref {ind-3}, we
  see that the angle between the complex numbers $\Z{G,u}{i} (w)$ and
  $\Z{G,u}{j} (w)$, when $i, j \in \good{u}$, is at most $d\eps_I$.  Applying
  \MakeUppercase Lemma\nobreakspace \ref {lem:geometric-proj} to the terms corresponding to the good colors and
  item\nobreakspace \ref {ind-4} to the terms corresponding to the bad colors, we then have
  \begin{align*}
    \maybealign\biggl|{\sum_{i \in L(u)} \Z{G,u}{i} (w)}\biggr|\eqbreak
    &\ge \inp{ \abs{\good{u}} \cos \frac{d\eps_I}{2} -
      \abs{L(u) \setminus \good{u}} \eps_w}
      \min_{i \in \good{u}}
      \abs{\Z{G,u}{i}(w)}\\
&\geq \inp{\cos \frac{d\eps_I}{2} -
      \deg_G(u) \cdot \eps_w}
      \min_{i \in \good{u}}
      \abs{\Z{G,u}{i}(w)},
  \end{align*}
  where we use the fact that $\abs{L(u) \setminus \good{u}} \leq \deg_G(u)$ and $|L(u)| \geq \deg_G(u) + 1$ in the last
  inequality. Since $d\eps_I \leq 0.01$ and $\eps_w \leq 0.01/\Delta$, we then
  have
  $\abs{\sum_{i \in L(u)} \Z{G,u}{i}(w)} \ge 0.9\min_{i \in \good{u}}
  \abs{\Z{G,v}{i}(w)}$, which in turn is positive from
  item\nobreakspace \ref {ind-0}.  \end{proof}

\begin{consequence}
The pseudo-probabilities approximate the real probabilities in the following
  sense:
  \begin{enumerate}
  \item for any $i \not \in \good{u}$, $\abs{\prob{G,w}{c(u) = i}} \le 1.2\eps_w$.
  \item for any $j\in \good{u}$,
    \begin{align*}
      \abs{\Im \ln
      \frac{\prob{G,w}{c(u) = j}}{\prob{G}{c(u) = j}}} 
      &= \abs{\Im \ln \prob{G,w}{c(u) = j}}\eqbreak
      \maybealign\leq d\eps_I + 2\Delta\eps_w,
    \end{align*}
    and,
    \begin{align*}
      \abs{\Re \ln
      \frac{\prob{G,w}{c(u) = j}}{\prob{G}{c(u) = j}}}
      &\leq d \eps_R + d\eps_I + 2\Delta\eps_w,
    \end{align*}
  \end{enumerate}
  where $d$ is the number of unpinned neighbors of $u$ in $G$.
  \label{lem:approx-prob}
\end{consequence}
\begin{proof}
	For part (1), by~Consequence\nobreakspace \ref {lem:sum-lowerbounds} we have 
	\begin{align*}
      \abs{\prob{G,w}{c(u) = i}} &= \frac{\abs{\Z{G,u}{i}(w)} }{\abs{Z_G(w)}}\\
      &\le \frac{\abs{\Z{G,u}{i}(w)} }{0.9\min_{j \in \good{u}} \abs{\Z{G,u}{j}(w)}} 
         \le 1.2\eps_w,
	\end{align*}
		where the last inequality follows from induction hypothesis~item\nobreakspace \ref {ind-4}.

		For part (2), by items\nobreakspace  \ref {ind-1} to\nobreakspace  \ref {ind-3}  of the induction hypothesis,
        there exist complex numbers $\xi_i$ (for all $i \in \good{u}$)
        satisfying $\abs{ \Re \xi_i} \le d \eps_R$ and
        $\abs{\Im \xi_i} \le d\eps_I$ such that
		\begin{align*}
          \maybealign\frac{1}{\prob{G,w}{c(u) = j}} \eqbreak
	  &= \sum_{i \in L(u)} \frac{\Z{G,u}{i} (w)}{\Z{G,u}{j} (w)} \\
          &=  \underbrace{\sum_{i \in \good{u}} \frac{\Z{G,u}{i}
             (0)}{\Z{G,u}{j} (0)} e^{\xi_i}}_{\defeq A} + \;\;\;\;
             \underbrace{\sum_{i \in L(u) \setminus \good{u}} \frac{\Z{G,u}{i} (w)}{\Z{G,u}{j} (w)}}_{\defeq B}.
		\end{align*}
Next we show that $A \approx \frac{1}{\prob{G}{c(u) = j}}$ and $B$ is negligible.
	From item\nobreakspace \ref {ind-4} of the
        induction hypothesis we have
        \begin{equation}
          \label{eq:2}
          \prob{G}{c(u) = j}\cdot \abs{B} \leq \Delta\eps_w.
        \end{equation}
        Now, note that
        $\sum_{i \in\good{u}} \frac{\Z{G,u}{i} (0)}{\Z{G,u}{j} (0)} =
        \frac{1}{\prob{G}{c(u) = j}}$.  Further, when $\eps_I \leq 0.1/\Delta$,
        we also have\footnote{Here, we also use the elementary facts that if $z$
          is a complex number satisfying $\Re z = r$ and
          $\abs{\Im z} = \theta \leq 0.1$ then
          $\abs{\arg e^z} = \abs{\Im z} = \theta$, and
          $e^r \geq \Re e^z = e^r\cos \theta = \exp(r+ \ln \cos \theta) \geq
	\exp(r - \theta^2) \ge e^r - e^r\theta^2$. Hence if $r<0$, we have $\Re e^z \ge e^r - \theta^2$.
	}
        \begin{equation}
		\Re{e^{\xi_i}} \in \inp{e^{-d\eps_R} - d^2\eps_I^2,\;
          e^{d\eps_R}}\text{, and } |\arg {e^{\xi_i}}| \leq
          d\eps_I.\label{eq:ind-origin-1}
        \end{equation}
        The above will therefore be true also for any convex combination of the
        $e^{\xi_i}$.  Noting that ${\prob{G}{c(u) = j}} \cdot A$ is just such a
        convex combination (as the coefficients of the $e^{\xi_i}$ are
        non-negative reals summing to $1$), we have
        \begin{align}
          \label{eq:3}
          {\prob{G}{c(u) = j}} \cdot \Re A
          \in (e^{-d\eps_R} - d^2\eps_I^2,\; e^{d\eps_R}), \\ \label{eq:4} |\arg \inp{{\prob{G}{c(u) = j}} \cdot  A}|
          \leq d\eps_I.
        \end{align}
        Together, eqs.\nobreakspace \textup {(\ref {eq:2})},  \textup {(\ref {eq:3})} and\nobreakspace  \textup {(\ref {eq:4})} imply that if
        $C \defeq \frac{\prob{G}{c(u) = j}}{\prob{G,w}{c(u) = j}}$ then (using
        the values of $\eps_R, \eps_I$, and $\eps_w$)\footnote{Here, for the
          second inclusion, we use the following elementary computation.  Let
          $z, s$ be complex numbers such that $\Re z = r \in [0.9, 1.1]$,
          $\abs{\arg z} = \theta \leq 0.1$ and $\abs{s} \leq 0.1$.  Then, we
	  have $\Re(z + s) \geq r - \abs{s}$ and $\abs{\Im (z + s)} \leq r\theta + \abs{s}$.
          Thus,
          $\abs{\arg (z+s)} \leq \frac{\abs{\Im(z + s)}}{\abs{\Re(z + s)}} \leq
	  \frac{r\theta + \abs{s}}{r - \abs{s}} = \theta + \abs{s}\cdot\frac{1+\theta}{r -\abs{s}} \leq
	  \theta + 2\abs{s}$.}
        \begin{align*}
          \Re C &\in \inp{e^{-d\eps_R} - d^2\eps_I^2 - \Delta\eps_w,\; e^{d\eps_R} +
                  \Delta\eps_w},\text{ and }\\
	  \arg C &\in \inp{-d\eps_I - 2\Delta\eps_w, d\eps_I + 2\Delta\eps_w}.
        \end{align*}
        Thus, since $\eps_I, \eps_R$ are small enough and
        $\eps_w \leq 0.01\min\inbr{\eps_I, \eps_R}$, we have
        \begin{align*}
          |\Re \ln C| &\leq d\eps_R + d\eps_I + 2\Delta\eps_w,
                         \text{ and }\\
          |\Im \ln C| &\leq d\eps_I + 2\Delta\eps_w.
        \end{align*}
	Here we use the elementary fact that for $z \in \C$,
          $\Re \ln z = \ln \abs{z}$ and $\Im \ln z = \arg z$.  
	  Further, for $z$ satisfying $\Re z = r\in [0.9,1.1]$ and $\abs{\arg z} = \theta \le 0.1$, 
	  we also have
          $\ln r \leq \Re \ln z \leq \ln r + \ln \sec \theta \leq \ln r +
          \theta^2$.
\end{proof}

In the next consequence, we show that the error contracts during the induction.
We first set up some notation.  For a graph $G$, a
vertex $u$, and a color $i \in \good{u}$, we let
$\as{G,u}{i}(w) = \ln \prob{G,w}{c(u)=i}$.  We also recall that $\gamma\defeq 1 - w$, 
and the definition of the function $f_\gamma(x) \defeq - \ln ( 1 - \gamma e^x)$ from eq.\nobreakspace \textup {(\ref {eq:28})}.
\begin{consequence}
  \label{lem:real-part-ind} There exists a positive constant
  $\eta \in [0.9, 1)$ so that the following is true. Let $d$ be the number of
  unpinned neighbors of $u$.  Assume further that $u$ is \nice{} in $G$. Then,
  for any colors $i, j \in \good{u}$, there exists a real number
  $C = C_{G,u,i} \in [0, \frac{1}{d + \eta}]$ such that
  \begin{align}
    \Big| \Re f_\gamma( \maybealign  \as{G,u}{i}(w)) -  f_1(\as{G,u}{i}(0))\eqbreak
	  \maybealign - C\cdot \Re (\as{G,u}{i}(w) - \as{G,u}{i}(0))\Big|
    &\le \eps_I + \eps_w ;\label{eq:6}\\ 
    \Big|\Im  f_\gamma( \maybealign \as{G,u}{i}(w)) - \Im f_\gamma(\as{G,u}{j}(w))\Big|\eqbreak
    &\maybealign\le  \frac{1}{d + \eta} \cdot (d\eps_I + 4 \Delta \eps_w) +
      2\eps_w;\label{eq:7}\\
    \Big|\Im  f_\gamma( \maybealign \as{G,u}{i}(w))\Big|  
    &\le  \frac{1}{d + \eta} \cdot (d\eps_I + 4 \Delta \eps_w) +
      \eps_w.\label{eq:18}
  \end{align}
\end{consequence}
\begin{proof}
  Since $u$ is \nice{} in $G$, the bound
  $\prob{G, 0}{c(u) = k} \le \frac{1}{d+2}$ (for any $k \in \good{G,
    u}$) applies.  Combining them with Consequence\nobreakspace \ref {lem:approx-prob} we see that
  $\as{G,u}{i}(w), \as{G,u}{i}(0), \as{G,u}{j}(w), \as{G,u}{j}(0)$ lie in a
  domain $D$ as described in \MakeUppercase Lemma\nobreakspace \ref {obv:f-props-int} (with the parameter $\kappa$
  therein set to $1$), with the parameters $\zeta$ and $\tau$ in that
  observation chosen as
  \begin{equation}
    \begin{aligned}
      \zeta &= \ln(d+2) - d\eps_R - d\eps_I - 2\Delta\eps_w\,\text{, and }\\
      \tau &= d\eps_I + 2\Delta\eps_w.
    \end{aligned}
  \end{equation}
  Here, for the bound on $\zeta$, we use the fact that for $j \in \good{G, u}$,
  $\prob{G}{c(u) = j} \leq \frac{1}{d+2}$, which is due to $u$ being \nice{} in
  $G$. 

  The bounds on $\eps_w, \eps_I$ and $\eps_R$ now imply
  $e^\zeta \geq (d+2)\inp{1 - \frac{0.02}{\Delta}} \geq d + 1.94$, and also that
  $\tau \leq 0.02/\Delta$.  Thus, the conditions required on $\zeta$ and $\tau$
  in Lemma\nobreakspace \ref {obv:f-props-int} (i.e. that $\tau < 1/2$ and $\tau^2 + e^{-\zeta} < 1$)
  are satisfied. Further, $\rho_R$ and $\rho_I$ as set in the observation
  satisfy $\rho_R \le \frac{1}{d + \eta}$, where $\eta$ can be taken to be
  $0.94$, and $\rho_I < 3\eps_I$.  
  
  Using~Lemma\nobreakspace \ref {thm:signed-mean} followed by the value of $\eps_w$, and noting that $\as{G,u}{i}(0)$ is a real number, we then have
  \begin{align}
	\Big| \Re f_1( \maybealign \as{G,u}{i}(w)) - f_1(\as{G,u}{i}(0)) \eqbreak
	\maybealign  - C \cdot \Re \inp{ \as{G,u}{i}(w) - \as{G,u}{i}(0)}\Big| \eqbreak
	&\le \rho_I \cdot  \maybealign\abs{ \Im \inp{ \as{G,u}{i}(w) - \as{G,u}{i}(0)}} \nonumber \\
	&\le 3\eps_I\maybealign(d\eps_I + 2\Delta\eps_w) \leq 4d\eps_I^2 \leq \eps_I,\label{eq:19}
  \end{align}
  for an appropriate non-negative $C \leq 1/(d+\eta)$.  This is almost~eq.\nobreakspace \textup {(\ref {eq:6})};
  the difference will be handled later.

  Similarly, applying Lemma\nobreakspace \ref {thm:signed-mean} to the imaginary part we have
  \begin{multline}
    \abs{ \Im f_1(\as{G,u}{i}(w)) - \Im f_1(\as{G,u}{j}(w))\Big)}\\
    \le \rho_R\cdot \max\set{\abs{\Im \inp{\as{G,u}{i}(w) - \as{G,u}{j}(w)}} , \right. \eqbreakno \left.
      \abs{\Im \as{G,u}{i}(w)}, \abs{\Im \as {G,u}{j}(w)}},\label{eq:21}
  \end{multline}
  where, as noted above, $\rho_R \leq \frac{1}{d + \eta}$.  Now, note that the
  first term in the above maximum is less than $d\eps_I$ by item\nobreakspace \ref {ind-3} of the
  induction hypothesis, while the other two terms are at most
  $d\eps_I + 2\Delta\eps_w$ from item 2 of \MakeUppercase Consequence\nobreakspace \ref {lem:approx-prob}.
  This is almost the bound in~eq.\nobreakspace \textup {(\ref {eq:7})}; again, the difference will be handled later.

  To prove the bound in eq.\nobreakspace \textup {(\ref {eq:18})}, we first apply the imaginary part of
  \MakeUppercase Lemma\nobreakspace \ref {thm:signed-mean} along with the fact that $\Im \as{G, u}{i}(0) = 0$ to
  get
  \begin{align}
    \abs{\Im f_1(\as{G,u}{i}(w))} \maybealign= \abs{ \Im f_1(\as{G,u}{i}(w)) - f_1(\as{G,u}{i}(0))} \eqbreak
      \maybealign\le \rho_R\cdot \abs{\Im \inp{\as{G,u}{i}(w)}} \eqbreak
      \maybealign\leq \frac{1}{d+\eta}(d\eps_I + \Delta\eps_w).\label{eq:22}
  \end{align}

  Finally, we use item 2 of \MakeUppercase Lemma\nobreakspace \ref {obv:f-props-int} (with the parameter $\kappa'$
  therein set to $\gamma$) to conclude the proofs of~eqs.\nobreakspace  \textup {(\ref {eq:6})} to\nobreakspace  \textup {(\ref {eq:18})} . To
  this end, we note that $\gamma$ satisfies $\abs{\gamma - 1} \leq \eps_w$, so
  that the condition $(1 + \eps_w) < e^\zeta$ required for item 2 to apply is
  satisfied.  Thus we see that for any $z \in D$,
  \[
	\abs{ f_\gamma(z) - f_1(z) } \le \eps_w,
  \]
  so that the quantities
  $|\Re f_\gamma(\as{G,u}{i}(w)) - \Re f_1(\as{G,u}{i}(w))|$,
  $|\Im f_\gamma(\as{G,u}{i}(w)) - \Im f_1(\as{G,u}{i}(w))|$,
  $|\Im f_\gamma(\as{G,u}{j}(w)) - \Im f_1(\as{G,u}{j}(w))|$, and
  $|\Im f_\gamma(\as{G,u}{j}(w)) - \Im f_1(\as{G,u}{j}(w))|$ are all at most
  $\eps_w$. The desired bounds of~eqs.\nobreakspace  \textup {(\ref {eq:6})} to\nobreakspace  \textup {(\ref {eq:18})}  now follow from the triangle inequality and the
  bounds in eqs.\nobreakspace  \textup {(\ref {eq:19})} to\nobreakspace  \textup {(\ref {eq:22})} .
\end{proof}

We set up some further notation for the next consequence.  For a color
$i \in L(u) \setminus \good{u}$ we let $\bs{G,u}{i}(w) = \prob{G,w}{c(u)=i}$.
We then consider the function $g_\gamma(x) \defeq - \ln (1-\gamma x)$.
\begin{consequence}
  For every color $i \not\in \good{u}$,
  $\abs{g_\gamma(\bs{G,u}{i}(w))} \le 2\eps_w.$
  \label{lem:w-part-ind}
\end{consequence}
\begin{proof}
  Item 1 of \MakeUppercase Consequence\nobreakspace \ref {lem:approx-prob} implies that
  $\abs{\bs{G,u}{i}(w)} \le 1.2\eps_w$.  Thus, recalling that
  $\abs{\gamma - 1} \leq \eps_w$, we get that for all $\eps_w < 0.01$,
  $\abs{g_\gamma(\bs{G,u}{i}(w))} = \abs{\ln (1 - \gamma \bs{G,u}{i}(w))} \le
  2\eps_w$.
\end{proof}

\subsection*{Inductive proof of Lemma\nobreakspace \ref {lem:origin-induction}} We are now ready to
see the induction step in the proof of Lemma\nobreakspace \ref {lem:origin-induction}; recall that
the base case (when $u$ is the only unpinned vertex in~$G$)
was already established immediately following the statement of the lemma.  Let
$G$ be any unconflicted graph which satisfies the admissible list condition $\mathcal{L}$ and has at
least two unpinned vertices.
We first prove induction~item\nobreakspace \ref {ind-0} for any marked unpinned
  vertex $u$ in $G$.  Consider the graph $G'$ obtained from $G$
by pinning vertex $u$ to color $i$.  Note that by the definition of the pinning
operation, $\Z{G, u}{i}(w) = Z_{G'}(w)$.  When $i \in \good{G, u}$, the graph
$G'$ is also unconflicted and, further, since $\mathcal{L}$ is
  hereditary (because it is admissible), satisfies the admissible list condition
  $\mathcal{L}$.  Also, $G'$ has one fewer unpinned vertex than $G$.  Thus,
from Consequence\nobreakspace \ref {lem:sum-lowerbounds} of the induction
hypothesis applied to $G'$, we have that
$\abs{\Z{G, u}{i}(w)} = \abs{Z_{G'}(w)} > 0$.

We now consider~item\nobreakspace \ref {ind-1}.  When all neighbors of $u$ in
$G$ are pinned, the fact that all pinned vertices have degree at most one
implies that $G$ can be decomposed into two disjoint components $G_1$ and $G_2$,
where $G_1$ consists of $u$ and its pinned neighbors, while $G_2$ is also
unconflicted (when $G$ is unconflicted) and has one fewer unpinned vertex than
$G$.  Note also that $G_2$, being a connected component of $G$, also satisfies
the admissible list condition $\mathcal{L}$ (since $\mathcal{L}$ is hereditary).
Thus, from Consequence\nobreakspace \ref {lem:sum-lowerbounds} of the induction
hypothesis applied to $G_2$, we get that $Z_{G_2}(w)$ and $Z_{G_2}(0)$ are both
non-zero.  Now, since $G_1$ and $G_2$ are disjoint components, we have
$\Z{G, u}{k}(x) = Z_{G_2}(x)$ for all $k \in \good{G, u}$ and all $x \in \C$.
It therefore follows that when $i, j \in \good{G, u}$,
$\ratio{G, u}{i,j}(w) = \ratio{G, u}{i,j}(0) = 1$.

We now consider~items\nobreakspace \ref {ind-2} and\nobreakspace  \ref {ind-3}.  Recall that by~Lemma\nobreakspace \ref {thm:recurrence}, we
have
\begin{equation}
  \ratio{G,u}{i,j}(w) = \prod_{k=1}^{\deg_G(u)} \frac{1 - \gamma
      \prob{\G{k}{i,j},w}{ c(v_k) = i}  }{ 1 - \gamma \prob{\G{k}{i,j},w}{
        c(v_k) = j}  }.
\label{eq:26}
\end{equation}
For simplicity we write $G_k \defeq \G{k}{i,j}$.  Note that when
$i, j \in \good{G, u}$, and $G$ is unconflicted, so are the $G_k$.  Note also
that when $i, j \in \good{G, u}$, we can restrict the product above to the $d$
unpinned neighbors of $u$, since for such $i, j$, the contribution of the factor
corresponding to a pinned neighbor is $1$, irrespective of the value of $w$.
Without loss of generality, we relabel these unpinned neighbors as
$v_1, v_2, \dots, v_d$.

Since $\mathcal{L}$ is hereditary, $G_k$ also satisfies $\mathcal{L}$, and the
vertex $v_k$ is marked in $G_k$ (since $u$ was marked in
$G$).  Further, each $G_k$ has exactly one fewer unpinned vertex than $G$, so
that the induction hypothesis applies to each $G_k$ at the vertex $v_k$.

Now, as before, for $s \in \good{G_k,v_k}$ we define
$\as{G_k,v_k}{s}(w) \defeq \ln \prob{G_k,w}{c(v_k)=s}$; while for
$t \in L(v_k) \setminus \good{G_k, v_k}$ we let
$ \bs{G_k,v_k}{t}(w) \defeq \prob{G_k,w}{c(v_k)=t}$.  For a graph $G$, a vertex
$u$ and a color $s$, we let $\Bad{G,u}(s)$ be the set of those neighbors of $u$
for which $s$ is a bad color in $G\setminus \set{u}$.  For simplicity we will
also write $\Bi[s] \defeq \Bad{G,u}(s)$ when it is clear from the context.  As
before, we have $\gamma = 1 - w$,
$f_\gamma(x) = - \ln (1 - \gamma e^x), g_\gamma(x) = -\ln(1-\gamma x)$.  From
the above recurrence, we then have:
\begin{align}
 \maybealign -\ln \ratio{G,u}{i,j}(w)\eqbreak 
 &=\sum_{v_k \in \overline{\Bi} \cap \overline{\Bj}}  \Bigl(f_\gamma\inp{\as{G_k,v_k}{i}(w)} - f_\gamma\inp{\as{G_k,v_k}{j}(w)}\Bigr) \nonumber \\
  &+ \sum_{v_k \in \overline{\Bi} \cap \Bj} f_\gamma\inp{\as{G_k,v_k}{i}(w)} \eqbreak
	\maybealign- \sum_{v_k \in \Bi \cap \overline{\Bj}} f_\gamma\inp{\as{G_k,v_k}{j}(w)}  \nonumber\\
  &-  \sum_{v_k \in \overline{\Bi} \cap \Bj } g_\gamma\inp{\bs{G_k,v_k}{j}(w)} \eqbreak
	\maybealign+  \sum_{v_k \in \Bi \cap \overline{\Bj}} g_\gamma\inp{\bs{G_k,v_k}{i}(w)}\nonumber\\
  &+ \sum_{v_k \in {\Bi} \cap {\Bj}} \Bigl(g_\gamma\inp{\bs{G_k,v_k}{i}(w)} -
    g_\gamma\inp{\bs{G_k,v_k}{j}(w)}\Bigr). \label{eq:w-rec}
\end{align}
Note that the same recurrence also applies when $w$ is replaced by $0$ (and hence $\gamma$ by $1$), except in that case the last three sums are $0$ (as, when $i$
is bad for $v_k$ in $G_k$, we have
$\bs{G_k,v_k}{i}(0) \defeq \Pr{G_k}{c(v_k)=i} = 0$):
\begin{align}
  \maybealign-\ln \ratio{G,u}{i,j}(0) \eqbreak
  =&\sum_{v_k \in \overline{\Bi} \cap \overline{\Bj}} \Bigl(f_1\inp{\as{G_k,v_k}{i}(0)} - f_1\inp{\as{G_k,v_k}{j}(0)}\Bigr)\nonumber \\
  &+ \sum_{v_k \in \overline{\Bi} \cap \Bj} f_1\inp{\as{G_k,v_k}{i}(0)} \eqbreak
  \maybealign- \sum_{v_k \in \Bi \cap \overline{\Bj}} f_1\inp{\as{G_k,v_k}{j}(0)}.
\end{align}
Further, by \MakeUppercase Consequence\nobreakspace \ref {lem:w-part-ind} of the induction hypothesis applied to the
graph $G_k$ at a vertex $v_k \in \Bi$ (respectively, $v_k \in \Bj$) we see that
$\abs{g_\gamma\inp{\bs{G_k,v_k}{i}(w)}} \leq 2\eps_w$ (respectively,
$\abs{g_\gamma\inp{\bs{G_k,v_k}{j}(w)}} \leq 2\eps_w$).  Thus, applying the triangle
inequality to the real part of the difference of the two recurrences, we get
\begin{align}
  \frac{1}{d} \abs{ \Re \ln \ratio{G,u}{i,j}(0)-\ln \ratio{G,u}{i,j}(w) }
  &\le 2\Delta \eps_w\nonumber\\
  &+\max\left\{ \max_{v_k \in \overline{\Bi} \cap \overline{\Bj}}
    \left\{\left\vert\inp{\Re f_\gamma\inp{\as{G_k,v_k}{i}(w)} - f_1\inp{\as{G_k,v_k}{i}(0)}}\right.\right.\right.\nonumber\\
  &\quad\qquad\qquad\qquad\qquad\left.\left.\left.
    - \inp{\Re f_\gamma\inp{\as{G_k,v_k}{j}(w)} - f_1\inp{\as{G_k,v_k}{j}(0)}}\right\vert\right\},   \right.\nonumber \\
  &\qquad\qquad \left.\max_{v_k \in \overline{\Bi} \cap \Bj} \set{\abs{\Re
    f_\gamma\inp{\as{G_k,v_k}{i}(w)} - f_1\inp{\as{G_k,v_k}{i}(0)} }}\right.,\nonumber\\
  &\qquad\qquad \left.\max_{v_k \in \overline{\Bj} \cap \Bi} \set{ \abs{\Re
    f_\gamma\inp{\as{G_k,v_k}{j}(w)}
    - f_1\inp{\as{G_k,v_k}{j}(0)} }}  \right\}. \label{eq:5}
\end{align}

In what follows, we let $v_k$ be the vertex that maximizes the above expression,
and $d_k$ be the number of unpinned neighbors of $v_k$ in $G_k$.  Before
proceeding with the analysis, we recall the observation above that the
graphs $G_k$ are unconflicted and satisfy the admissible list condition
$\mathcal{L}$.  Further, we note that $v_k$ is (i) marked in
  $G_k$ (this follows from the fact that $\mathcal{L}$ is hereditary); and
(ii) \nice{} in $G_k$ (this last fact follows from Lemma\nobreakspace \ref
{obv:good-to-nice} and the fact that $G$ satisfies the admissible list
  condition $\mathcal{L}$).  Thus, the preconditions of \MakeUppercase
Consequence\nobreakspace \ref {lem:real-part-ind} apply to the vertex $v_k$ in
graph $G_k$.  We now proceed with the analysis.

We first consider $v_k \in \overline{\Bi} \cap \Bj$.  Note that this implies
that $i \in \good{G_k, v_k}$.  Thus, the conditions of \MakeUppercase Consequence\nobreakspace \ref {lem:real-part-ind}
of the induction hypothesis instantiated on $G_k$ apply to $v_k$ with color $i$,
and we thus have from eq.\nobreakspace \textup {(\ref {eq:6})} that
\begin{align*}
  \maybealign\abs{\Re f_\gamma\inp{\as{G_k,v_k}{i}(w)} - f_1\inp{\as{G_k,v_k}{i}(0)} } \eqbreak
  \maybealign\leq \frac{1}{d_k+\eta}\abs{\Re\as{G_k,v_k}{i}(w) - \as{G_k,v_k}{i}(0)} + \eps_I + \eps_w,
\end{align*}
where $d_k$ is the number of unpinned neighbors of $v_k$ and $\eta \in [0.9, 1)$
is as in the statement of \MakeUppercase Consequence\nobreakspace \ref
{lem:real-part-ind}.  Applying item 2 of \MakeUppercase Consequence\nobreakspace
\ref {lem:approx-prob} (which, again, is applicable because
$i \in \good{G_k, v_k}$), we then have
$\abs{\Re\as{G_k,v_k}{i}(w) - \as{G_k,v_k}{i}(0)} \leq d_k(\eps_R + \eps_I) +
2\Delta\eps_w$, so that (recalling $\Delta \geq 3$ and $\eta \geq 0.9$,
notably for the case $d_k = 0$)
\begin{align}
  \maybealign\abs{\Re f_\gamma\inp{\as{G_k,v_k}{i}(w)} - f_1\inp{\as{G_k,v_k}{i}(0)} }\eqbreak
  \maybealign\leq \frac{d_k}{d_k + \eta}\eps_R + 2\eps_I + 3\Delta\eps_w.\label{eq:8}
\end{align}
By interchanging the roles of $i$ and $j$ in the above argument, we see that, for
$v_k \in \overline{\Bj} \cap \Bi$
\begin{align}
  \maybealign\abs{\Re f_\gamma\inp{\as{G_k,v_k}{j}(w)} - f_1\inp{\as{G_k,v_k}{j}(0)} }\eqbreak
  \maybealign\leq \frac{d_k}{d_k + \eta}\eps_R + 2\eps_I + 3\Delta\eps_w.\label{eq:9}
\end{align}
We now consider $v_k \in \overline\Bi \cap \overline\Bj$.  Note that both $i$ and $j$ are good for $v_k$ in $G_k$, so that
\begin{multline*}
	\abs{\inp{\Re f_\gamma\inp{\as{G_k,v_k}{i}(w)} - f_1\inp{\as{G_k,v_k}{i}(0)}} \right. \eqbreak \left.
    -\inp{\Re f_\gamma\inp{\as{G_k,v_k}{j}(w)} - f_1\inp{\as{G_k,v_k}{j}(0)}}}\\
    \leq \max_{i', j' \in \good{G_k, v_k}} \abs{\inp{\Re f_\gamma\inp{\as{G_k,v_k}{i'}(w)} - f_1\inp{\as{G_k,v_k}{i'}(0)}}  \right. \eqbreak \left.
    -\inp{\Re f_\gamma\inp{\as{G_k,v_k}{j'}(w)} - f_1\inp{\as{G_k,v_k}{j'}(0)}}}.
\end{multline*}
Now, for any color $s \in \good{G_k, v_k}$, \MakeUppercase Consequence\nobreakspace \ref {lem:real-part-ind} of the
induction hypothesis instantiated on $G_k$ and applied to $v_k$ and $s$ shows
that there exists a $C_s = C_{s, v_k, G_k} \in [0, 1/(d_k + \eta)]$ such that
\begin{equation}
	\abs{\Re f_\gamma\inp{\as{G_k,v_k}{s}(w)} - f_1\inp{\as{G_k,v_k}{s}(0)}
	  - C_s\inp{\Re \as{G_k,v_k}{s}(w) - \as{G_k,v_k}{s}(0)}} \leq \eps_I + \eps_w.
\end{equation}
Substituting this in the previous display shows that
\begin{align}
	&\abs{\inp{\Re f_\gamma\inp{\as{G_k,v_k}{i}(w)}  -
   f_1\inp{\as{G_k,v_k}{i}(0)}}
	-\inp{\Re f_\gamma\inp{\as{G_k,v_k}{j}(w)} - f_1\inp{\as{G_k,v_k}{j}(0)}}}\nonumber\\
	&\quad\leq \max_{i', j' \in \good{G_k, v_k}} \abs{ C_{i'}(\Re \as{G_k,v_k}{i'}(w) - \as{G_k,v_k}{i'}(0))
		  - C_{j'}(\Re \as{G_k,v_k}{j'}(w) - \as{G_k,v_k}{j'}(0)) } + 2\eps_I + 2\eps_w\nonumber\\
  &\quad= 2\eps_I + 2\eps_w + \max_{i', j' \in \good{G_k, v_k}} \abs{C_{i'}\Re \xi_{i'}
    - C_{j'}\Re \xi_{j'}}\nonumber\\
  &\quad= 2\eps_I + 2\eps_w + C_{s}\Re \xi_{s}
    - C_{t}\Re \xi_{t}\label{eq:14},
\end{align}
where $\xi_{l} \defeq \as{G_k,v_k}{l}(w) - \as{G_k,v_k}{l}(0)$ for
$l \in \good{G_k, v_k}$, and $s$ and $t$ are given by
\[
  s \defeq \argmax_{i' \in \good{G_k, v_k}} C_{i'}\Re \xi_{i'}\quad\text{ and }\quad
  t \defeq \argmin_{i' \in \good{G_k, v_k}} C_{i'}\Re \xi_{i'}.
\]
We now have the following two cases:

\medskip
\noindent \textup{\textbf{Case 1: $(\Re \xi_{s})\cdot(\Re \xi_{t}) \leq 0$.}} Recall
that $C_s, C_t$ are non-negative and lie in $[0, 1/(d_k + \eta)]$.  Thus, in
this case, we must have $\Re \xi_{s} \geq 0$ and $\Re \xi_{t} \leq 0$, so that
  \begin{align}
    C_{s}\Re \xi_{s} - C_{t}\Re \xi_{t}  \maybealign =  C_{s}\Re \xi_{s} + C_{t}\abs{\Re \xi_{t}} \eqbreak
    \maybealign \leq  \frac{\Re \xi_{s} + \abs{\Re \xi_{t}}}{d_k+\eta}
    = \frac{\abs{\Re \xi_{s}  - \Re \xi_{t}}}{d_k+\eta}.\label{eq:10}
  \end{align}
  Now, note that
  \begin{align*}
    \maybealign\Re \xi_{s}  - \Re \xi_{t}\eqbreak
    &= \Re \ln \frac{\prob{G_k,w}{c(v_k)=s}}{
      \prob{G_k}{c(v_k)=s}} - \Re \ln \frac{\prob{G_k,w}{c(v_k)=t}}{
      \prob{G_k}{c(v_k)=t}}\\
    &= \Re \ln \frac{\prob{G_k,w}{c(v_k)=s}}{\prob{G_k,w}{c(v_k)=t}}
      - \Re \ln \frac{\prob{G_k}{c(v_k)=s}}{\prob{G_k}{c(v_k)=t}}\\
    &= \Re \ln \ratio{G_k, v_k}{s,t}(w) - \ln \ratio{G_k, v_k}{s,t}(0).
  \end{align*}
  Note that all the logarithms in the above are well defined from
  Consequence\nobreakspace \ref {lem:approx-prob} of the induction hypothesis applied to $G_k$ and $v_k$
  (as $s, t \in \good{G_k, v_k})$.  Further, from items\nobreakspace \ref {ind-1} and\nobreakspace  \ref {ind-2} of the
  induction hypothesis, the last term is at most $d_k\eps_R$ in absolute value.
  Substituting this in eq.\nobreakspace \textup {(\ref {eq:10})}, we get
  \begin{equation}
    C_{s}\Re \xi_{s} - C_{t}\Re \xi_{t}  \leq \frac{d_k}{d_k + \eta}\eps_R. \label{eq:11}
  \end{equation}
  This concludes the analysis of Case 1.

  \medskip
  \noindent \textbf{\textbf{Case 2: $\Re \xi_{i'}$ for $i' \in \good{G_k, v_k}$ all
      have the same sign}.} Suppose first that $\Re \xi_{i'} \geq 0$ for all
  $i' \in \good{G_k, v_k}$.  Then, we have
  \begin{equation}
    0 \leq  C_{s}\Re \xi_{s} - C_{t}\Re \xi_{t} \leq \frac{\Re \xi_{s}}{d_k+\eta}
    \leq \frac{d_k \cdot \eps_R}{d_k + \eta} + \eps_I + 4\Delta\eps_w,\label{eq:12}
  \end{equation}
  where the last inequality follows from item 2 of \MakeUppercase Consequence\nobreakspace \ref{lem:approx-prob} of
  the induction hypothesis applied to $G_k$ at vertex $v_k$ with color $s$,
  which states that $\abs{\Re \xi_{s}} \leq d_k(\eps_R + \eps_I) + 4\Delta\eps_w$.
  Similarly, when $\Re \xi_{i'} \leq 0$ for all $i' \in \good{G_k, v_k}$, we
  have
  \begin{align}
    0 \leq  C_{s}\Re \xi_{s} - C_{t}\Re \xi_{t}
    &=  C_{t}|\Re \xi_{t}| -  C_{s}|\Re \xi_{s}|\nonumber\\&
    \leq \frac{\abs{\Re \xi_{t}}}{d_k+\eta}\nonumber\\
    &\leq \frac{d_k \cdot \eps_R}{d_k + \eta} + \eps_I + 4\Delta\eps_w,\label{eq:13}
  \end{align}
  where the last inequality follows from item 2 of \MakeUppercase Consequence\nobreakspace \ref{lem:approx-prob} of
  the induction hypothesis applied to $G_k$ at vertex $v_k$ with color $t$,
  which states that $\abs{\Re \xi_{t}} \leq d_k(\eps_R + \eps_I) +
  4\Delta\eps_w$.  This concludes the analysis of Case 2.

\medskip
Now, substituting eqs.\nobreakspace  \textup {(\ref {eq:11})} to\nobreakspace  \textup {(\ref {eq:13})}  into eq.\nobreakspace \textup {(\ref {eq:14})}, we get
\begin{multline}
  \label{eq:15}
  \abs{\inp{\Re f_\gamma\inp{\as{G_k,v_k}{i}(w)} - f_1\inp{\as{G_k,v_k}{i}(0)}} \right. \eqbreakno \left.
    -\inp{\Re f_\gamma\inp{\as{G_k,v_k}{j}(w)} - f_1\inp{\as{G_k,v_k}{j}(0)}}}\\
  \leq \frac{d_k}{d_k + \eta}\eps_R + 3\eps_I + 5\Delta\eps_w.
\end{multline}

Substituting eqs.\nobreakspace \textup {(\ref {eq:8})},  \textup {(\ref {eq:9})} and\nobreakspace  \textup {(\ref {eq:15})} into eq.\nobreakspace \textup {(\ref {eq:5})}, we get
\begin{align*}
  \frac{1}{d} \maybealign\abs{ \Re \ln \ratio{G,u}{i,j}(w)-\ln \ratio{G,u}{i,j}(0) } \eqbreak
  \maybealign\leq \frac{d_k \cdot \eps_R}{d_k + \eta} + 3\eps_I + 7\Delta\eps_w 
  < \eps_R,
\end{align*}
where the last inequality follows since
$\eta\eps_R > (\Delta + 1)(3\eps_I + 7\Delta\eps_w)$ (recalling that
$0 \leq d_k \leq \Delta$ and $\eta \in [0.9, 1)$). This verifies item\nobreakspace \ref {ind-2} of the induction hypothesis.

For item\nobreakspace \ref {ind-3}, we consider the imaginary part of eq.\nobreakspace \textup {(\ref {eq:w-rec})}.  As in the
derivation of eq.\nobreakspace \textup {(\ref {eq:5})}, we use the fact that the induction hypothesis applied
to the graph $G_k$ at the vertex $v_k \in \Bi$ (respectively, $v_k \in \Bj$)
implies that $\abs{g_\gamma\inp{\bs{G_k,v_k}{i}(w)}} \leq 2\eps_w$
(respectively, $\abs{g_\gamma\inp{\bs{G_k,v_k}{j}(w)}} \leq 2\eps_w$).  This yields
\begin{align}
  \frac{1}{d} \abs{ \Im \ln \ratio{G,u}{i,j}(w) }
  &\le 2\Delta \eps_w\nonumber\\
  &\quad+\max\left\{
    \max_{v_k \in \overline{\Bi} \cap \overline{\Bj}}
    \abs{\Im f_\gamma\inp{\as{G_k,v_k}{i}(w)} - \Im
    f_\gamma\inp{\as{G_k,v_k}{j}(w)}},
    \right.\nonumber \\
  &\quad\qquad\qquad \left.\max_{v_k \in \overline{\Bi} \cap \Bj} \abs{\Im
    f_\gamma\inp{\as{G_k,v_k}{i}(w)}}, \max_{v_k \in \overline{\Bj} \cap \Bi} \abs{\Im
    f_\gamma\inp{\as{G_k,v_k}{j}(w)}}\right\}.\label{eq:16}
\end{align}
Again, let $v_k$ be the vertex that maximizes the above expression, and $d_k$ be the number of unpinned neighbors of $v_k$ in $G_k$.
We first consider $v_k \in \overline{\Bi} \cap \overline{\Bj}$. Applying
eq.\nobreakspace \textup {(\ref {eq:7})} of Consequence\nobreakspace \ref {lem:real-part-ind} of the induction hypothesis to the graph
$G_k$ at vertex $v_k$ with colors $i, j \in \good{G_k, v_k}$ gives
\begin{equation}
  \label{eq:17}
  \abs{\Im f_\gamma\inp{\as{G_k,v_k}{i}(w)} - \Im
    f_\gamma\inp{\as{G_k,v_k}{j}(w)}} \eqbreakno
    \leq \frac{d_k}{d_k + \eta}\eps_I + 6\Delta\eps_w.
\end{equation}
Now consider $v_k \in \overline{\Bi} \cap {\Bj}$.  For this case, eq.\nobreakspace \textup {(\ref {eq:18})}
of Consequence\nobreakspace \ref {lem:real-part-ind} of the induction hypothesis applied to $G_k$ at
vertex $v_k$ with color $i \in \good{G_k, v_k}$ gives
\begin{equation}
  \label{eq:23}
  \abs{\Im f_\gamma\inp{\as{G_k,v_k}{i}(w)}} \leq \frac{d_k}{d_k + \eta}\eps_I + 5\Delta\eps_w.
\end{equation}
Similarly, for $v_k \in \overline{\Bj} \cap {\Bi}$, eq.\nobreakspace \textup {(\ref {eq:18})}
of Consequence\nobreakspace \ref {lem:real-part-ind} of the induction hypothesis applied to $G_k$ at
vertex $v_k$ with color $j \in \good{G_k, v_k}$ gives
\begin{equation}
  \label{eq:24}
  \abs{\Im f_\gamma\inp{\as{G_k,v_k}{j}(w)}} \leq \frac{d_k}{d_k + \eta}\eps_I + 5\Delta\eps_w.
\end{equation}
Substituting eqs.\nobreakspace  \textup {(\ref {eq:17})} to\nobreakspace  \textup {(\ref {eq:24})}  into eq.\nobreakspace \textup {(\ref {eq:16})} we have
\begin{displaymath}
  \frac{1}{d} \abs{ \Im \ln \ratio{G,u}{i,j}(w) } \leq \frac{d_k}{d_k + \eta}\eps_I +
  8\Delta\eps_w < \eps_I,
\end{displaymath}
where the last inequality holds since $\eta\eps_I > 8(\Delta + 1)\Delta\eps_w$
(recalling that $0 \leq d_k \leq \Delta$ and
$\eta \in [0.9, 1)$). This completes the proof
of item\nobreakspace \ref {ind-3} of the induction hypothesis.

Finally, we prove~item\nobreakspace \ref {ind-4}.  Since $i \not\in \good{u}$,
there exist $n_i >0$ neighbors of $u$ that are pinned to color $i$.  Let $H$ be
the graph obtained by removing these neighbors of $u$ from $G$.  Then, $H$ is an
unconflicted graph with the \emph{same} number of unpinned vertices as $G$, which also
satisfies the admissible list condition $\mathcal{L}$ (since
  $\mathcal{L}$ is hereditary). Further, $u$ remains marked in $H$,
  and $H$ further satisfies $i, j \in \good{H, u}$.  We can therefore apply the
already proved items\nobreakspace \ref {ind-0} to\nobreakspace \ref {ind-2} to
$H$ to conclude that
\begin{equation}
  \abs{\ratio{H}{i,j}(w)} \leq \abs{\ratio{H}{i,j}(0)}e^{d\eps_R}.\label{eq:25}
\end{equation}
Now, since $i, j \in \good{H, u}$, we can apply the recurrence of
\MakeUppercase Lemma\nobreakspace \ref {thm:recurrence} in the same way as in the derivation of eq.\nobreakspace \textup {(\ref {eq:26})} above
to get
\begin{equation}
  \label{eq:27}
  \ratio{H,u}{i,j}(w) = \prod_{k=1}^{\deg_H(u)} \frac{1 - 
      \prob{H_{k}^{(i,j)},w}{ c(v_k) = i}  }{ 1 -  \prob{H_{k}^{(i,j)},w}{c(v_k) = j} } ,
\end{equation}
where, for the reasons described in the discussion following eq.\nobreakspace \textup {(\ref {eq:26})}, the
product can be restricted to unpinned neighbors of $u$ in $H$.  Renaming these
unpinned neighbors as $v_1, v_2, \dots, v_d$, we then have
\begin{equation}
  0 \leq \ratio{H}{i,j}(0)
  = \prod_{k=1}^{d} \frac{\inp{1 -  \prob{H_{k}}{ c(v_k) = i}}
  }{\inp{1 -  \prob{H_{k}}{c(v_k) = j}}
  },\label{eq:20}
\end{equation}
where, as before, $H_k \defeq H_k^{(i,j)}$.  Now, as observed above $H$
satisfies the admissible list condition $\mathcal{L}$.  Thus, for $1 \leq k \leq d$, $v_k$ is
\nice{} in $H_k$ (Lemma\nobreakspace \ref {obv:good-to-nice}), and hence, $\prob{H_k}{c(v_k) = j} \leq \frac{1}{d_k + 2}$ for $1 \leq k \leq d$, where
$d_k \geq 0$ is the number of unpinned neighbors of $v_k$ in $H_k$.  We then
have
\begin{align*}
  0 \leq \ratio{H}{i,j}(0)
  &= \prod_{k=1}^{d} \frac{\inp{1 -  \prob{H_{k}}{ c(v_k) = i}}
    }{\inp{1 -  \prob{H_{k}}{c(v_k) = j}}
    }
    \eqbreak
    \maybealign\leq \prod_{k=1}^d \frac{1}{1 - \frac{1}{d_k + 2}} 
    = \prod_{k=1}^d \frac{d_k + 2}{d_k + 1} 
    \leq 2^\Delta.
\end{align*}
(As an aside, we note that one could get a better bound under the slightly
stronger assumption of uniformly large list sizes considered in
\protect \MakeUppercase {R}emark\nobreakspace \ref {rem:nicer}.  Under the conditions of that remark, we have
$\prob{H_k}{c(v_k) = j} < \min\set{\frac{4}{3\Delta},1}$, so that the above
upper bound can be improved to $\ratio{H}{i,j}(0) \le e^4$ for $\Delta> 1$.)

Combining the estimate with eq.\nobreakspace \textup {(\ref {eq:25})}, we get
$\abs{\ratio{H}{i,j}(w)} \leq 5\cdot2^{\Delta}$ since
$d\eps_R \leq 1/2$. Now note that since
$j \in \good{G, u}$,
\begin{displaymath}
  \Z{G,u}{i}(w) = w^{n_i}\Z{H, u}{i}(w), \text{\ \ and\ \ }
  \Z{G,u}{j}(w) = \Z{H, u}{j}(w),
\end{displaymath}
so that
$\abs{\ratio{G,u}{i,j}(w)} = \abs{w}^{n_i}\abs{\ratio{H,u}{i,j}(w)} \leq
5\cdot2^{\Delta}\cdot \abs{w}^{n_i}$.  The latter is at most $\eps_w$ whenever
$\abs{w} \leq 0.2 \eps_w/2^\Delta$.  This proves item\nobreakspace \ref {ind-4}, and also
completes the inductive proof of \MakeUppercase Lemma\nobreakspace \ref {lem:origin-induction}.  (Note also that
using the stronger upper bound above under the condition of uniformly large list
sizes, we can in fact relax the requirement further to
$\abs{w} \leq \eps_w/(300\Delta)$.) \qed

We conclude this section by using Lemma\nobreakspace \ref {lem:origin-induction} to prove
Theorem\nobreakspace \ref {thm:main-origin-restated}.
\begin{proof}[Proof of Theorem\nobreakspace \ref {thm:main-origin-restated}]
  Let $G$ be a graph of maximum degree $\Delta$ satisfying the admissible list condition $\mathcal{L}$.
  Since $G$ has no pinned vertices, $G$ is unconflicted.  Let $u$ be an unpinned
  vertex that is marked in $G$.  By \MakeUppercase Consequence\nobreakspace \ref
  {lem:sum-lowerbounds} of the induction hypothesis (which we proved in
  Lemma\nobreakspace \ref {lem:origin-induction}), we then have $Z_w(G) \neq 0$
  provided $\nu_w \leq 0.2\eps_w/2^\Delta$.

  Furthermore, as discussed above, under a slightly stronger assumption of
  uniformly large list sizes considered in \protect \MakeUppercase {R}emark\nobreakspace \ref {rem:nicer}, $\nu_w$ can be
  chosen to be $\eps_w/(300\Delta)$.
\end{proof}

\section{Zero-free region around the interval $(0,1]$}
\label{sec:induction-interval}
In this section, we consider the case of $w$ close to $[0, 1]$ but bounded away
from $0$.  In particular, we prove the following theorem, which complements
Theorem\nobreakspace \ref {thm:main-origin-restated}.
\begin{theorem}
  \label{thm:main-interval-restated}
  Fix a positive integer $\Delta$ and an admissible list condition
    $\mathcal{L}$.  Let $\nu_w = \nu_w(\Delta)$ be as in \MakeUppercase
  Theorem\nobreakspace \ref {thm:main-origin-restated}.  Then, for any $w$
  satisfying
  \begin{equation}
    \Re w \in [\nu_w/2, 1 + \nu_w^2/8] \qquad\text{and}\qquad \abs{\Im{w}} \leq \nu_w^2/8,\label{eq:w-bounds}
  \end{equation}
  and any graph $G$ of maximum degree $\Delta$ which satisfies
    $\mathcal{L}$, we have $Z_G(w) \neq 0$.
\end{theorem}
(Here, we recall that as described in the discussion following
Theorem\nobreakspace \ref {thm:main-origin-restated}, $\nu_w$ can be chosen to be
$\eps_w/(300\Delta)$ when the uniformly large list size condition of
\protect \MakeUppercase {R}emark\nobreakspace \ref {rem:nicer} is satisfied.  However, as in that theorem, in the case of
general list coloring, one chooses $\nu_w = 0.2\eps_w/2^\Delta$.)

For $w$ as in eq.\nobreakspace \textup {(\ref {eq:w-bounds})}, we define $\wt$ to be the point on the
interval $[0, 1]$ which is closest to $w$.  Thus
\begin{displaymath}
  \wt \defeq
  \begin{cases}
    \Re w & \text{ when $\Re w \in [\nu_w/2, 1]$;}\\
    1 & \text{ when $\Re w \in (1, 1 + \nu_w^2/8]$.}
  \end{cases}
\end{displaymath}
We also define, in analogy with the last section, $\gamma \defeq 1 - w$ and
$\gammat \defeq 1 - \wt$.  We record a few properties of these quantities in the
following observation.
\begin{observation}
  With $w, \gamma, \wt$ and $\gammat$ as above, we have
  \begin{enumerate}
  \item $0 \leq \gammat, \abs{\gamma} < 1$.
  \item $|\ln w - \ln \wt| \leq \nu_w$.
  \end{enumerate}\label{obv:wg-prop}
\end{observation}
\begin{proof}
  We have $\tilde{\gamma} \in [0, 1 - \nu_w/2]$,
  $\Re \gamma \in [-\nu_w^2/8, 1 - \nu_w/2]$ and $\abs{\Im \gamma} \leq \nu_w^2/8$.
  Since $\nu_w \leq 0.01$, these bounds taken together imply item 1.  We also
  have $0 \leq \tw \leq \abs{w} \leq \tw + \nu_w^2/4$ and $\tw \geq
  \nu_w/2$. Thus
  \begin{displaymath}
    0 \leq \Re (\ln w - \ln \tilde{w}) = \ln \frac{\abs{w}}{\wt} \leq \ln\inp{1
      + \frac{\nu_w^2}{4\tw}} \leq \frac{\nu_w}{2}.
  \end{displaymath}
  Similarly, $\Im (\ln w - \ln \tilde{w}) = \Im \ln w = \arg w$, so that
  \begin{displaymath}
    \abs{\Im (\ln w - \ln \tilde{w})} \leq \abs{\arg w} \leq \frac{\abs{\Im
        w}}{\Re w} \leq \frac{\nu_w}{4}.
  \end{displaymath}
  Together, the above two bounds imply item 2.
\end{proof}

In analogous fashion to the proof of Theorem\nobreakspace \ref {thm:main-origin-restated}, we would
like to show that $\ratio{G,u}{i,j}(w) \approx \ratio{G,u}{i,j}(\wt)$
independent of the size of $G$. (Note that for positive $\wt$,
$\ratio{G,u}{i,j}(\wt)$ is a well defined positive real number for any graph.)
To this end, we will prove the following analog of Lemma\nobreakspace \ref {lem:origin-induction}
via an induction on the number of unpinned vertices in $G$.  The induction is
very similar in structure to that used in the proof of
Lemma\nobreakspace \ref {lem:origin-induction}, except that the fact that $w$ has strictly positive
real part allows us to simplify several aspects of the proof.  In particular, we
do not need to consider good and bad colors separately, and do not require the
underlying graphs to be unconflicted.

As in the previous section, we assume that all graphs in this section have
maximum degree at most $\Delta \geq 1$, and define the quantities
$\eps_w, \eps_R, \eps_I$ in terms of $\Delta$ using eq.\nobreakspace \textup
{(\ref {eq:def-eps})}.  We again fix an admissible
  list condition $\mathcal{L}$ throughout this section.

\begin{lemma}
  \label{lem:interval-induction}
  Let $G$ be a graph of maximum degree $\Delta$ satisfying the admissible list
  condition $\mathcal{L}$ and let $u$ be any marked unpinned vertex in $G$.  Then, the
  following are true (here, $\eps_w, \eps_I, \eps_R$ are as defined in
  eq.\nobreakspace \textup {(\ref {eq:def-eps})}):
  \begin{enumerate}\item \label{wint-ind-0} For $i \in L(u)$, $\abs{ \Z{G,u}{i} (w)} > 0 $.
  \item \label{wint-ind-1} For $i,j\in L(u)$, if $u$ has all neighbors pinned,
    then $|\ln \ratio{G,u}{i,j}(w) - \ln \ratio{G,u}{i,j}(\wt)| < \eps_w$.
  \item \label{wint-ind-2} For $i,j\in L(u)$, if $u$ has $d \geq 1$ unpinned
    neighbors, then
    \[
      \frac{1}{d}\abs{\Re \ln \ratio{G,u}{i,j}(w) - \Re \ln
        \ratio{G,u}{i,j}(\wt)} < \eps_R.
    \]
  \item \label{wint-ind-3} For any $i,j\in L(u)$, if $u$ has $d \geq 1$ unpinned
    neighbors, then
    $\frac{1}{d} \abs{\Im \ln \ratio{G,u}{i,j}(w)} < \eps_I$.
  \end{enumerate}
\end{lemma}
We will refer to items\nobreakspace  \ref {wint-ind-0} to\nobreakspace  \ref {wint-ind-3}  as ``items
of the induction hypothesis''.  The rest of this section is devoted to the proof
of this lemma via an induction on the number of unpinned vertices in $G$.

We begin by verifying that the induction hypothesis holds in the base case
when $u$ is the only unpinned vertex in a graph $G$.  In this case,
items\nobreakspace \ref {wint-ind-2} and\nobreakspace  \ref {wint-ind-3} are vacuously true since $u$ has no unpinned
neighbors.  Since all neighbors of $u$ in $G$ are pinned, the fact that all
pinned vertices have degree at most one implies that $G$ can be decomposed into
two disjoint components $G_1$ and $G_2$, where $G_1$ consists of $u$ and its
pinned neighbors, while $G_2$ consists of a disjoint union of edges with pinned
end-points.  Let $m$ be the number of conflicted edges on $G_2$, and let $n_k$
denote the number of neighbors of $u$ pinned to color $k$.  We then have
$\Z{G, u}{k}(x) = x^{n_k} Z_{G_2}(x) = x^{n_k+m}$ for all $x \in \C$.  This
already proves item\nobreakspace \ref {wint-ind-0} since $w, \wt \neq 0$.  Item\nobreakspace \ref {wint-ind-1}
follows via the following computation (which uses item 2 of \MakeUppercase Observation\nobreakspace \ref {obv:wg-prop}):
\begin{align*}
  |\ln \ratio{G,u}{i,j}(w) - \ln \ratio{G,u}{i,j}(\wt)| &= |n_i -
  n_j|\cdot|\ln w - \ln \tw| \eqbreak
  \maybealign\leq \Delta \nu_w < \eps_w.
\end{align*}
We now derive some consequences of the above induction hypothesis that will be
helpful in carrying out the induction.  Throughout, we fix the graph $G$ and the
vertex $u$ as in the statement of Lemma~\ref{lem:interval-induction}.
\begin{consequence}
  \label{lem:wint-sum-lowerbounds}
  If $|L(u)| \geq 1$, then $\abs{Z_G(w)} > 0.$
\end{consequence}
\begin{proof}
  Note that $Z_G(w) = \sum_{i \in L(u)} \Z{G,u}{i} (w)$.  From
  item\nobreakspace \ref {wint-ind-3}, we see that the angle between the complex numbers
  $\Z{G,u}{i} (w)$ and $\Z{G,u}{j} (w)$, for all $i, j \in L(u)$, is at most
  $d\eps_I$.  Applying \MakeUppercase Lemma\nobreakspace \ref {lem:geometric-proj} we then have
  \begin{align*}
   \Bigl|\sum_{i \in L(u)} \Z{G,u}{i} (w)\Bigr|
    &\ge \abs{L(u)} \cos \frac{d\eps_I}{2}
      \cdot \min_{i \in \good{u}}
      \abs{\Z{G,u}{i}(w)} \eqbreak
      \maybealign\ge  0.9\min_{i \in \good{u}} \abs{\Z{G,u}{i}(w)},
  \end{align*}
  when $|L(u)| \ge 1$ and $d\eps_I \leq 0.01$.  This last quantity is positive
  from item\nobreakspace \ref {wint-ind-0}.
\end{proof}

\begin{consequence}
  For all $\eps_R, \eps_I, \eps_w$ small enough such that $\eps_I \leq \eps_R$
  and $\eps_w\leq 0.01\eps_I$, the pseudo-probabilities approximate the real
  probabilities in the following sense: for any $j\in L(u)$,
    \begin{align*}
      \abs{\Im \ln
      \frac{\prob{G,w}{c(u) = j}}{\prob{G,\wt}{c(u) = j}}}
      &= \abs{\Im \ln \prob{G,w}{c(u) = j}} \eqbreak
      \maybealign\leq d\eps_I + 2\Delta\eps_w; \\ \abs{\Re \ln
      \frac{\prob{G,w}{c(u) = j}}{\prob{G,\wt}{c(u) = j}}}
      &\leq d \eps_R + d\eps_I + 2\Delta\eps_w,
    \end{align*}
    where $d$ is the number of unpinned neighbors of $u$ in $G$.
  \label{lem:wint-approx-prob}
\end{consequence}
\begin{proof}
  Using items\nobreakspace  \ref {wint-ind-1} to\nobreakspace  \ref {wint-ind-3}  of the induction hypothesis,
  there exist complex numbers $\xi_i$ (for all $i \in \good{u}$) satisfying
  $\abs{ \Re \xi_i} \le d \eps_R + \eps_w$ and
  $\abs{\Im \xi_i} \le d\eps_I + \eps_w$ such that
  \begin{align}
    \frac{\prob{G,\tilde{w}}{c(u) = j}}{\prob{G,w}{c(u) = j}}
    \maybealign= \prob{G,\tilde{w}}{c(u) = j}\sum_{i \in L(u)} \frac{\Z{G,u}{i} (w)}{\Z{G,u}{j} (w)} \eqbreak
    \maybealign=  \prob{G,\tilde{w}}{c(u) = j}\sum_{i \in L(u)} \frac{\Z{G,u}{i} (\wt)}{\Z{G,u}{j} (\wt)} e^{\xi_i}.\label{eq:w-ind-26}
  \end{align}
Now, note that
  $\sum_{i \in L\inp{u}} \frac{\Z{G,u}{i} (\tw)}{\Z{G,u}{j} (\tw)} =
  \frac{1}{\prob{G,\tw}{c(u) = j}}$, so that the sum above is a convex
  combination of the $\exp(\xi_i)$.  From the bounds on the real and imaginary
  parts of the $\xi_i$ quoted above, by a calculation similar to that in~eq.\nobreakspace \textup {(\ref {eq:ind-origin-1})}, we also have (when $\eps_I, \eps_w \leq 0.01/\Delta$)
  \begin{align*}
	  \Re{e^{\xi_i}} \maybealign\in (e^{-d\eps_R-\eps_w} - (d\eps_I + \eps_w)^2,\;
    e^{d\eps_R + \eps_w})\text{, and } \eqbreak
    |\arg {e^{\xi_i}}| \maybealign\leq
    d\eps_I + \eps_w.\end{align*}
  The above will therefore be true also for any convex combination of the
  $e^{\xi_i}$, in particular the one in eq.\nobreakspace \textup {(\ref {eq:w-ind-26})}.  We therefore have,
  for $C \defeq \frac{\prob{G,\tilde{w}}{c(u) = j}}{\prob{G,w}{c(u) = j}}$,
  \begin{align*}
\Re C
    &\in \big(e^{-d\eps_R - \eps_w} - (d\eps_I + \eps_w)^2,\; e^{d\eps_R + \eps_w}\big),\\ |\arg C|
    &\leq d\eps_I + \eps_w.
  \end{align*}
  Now recall that for $\abs{\theta} \leq \pi/4$, we have
  $-\theta^2 \leq \ln \cos \theta \leq -\theta^2/2$. Thus, using the values of
  $\eps_w, \eps_I$ and $\eps_R$, we have
  \begin{align*}
    |\Re \ln C| &\leq d\eps_R + d\eps_I + 2\Delta\eps_w,
                   \text{ and }\\
    |\Im \ln C| &\leq d\eps_I + \eps_w\qedhere.
  \end{align*}
\end{proof}
As before we define $\as{G,u}{i}(w) = \ln \prob{G,w}{c(u)=i}$, and recall the
definition of the function $f_\gamma(x) \defeq -\ln ( 1 - \gamma e^x)$.
\begin{consequence}
  \label{lem:wint-real-part-ind} There exists a positive constant
  $\eta \in [0.9, 1)$ so that the following is true. Let $d$ be the number of
  unpinned neighbors of $u$.  Assume further that the vertex $u$ is \nice{} in
  $G$. Then, for any colors $i, j\in L(u)$, there exist a real number
  $C = C_{G,u,i} \in [0, \frac{1}{d + \eta}]$ such that
  \begin{align}
    \abs{ \Re f_\gamma(\as{G,u}{i}(w)) - f_\gammat(\as{G,u}{i}(\wt))
    - C\cdot \Re \inp{\as{G,u}{i}(w) - \as{G,u}{i}(\wt)}}
    &\le \eps_I + \eps_w;\label{eq:wint-6}\\
    \abs{ \Im  f_\gamma(\as{G,u}{i}(w)) - \Im f_\gamma(\as{G,u}{j}(w))}\eqbreak
    &\le  \frac{1}{d + \eta} \cdot (d\eps_I + 4 \Delta \eps_w) +
      2\eps_w.\label{eq:wint-7}\end{align}
\end{consequence}
\begin{proof}
  Since $u$ is \nice{} in $G$, the bound
  $\prob{G,\tw}{c(u) = k} \leq \frac{1}{d+2}$ (for any $k \in
  L(u)$) applies.  Combining them with Consequence\nobreakspace \ref {lem:wint-approx-prob} we see that
  $\as{G,u}{i}(w), \as{G,u}{i}(\wt), \as{G,u}{j}(w), \as{G,u}{j}(\wt)$ lie in a
  domain $D$ as described in \MakeUppercase Lemma\nobreakspace \ref {obv:f-props-int}, with the parameters $\zeta$
  and $\tau$ in that lemma chosen as
  \begin{equation*}
    \begin{aligned}
      \zeta &= \ln(d+2) - d\eps_R - d\eps_I - 2\Delta\eps_w\,\text{, and }\\
      \tau &= d\eps_I + 2\Delta\eps_w.
    \end{aligned}
  \end{equation*}
  Here, for the bound on $\zeta$, we use the fact that for $k \in L(u)$,
  $\prob{G,\tw}{c(u) = k} \leq \frac{1}{d+2}$, since $u$ is \nice{} in
  $G$. As in the proof of \MakeUppercase Consequence\nobreakspace \ref {lem:real-part-ind}, we use the values of
  $\eps_w, \eps_I, \eps_R$ to verify that the condition $\tau < 1/2$ and
  $\tau^2 + e^{-\zeta} < 1$ are satisfied, so that item 1 of
  Lemma\nobreakspace \ref {obv:f-props-int} applies (with the parameter $\kappa$ therein set to
  $\gammat$) and further that $\rho_R$ and $\rho_I$ as set there satisfy
  $\rho_R \le \frac{1}{d + \eta}$ and $\rho_I < 3\eps_I$, with $\eta =
  0.94$. Using~Lemma\nobreakspace \ref {thm:signed-mean} followed by the bound on $\eps_w$, we then
  have
  \begin{equation}
	  \abs{ \Re f_\gammat(\as{G,u}{i}(w)) - f_\gammat(\as{G,u}{i}(\tw))
		- C \cdot \Re \inp{ \as{G,u}{i}(w) - \as{G,u}{i}(\tw)}} \le 3\eps_I(d\eps_I + 2\Delta\eps_w)
		\leq 4d\eps_I^2 \leq \eps_I,\label{eq:wint-19}
  \end{equation}
  for an appropriate non-negative $C \leq 1/(d+\eta)$.  
  This is almost~eq.\nobreakspace \textup {(\ref {eq:wint-6})}, whose difference will be handled later.

  Similarly, applying
  Lemma\nobreakspace \ref {thm:signed-mean} to the imaginary part we have
  \begin{multline}
    \abs{ \Im \Big(f_\gammat(\as{G,u}{i}(w)) - f_\gammat(\as{G,u}{j}(w))\Big)}\\
    \le \rho_R\cdot \max\set{\abs{\Im \inp{\as{G,u}{i}(w) - \as{G,u}{j}(w)}} ,\right. \eqbreakno\left.
      \abs{\Im \as{G,u}{i}(w)}, \abs{\Im \as {G,u}{j}(w)}},\label{eq:wint-21}
  \end{multline}
  where, as noted above, $\rho_R \leq \frac{1}{d + \eta}$.  Now, note that the
  first term in the above maximum is less than $d\eps_I + \eps_w$ by
  items\nobreakspace \ref {wint-ind-1} and\nobreakspace  \ref {wint-ind-3} of the induction hypothesis, while the other two
  are at most $d\eps_I + 2\Delta\eps_w$ from item 2 of
  \MakeUppercase Consequence\nobreakspace \ref {lem:wint-approx-prob}.

Finally, we use item 2 of \MakeUppercase Lemma\nobreakspace \ref {obv:f-props-int} with the parameter $\kappa'$
  therein set to $\gamma$.  To this end, we note that
  $\abs{\gamma - \gammat} \leq \eps_w$, and that with the fixed values of
  $\eps_w, \eps_R$, and $\eps_I$, the condition $(1 + \eps_w) < e^\zeta$ is
  satisfied, so that the item applies.  Using the item, we then see that for any
  $z \in D$,
  \[
	\abs{ f_\gamma(z) - f_\gammat(z) } \le \eps_w.
  \]
  Thus, the quantities
  $|\Re f_\gamma(\as{G,u}{i}(w)) - \Re f_\gammat(\as{G,u}{i}(w))|$,
  $|\Im f_\gamma(\as{G,u}{i}(w)) - \Im f_\gammat(\as{G,u}{i}(w))|$,
  $|\Im f_\gamma(\as{G,u}{j}(w)) - \Im f_\gammat(\as{G,u}{j}(w))|$, and
  $|\Im f_\gamma(\as{G,u}{j}(w)) - \Im f_\gammat(\as{G,u}{j}(w))|$ are all at
  most $\eps_w$. The desired bounds now follow from the triangle inequality and
  the bounds in eqs.\nobreakspace \textup {(\ref {eq:wint-19})} and\nobreakspace  \textup {(\ref {eq:wint-21})}.
\end{proof}

\subsection*{Inductive proof of Lemma\nobreakspace \ref {lem:interval-induction}} We are now ready
to see the inductive proof of Lemma\nobreakspace \ref {lem:interval-induction}; recall that the
base case was already established immediately following the statement of the lemma. Let $G$
be any graph which satisfies the admissible list condition $\mathcal{L}$ and has at least two
unpinned vertices.
We first prove induction~item\nobreakspace \ref {wint-ind-0} for any
marked unpinned vertex $u$ in $G$. Consider the graph $G'$
obtained from $G$ by pinning vertex $u$ to color $i$.  Note that by the
definition of the pinning operation, $Z_{G, u}^{i}(w) = Z_{G'}(w)$.  Further,
since $\mathcal{L}$ is hereditary (because it is admissible), the graph
  $G'$ also satisfies $\mathcal{L}$, and has one fewer unpinned vertex than
$G$.  Thus, from Consequence\nobreakspace \ref {lem:wint-sum-lowerbounds} of the
induction hypothesis applied to $G'$, we have that
$\abs{\Z{G, u}{i}(w)} = \abs{Z_{G'}(w)} > 0$.

We now consider~item\nobreakspace \ref {wint-ind-1}.  When all neighbors of $u$
in $G$ are pinned, the fact that all pinned vertices have degree at most one
implies that $G$ can be decomposed into two disjoint components $G_1$ and $G_2$,
where $G_1$ consists of $u$ and its pinned neighbors, while $G_2$ has one fewer
unpinned vertex than $G$.  Note also that $G_2$, being a connected
  component of $G$, also satisfies the admissible list condition $\mathcal{L}$
  (since $\mathcal{L}$ is hereditary). Thus, from Consequence\nobreakspace \ref
{lem:wint-sum-lowerbounds} of the induction hypothesis applied to $G_2$, we have
that $Z_{G_2}(w)$ and $Z_{G_2}(\wt)$ are both non-zero.  Let $n_k$ be the number
of neighbors of $u$ pinned to color $k$.  Now, since $G_1$ and $G_2$ are disjoint
components, we get $\Z{G, u}{k}(x) = x^{n_k}Z_{G_2}(x)$ for all $k \in L(u)$
and all $x \in \C$.  It therefore follows that
\begin{align*}
  |\ln \ratio{G,u}{i,j}(w) - \ln \ratio{G,u}{i,j}(\wt)| \maybealign= |n_i - n_j|\cdot|\ln w - \ln \tw| \eqbreak
  \maybealign\leq \Delta \nu_w < \eps_w.
\end{align*}

We now consider~items\nobreakspace \ref {wint-ind-2} and\nobreakspace  \ref {wint-ind-3}.  Recall that by~Lemma\nobreakspace \ref {thm:recurrence}, we
have
\[
  \ratio{G,u}{i,j}(w) = \prod_{k=1}^{\deg_G(u)} \frac{\inp{1 - \gamma
      \prob{\G{k}{i,j},w}{ c(v_k) = i} } }{ \inp{1 - \gamma \prob{\G{k}{i,j},w}{
        c(v_k) = j} } }.
\]
Without loss of generality, we relabel the unpinned neighbors of $u$ as
$v_1, v_2, \dots, v_d$.  As before, for simplicity we write
$G_k \defeq \G{k}{i,j}$.  Note that each $G_k$ has exactly one fewer
  unpinned vertex than $G$ and satisfies $\mathcal{L}$ (since $\mathcal{L}$ is
  hereditary).  Further, the vertex $v_k$ is marked in $G_k$ (as $u$
  was marked in $G$).  Thus, the induction hypothesis applies to each
  $G_k$ at the vertex $v_k$.  Now, let $n_k$ be the number of neighbors of $u$
pinned to color $k$. Recalling that $1 - \gamma = w$, we can then simplify the
above recurrence to
\[
  \ratio{G,u}{i,j}(w) = w^{n_i - n_j} \prod_{k=1}^{d} \frac{\inp{1 - \gamma
      \prob{\G{k}{i,j},w}{ c(v_k) = i} } }{ \inp{1 - \gamma \prob{\G{k}{i,j},w}{
        c(v_k) = j} } }.
\]
Now, as before, for $s \in L(v_k)$ we define
$\as{G_k,v_k}{s}(w) \defeq \ln \prob{G_k,w}{c(v_k)=s}$. From the above
recurrence, we then have,
\begin{equation}
  \label{eq:wint-rec}
  -\ln \ratio{G,u}{i,j}(w) =(n_i - n_j)\ln w \eqbreakno +\sum_{k=1}^d \Big(f_\gamma\inp{\as{G_k,v_k}{i}(w)}
  - f_\gamma\inp{\as{G_k,v_k}{j}(w)}\Big).
\end{equation}
Note that the same recurrence also applies when $w$ is replaced by $\tw$ (and
hence $\gamma$ by $\gammat$):
\begin{equation*}
  -\ln \ratio{G,u}{i,j}(\wt) = (n_i - n_j)\ln \tw \eqbreak + \sum_{k=1}^d \Big(f_\gammat\inp{\as{G_k,v_k}{i}(\tw)} 
  - f_\gammat\inp{\as{G_k,v_k}{j}(\tw)}\Big).
\end{equation*}
(Recall that since $\Re w, \tw > 0$, $\ln w$ and $\ln \tw$ are well
defined).

Using item 2 of \MakeUppercase Observation\nobreakspace \ref {obv:wg-prop}, $|n_i - n_j| \leq \Delta$, and the fact that
$\Delta\nu_w \leq \eps_w$, we have
\[\abs{n_i- n_j}\cdot \abs{\ln w - \ln \tw} \leq \eps_w.\]  Applying the triangle
inequality to the real part of the difference of the two recurrences, we
therefore get
\begin{multline}
  \frac{1}{d} \abs{\Re \ln \ratio{G,u}{i,j}(w)-\ln \ratio{G,u}{i,j}(\tw)}\\
\le \eps_w + \max_{1 \leq k \leq d} \inbr{\abs{\inp{\Re f_\gamma\inp{\as{G_k,v_k}{i}(w)} - f_\gammat\inp{\as{G_k,v_k}{i}(\tw)}}\right.\right.\eqbreakno\left.\left.
        - \inp{\Re f_\gamma\inp{\as{G_k,v_k}{j}(w)} - f_\gammat\inp{\as{G_k,v_k}{j}(\tw)}}}}.
\label{eq:wint-5}
\end{multline}

In what follows, we let $v_k$ be the vertex that maximizes the above expression,
and $d_k$ be the number of unpinned neighbors of $v_k$ in $G_k$.  Before
proceeding with the analysis, we recall the observation above that the
graphs $G_k$ satisfy the admissible list condition $\mathcal{L}$.  Further, we
note that $v_k$ is (i) marked in $G_k$ (this follows from the
  fact that $\mathcal{L}$ is hereditary); and (ii) \nice{} in $G_k$ (this
last fact follows from Lemma\nobreakspace \ref {obv:good-to-nice} and the fact
that $G$ satisfies the admissible list condition $\mathcal{L}$).
Thus, the preconditions of
\MakeUppercase Consequence\nobreakspace \ref {lem:wint-real-part-ind} applies to the vertex $v_k$ in graph $G_k$.  We
now proceed with the analysis.

We begin by noting that
\begin{multline*}
  \abs{\inp{\Re f_\gamma\inp{\as{G_k,v_k}{i}(w)} -
  f_\gammat\inp{\as{G_k,v_k}{i}(\tw)}} \right.\eqbreakno\left.
    -\inp{\Re f_\gamma\inp{\as{G_k,v_k}{j}(w)} - f_\gammat\inp{\as{G_k,v_k}{j}(\tw)}}}\\
  \leq \max_{i', j' \in L(v_k)} \abs{\inp{\Re f_\gamma\inp{\as{G_k,v_k}{i'}(w)}
  - f_\gammat\inp{\as{G_k,v_k}{i'}(\tw)}} \right.\eqbreakno\left.
      - \inp{\Re f_\gamma\inp{\as{G_k,v_k}{j'}(w)} -
      f_\gammat\inp{\as{G_k,v_k}{j'}(\tw)}}}.
\end{multline*}
On the other hand, for any color $s \in L(v_k)$, \MakeUppercase Consequence\nobreakspace \ref {lem:wint-real-part-ind}
of the induction hypothesis instantiated on $G_k$ and applied to $v_k$ and $s$
shows that there exists a $C_s = C_{s, v_k, G_k} \in [0, 1/(d_k + \eta)]$ such
that
\begin{align*}
  \abs{\Re f_\gamma\inp{\as{G_k,v_k}{s}(w)} - f_\gammat\inp{\as{G_k,v_k}{s}(\tw)}
  - C_s(\Re \as{G_k,v_k}{s}(w) - \as{G_k,v_k}{s}(\tw))} \leq \eps_I + \eps_w.
\end{align*}
Substituting this in the previous display shows that
\begin{align}
	&\abs{\inp{\Re f_\gamma\inp{\as{G_k,v_k}{i}(w)} - f_\gammat\inp{\as{G_k,v_k}{i}(\tw)}}\right.\eqbreak 
    \maybealign
   \left.  -\inp{\Re f_\gamma\inp{\as{G_k,v_k}{j}(w)} - f_\gammat\inp{\as{G_k,v_k}{j}(\tw)}}}\nonumber\\
  &\leq \max_{i', j' \in L(v_k)} \abs{C_{i'}(\Re \as{G_k,v_k}{i'}(w) - \as{G_k,v_k}{i'}(\tw)) \right.\eqbreak
    \maybealign
    -\left. C_{j'}(\Re \as{G_k,v_k}{j'}(w) - \as{G_k,v_k}{j'}(\tw))} + 2\eps_I + 2\eps_w\nonumber\\
  &= 2\eps_I + 2\eps_w + \max_{i', j' \in L(v_k)} \abs{C_{i'}\Re \xi_{i'}
    - C_{j'}\Re \xi_{j'}},\nonumber\\
  &= 2\eps_I + 2\eps_w + C_{s}\Re \xi_{s}
    - C_{t}\Re \xi_{t}\label{eq:wint-14},
\end{align}
where $\xi_{l} \defeq \as{G_k,v_k}{l}(w) - \as{G_k,v_k}{l}(\tw)$ for
$l \in \good{G_k, v_k}$, and $s$ and $t$ are given by
\[
  s \defeq \argmax_{i' \in L(v_k)} C_{i'}\Re \xi_{i'}\quad\text{ and }\quad t
  \defeq \argmin_{i' \in L(v_k)} C_{i'}\Re \xi_{i'}.
\]
We now have the following two cases:

\medskip
\noindent{\textbf{Case 1: $(\Re \xi_{s})\cdot(\Re \xi_{t}) \leq 0$.}} Recall that $C_s, C_t$
  are non-negative and lie in $[0, 1/(d_k + \eta)]$.  Thus, in this case, we
  must have $\Re \xi_{s} \geq 0$ and $\Re \xi_{t} \leq 0$, so that
  \begin{align}
    C_{s}\Re \xi_{s} - C_{t}\Re \xi_{t} \maybealign=  C_{s}\Re \xi_{s} + C_{t}\abs{\Re \xi_{t}} \eqbreak
    \maybealign\leq \frac{\Re \xi_{s} + \abs{\Re \xi_{t}}}{d_k+\eta} = \frac{\abs{\Re \xi_{s}  - \Re \xi_{t}}}{d_k+\eta}.\label{eq:wint-10}
  \end{align}
  Now, note that
  \begin{align*}
    \maybealign\Re \xi_{s}  - \Re \xi_{t}\eqbreak
    &= \Re \ln \frac{\prob{G_k,w}{c(v_k)=s}}{
      \prob{G_k,\tw}{c(v_k)=s}} - \Re \ln \frac{\prob{G_k,w}{c(v_k)=t}}{
      \prob{G_k,\tw}{c(v_k)=t}}\\
    &= \Re \ln \frac{\prob{G_k,w}{c(v_k)=s}}{\prob{G_k,w}{c(v_k)=t}}
      - \Re \ln \frac{\prob{G_k,\tw}{c(v_k)=s}}{\prob{G_k,\tw}{c(v_k)=t}}\\
    &= \Re \ln \ratio{G_k, v_k}{s,t}(w) - \ln \ratio{G_k, v_k}{s,t}(\tw).
  \end{align*}
  Note that all the logarithms in the above are well defined from
  Consequence\nobreakspace \ref {lem:wint-approx-prob} of the induction hypothesis applied to $G_k$ and
  $v_k$.  Further, from items\nobreakspace \ref {wint-ind-1} and\nobreakspace  \ref {wint-ind-2} of the induction hypothesis,
  the last term is at most $d_k\eps_R + \eps_w$ in absolute value.  Substituting
  this in eq.\nobreakspace \textup {(\ref {eq:wint-10})}, we get
  \begin{equation}
    C_{s}\Re \xi_{s} - C_{t}\Re \xi_{t}  \leq \frac{d_k}{d_k + \eta}\eps_R + \eps_w. \label{eq:wint-11}
  \end{equation}
  This concludes the analysis of Case 1.

  \medskip
\noindent{\textbf{Case 2: $\Re \xi_{i'}$ for $i' \in L(v_k)$ all have the same sign.}}
Suppose first that $\Re \xi_{i'} \geq 0$ for all $i' \in L(v_k)$.  Then, we have
  \begin{equation}
    0 \leq  C_{s}\Re \xi_{s} - C_{t}\Re \xi_{t} \leq \frac{\Re \xi_{s}}{d_k+\eta}
    \leq \frac{d_k\cdot \eps_R}{d_k + \eta} + \eps_I + 4\Delta\eps_w,\label{eq:wint-12}
  \end{equation}
  where the last inequality follows from the second inequality in \MakeUppercase Consequence\nobreakspace \ref {lem:wint-approx-prob} of
  the induction hypothesis applied to $G_k$ at vertex $v_k$ with color $s$,
  which states that $\abs{\Re \xi_{s}} \leq d_k(\eps_R + \eps_I) + 4\Delta\eps_w$.
  Similarly, when $\Re \xi_{i'} \leq 0$ for all $i' \in \good{G_k, v_k}$, we
  have
  \begin{align}
    0 \leq  C_{s}\Re \xi_{s} - C_{t}\Re \xi_{t}
    &=  C_{t}|\Re \xi_{t}| -  C_{s}|\Re \xi_{s}|\nonumber\\&
    \leq \frac{1}{d_k+\eta}\abs{\Re \xi_{t}}\eqbreak
    \maybealign\leq \frac{d_k}{d_k + \eta}\eps_R + \eps_I + 4\Delta\eps_w,\label{eq:wint-13}
  \end{align}
  where the last inequality follows from the second inequality in \MakeUppercase Consequence\nobreakspace \ref {lem:wint-approx-prob} of
  the induction hypothesis applied to $G_k$ at vertex $v_k$ with color $t$,
  which states that $\abs{\Re \xi_{t}} \leq d_k(\eps_R + \eps_I) +
  4\Delta\eps_w$.  This concludes the analysis of Case 2.

  \medskip

Now, substituting eqs.\nobreakspace  \textup {(\ref {eq:wint-11})} to\nobreakspace  \textup {(\ref {eq:wint-13})}  into eq.\nobreakspace \textup {(\ref {eq:wint-14})}, we get
\begin{multline}
  \label{eq:wint-15}
  \abs{\inp{\Re f_\gamma\inp{\as{G_k,v_k}{i}(w)} - f_\gammat\inp{\as{G_k,v_k}{i}(\tw)}}\right.\eqbreakno\left.
    -\inp{\Re f_\gamma\inp{\as{G_k,v_k}{j}(w)} - f_\gammat\inp{\as{G_k,v_k}{j}(\tw)}}}\\
  \leq \frac{d_k}{d_k + \eta}\eps_R + 3\eps_I + 5\Delta\eps_w.
\end{multline}

Substituting eq.\nobreakspace \textup {(\ref {eq:wint-15})} into eq.\nobreakspace \textup {(\ref {eq:wint-5})}, we get
\begin{equation}
  \frac{1}{d} \abs{ \Re \ln \ratio{G,u}{i,j}(w)-\ln \ratio{G,u}{i,j}(\tw) } \eqbreak
  \leq \frac{d_k}{d_k + \eta}\eps_R + 3\eps_I + 7\Delta\eps_w < \eps_R,
\end{equation}
where the last inequality holds since
$\eta\eps_R > (\Delta + 1)(3\eps_I + 7\Delta\eps_w)$ (recalling that
$0 \leq d_k \leq \Delta$ and $\eta \in [0.9, 1)$). This verifies item\nobreakspace \ref {wint-ind-2} of the induction hypothesis.

Finally, to prove item\nobreakspace \ref {wint-ind-3}, we consider the imaginary part of
eq.\nobreakspace \textup {(\ref {eq:wint-rec})}.  We first note that
\begin{displaymath}
  \abs{n_i-n_j}\cdot\abs{\Im \ln w} \leq \Delta\abs{\ln w - \ln\tw} \leq
  \Delta\nu_w \leq \eps_w.
\end{displaymath}
We then have
\begin{equation}
  \frac{1}{d} \abs{ \Im \ln \ratio{G,u}{i,j}(w) }
  \le \eps_w+\eqbreak
   \max_{1 \leq k \leq d}
    \abs{\Im f_\gamma\inp{\as{G_k,v_k}{i}(w)} - \Im f_\gamma\inp{\as{G_k,v_k}{j}(w)}}
    .\label{eq:wint-16}
\end{equation}
Again, let $v_k$ be the vertex that maximizes the above expression, and $d_k$ be the number of unpinned neighbors of $v_k$ in $G_k$.
Applying eq.\nobreakspace \textup {(\ref {eq:wint-7})} of Consequence\nobreakspace \ref {lem:wint-real-part-ind} of the induction
hypothesis to the graph $G_k$ at vertex $v_k$ with colors $i, j \in L(v_k)$
gives
\begin{equation}
  \label{eq:wint-17}
  \abs{\Im f_\gamma\inp{\as{G_k,v_k}{i}(w)} - \Im
    f_\gamma\inp{\as{G_k,v_k}{j}(w)}} \eqbreak
    \leq \frac{d_k}{d_k + \eta}\eps_I + 6\Delta\eps_w.
\end{equation}
Substituting eq.\nobreakspace \textup {(\ref {eq:wint-17})} into eq.\nobreakspace \textup {(\ref {eq:wint-16})} we then have
\begin{displaymath}
  \frac{1}{d} \abs{ \Im \ln \ratio{G,u}{i,j}(w) } \leq \frac{d_k}{d_k + \eta}\eps_I +
  8\Delta\eps_w < \eps_I,
\end{displaymath}
where the last inequality holds since $\eta\eps_I > 8(\Delta + 1)\Delta\eps_w$
(recalling that $0 \leq d_k \leq \Delta$ and
$\eta \in [0.9, 1)$). This proves
item\nobreakspace \ref {wint-ind-3}, and also completes the inductive proof of
\MakeUppercase Lemma\nobreakspace \ref {lem:interval-induction}.\qed

We now use Lemma\nobreakspace \ref {lem:interval-induction} to prove
Theorem\nobreakspace \ref {thm:main-interval-restated}.
\begin{proof}[Proof of Theorem\nobreakspace \ref {thm:main-interval-restated}]
  Let $G$ be any graph of maximum degree $\Delta$ satisfying
  the admissible list condition $\mathcal{L}$.  If $G$ has no unpinned vertices, then $Z_G(w) = 1$
  and there is nothing to prove.  Otherwise, let $u$ be an unpinned vertex
  that is marked in~$G$.
  By \MakeUppercase Consequence\nobreakspace \ref {lem:wint-sum-lowerbounds} of the induction hypothesis (which we proved in
  Lemma\nobreakspace \ref {lem:interval-induction}), we then have $Z_w(G) \neq 0$ for $w$ as in the
  statement of the theorem.
\end{proof}

The proof of \MakeUppercase Theorem\nobreakspace \ref {thm:zeros} is now immediate.

\begin{proof}[Proof of \MakeUppercase Theorem\nobreakspace \ref {thm:zeros}]
  Let the quantity $\nu_w = \nu_w(\Delta)$ be as in the statements of
  \MakeUppercase Theorems\nobreakspace \ref {thm:main-origin-restated} and\nobreakspace  \ref {thm:main-interval-restated}.  Fix the maximum
  degree $\Delta$, and suppose that $w$ satisfies
  \begin{equation}
    -\nu_w^2/8 \leq \Re w \leq 1 + \nu_w^2/8 \text{ and } \abs{\Im w} \leq \nu_w^2/8.\label{eq:33}
  \end{equation}
  Now, if $G$ satisfies the hypotheses of Theorem~\ref{thm:zeros-intro}
    (respectively, Theorem~\ref{thm:zeros-intro-176}), we mark all its vertices so that
    the resulting instance satisfies Condition~\ref{asm:2G} (respectively,
    Condition~\ref{asm:176}); whereas if $G$ is a tree satisfying the hypotheses
    of Proposition~\ref{thm:zeros-trees}, we root $G$ at an arbitrary vertex and
    mark the root, so that the resulting instance satisfies
    Condition~\ref{asm:tree}.

  By Lemma~\ref{obv:good-to-nice}, the list coloring instance for $G$ so
  generated then satisfies an admissible list condition.  When $w$
  satisfying eq.\nobreakspace \textup {(\ref {eq:33})} is such that
  $\Re w \leq \nu_w/2$, we have $\abs{w} \leq \nu_w$, so that $Z_G(w) \neq 0$ by
  \MakeUppercase Theorem\nobreakspace \ref {thm:main-origin-restated}, while
  when such a $w$ satisfies $\Re w \geq \nu_w/2$, we have $Z_G(w) \neq 0$ from
  \MakeUppercase Theorem\nobreakspace \ref {thm:main-interval-restated}.  It
  therefore follows that $Z_G(w) \neq 0$ for all $w$ satisfying eq.\nobreakspace
  \textup {(\ref {eq:33})}, and thus the quantity $\tau_\Delta$ in the statement
  of Theorem\nobreakspace \ref {thm:zeros} can be taken to be
  $\nu_w^2/8$. \end{proof}

We conclude with a brief discussion of the dependence of $\tau_\Delta$ on
$\Delta$.  We saw above that $\tau_\Delta$ can be taken to be
$\nu_w(\Delta)^2/8$, so it is sufficient to consider the dependence of
$\nu_w = \nu_w(\Delta)$ on $\Delta$.  Let $c = 10^{-6}$.  As stated in the
discussion following eq.\nobreakspace \textup {(\ref {eq:def-eps})}, $\nu_w$ can be chosen to be
$0.2c/(2^{\Delta}\Delta^{7})$ for the case of general list colorings, or
$c/(300\Delta^{8})$ with the assumption of uniformly large list sizes (which, we
recall from \protect \MakeUppercase {R}emark\nobreakspace \ref {rem:nicer}, is satisfied in the case of uniform
$q$-colorings).  We have not tried to optimize these bounds, and it is
conceivable that a more careful accounting of constants in our proofs can
improve the value of the constant $c$ by a few orders of magnitude.

{\renewcommand{\addcontentsline}[3]{}
\section*{Acknowledgments}
}

We thank Guus Regts and anonymous reviewers for their various helpful comments.

JL was a PhD student at UC Berkeley when this work was carried out.  JL and AS
were supported by US NSF grant CCF-1815328. PS was supported by a Ramanujan
Fellowship of SERB, Indian Department of Science and Technology, and by the
Department of Atomic Energy, Government of India, under project
nos.~12-R\&D-TFR-5.01-0500 and RTI4001. Part of this work was performed while the
authors were at the Simons Institute for the Theory of Computing.

\vspace{\baselineskip}

\end{document}